%% file: main.tex
\documentclass[11pt]{article}
\input{settings/packages}
\input{settings/environments}
\input{settings/macros}

\title{Nearly Tight Regret Bounds for Profit Maximization\\in Bilateral Trade\thanks{S.D.G. and F.F.~were supported by the PNRR MUR project
IR0000013-SoBigData.it and the FAIR (Future Artificial Intelligence Research) project
PE0000013, funded by the NextGenerationEU program
within the PNRR-PE-AI scheme (M4C2, investment 1.3,
line on Artificial Intelligence). S.D.G. was also supported by the Institute for Complex Systems (Italian National Research Council). C.S.~was supported by a Google Research Award and by the Independent Research Fund Denmark (DFF) under a Sapere Aude Research Leader grant No 1051-00106B.}}

\author{Simone Di Gregorio\thanks{Dept. of Computer, Control and Management Engineering, Sapienza University of Rome,Rome, Italy. Email: \texttt{simone.digregorio@uniroma1.it}}
\and Paul D\"utting\thanks{Google Research, Z\"urich, Switzerland. Email: \texttt{duetting@google.com}}
\and Federico Fusco\thanks{Dept. of Computer, Control and Management Engineering, Sapienza University of Rome,Rome, Italy. Email: \texttt{fuscof@diag.uniroma1.it}}
\and Chris Schwiegelshohn\thanks{Dept. of Computer Science, Aarhus University,
Aarhus, Denmark. Email: \texttt{schwiegelshohn@cs.au.dk}}}
\date{}
\begin{document}

\maketitle

\input{sections/00-abstract}

\pagenumbering{arabic}

\input{sections/10-introduction}

\input{sections/20-model}

\input{sections/30-proof_sketch}
\input{sections/40-stochastic_upper_bound}
\input{sections/50-lower_bounds}

\input{sections/60-joint_ads}

\bibliographystyle{plainnat}
\bibliography{references}
\newpage

\appendix

\input{sections/100-appendix}

\end{document}

%% file: settings/packages.tex
\usepackage[letterpaper,margin=1in]{geometry}

\usepackage{graphicx} 
\usepackage{booktabs} 
\usepackage{algorithm}
\usepackage[noend]{algpseudocode}
\usepackage{natbib} 
\usepackage{amsthm,amsmath,amssymb,xcolor}
\usepackage{mathrsfs}
\usepackage{thmtools} 
\usepackage{thm-restate}
\usepackage{microtype,nicefrac,bbm}
\usepackage{caption}
\usepackage{subcaption}
\usepackage{mathtools,xspace}
\usepackage{bold-extra}
\usepackage[symbol]{footmisc}

\usepackage{perpage} 
\MakePerPage{footnote} 

\usepackage[hidelinks]{hyperref}
\usepackage{cleveref}

%% file: settings/environments.tex
\newtheorem{lemma}{Lemma}
\newtheorem{claim}{Claim}

\newtheorem{observation}{Observation}
\newtheorem{corollary}{Corollary}
\newtheorem{proposition}{Proposition}
\newtheorem{remark}{Remark}
\newtheorem{theorem}{Theorem}
\newtheorem{definition}{Definition}

\newtheorem{example}{Example}

%% file: settings/macros.tex
\newcommand{\xhdr}[1]{\vspace{2mm} \noindent{\bf #1}}

\newcommand{\prof}{\textnormal{\textsf{profit}}}

\newcommand{\util}[3]{U_{#3}\left(#1, #2\right)}
\newcommand{\trade}[2]{x\left(#1, #2\right)}
\newcommand{\rev}{\textsf{rev}}
\newcommand{\gft}{\textsf{GFT}}
\newcommand{\sw}{\textsf{SW}}
\newcommand{\app}[2]{\phi_{#1}\left(#2\right)}

\newcommand{\vs}{v_s}
\newcommand{\vb}{v_b}
\newcommand{\bs}{b_s}
\newcommand{\bb}{b_b}
\newcommand{\vi}{v_i}
\newcommand{\Mplus}{\M^+}
\newcommand{\Meps}{\M_{\varepsilon}}

\newcommand{\Mopt}{M_{\textrm{opt}}}
\newcommand{\Memp}{\hat M_{\textrm{opt}}}
\newcommand{\simplify}{\textsc{Simplify-the-Best-Mechanism}\xspace}
\newcommand{\augment}{\textsc{Augment-the-Best-Mechanism}\xspace}

\newcommand{\E}[1]{\mathbb{E}\left[#1\right]}
\newcommand{\Var}[1]{\mathbb{V}\left[#1\right]}
\newcommand{\Esub}[2]{\mathbb{E}_{#1}\left[#2\right]}

\newcommand{\G}{\mathcal G}

\newcommand{\X}{\mathcal{X}}

\newcommand{\eqd}{\stackrel{\text{(d)}}{=}}

\newcommand{\Ei}[1]{\mathbb{E}^i\left[#1\right]}
\newcommand{\Ez}[1]{\mathbb{E}^0\left[#1\right]}
\newcommand{\Eo}[1]{\mathbb{E}^1\left[#1\right]}
\newcommand{\Et}[1]{\mathbb{E}^2\left[#1\right]}

\renewcommand{\P}[1]{\mathbb{P}\left(#1\right)}
\newcommand{\Psub}[2]{\mathbb{P}_{#1}\left(#2\right)}

\newcommand{\Pz}[1]{\mathbb{P}^0\left(#1\right)}
\newcommand{\Pt}[1]{\mathbb{P}^2\left(#1\right)}

\newcommand{\Pbz}{\mathbb{P}^0}
\newcommand{\Pbo}{\mathbb{P}^1}
\newcommand{\Pbt}{\mathbb{P}^2}

\renewcommand{\Pi}[1]{\mathbb{P}^i\left(#1\right)}
\newcommand{\Pbi}{\mathbb{P}^i}

\newcommand{\KL}{\mathrm{KL}}
\newcommand{\kl}[2]{\mathcal{D}_{\mathrm{KL}}\left(#1,#2\right)}
\newcommand{\tv}{\mathrm{TV}}

\newcommand{\eps}{\varepsilon}

\newcommand{\A}{\mathcal{A}}
\newcommand{\C}{\mathcal{C}}
\newcommand{\D}{\mathcal{D}}
\newcommand{\cE}{\mathcal{E}}
\newcommand{\F}{\mathcal{F}}

\newcommand{\M}{\mathcal{M}}
\newcommand{\cN}{\mathcal{N}}

\newcommand{\R}{\mathbb{R}}

\newcommand{\ind}[1]{\mathbbm{1}_{\{{#1\}}}}

%% file: sections/00-abstract.tex
\begin{abstract}
    Bilateral trade models the task of intermediating between two strategic agents, a seller and a buyer, willing to trade a good for which they hold private valuations. We study this problem from the perspective of a broker, in a regret minimization framework. At each time step, a new seller and buyer arrive, and the broker has to propose a mechanism that is incentive-compatible and individually rational, with the goal of maximizing profit. 

    We propose a learning algorithm that guarantees a {nearly} tight {$\Tilde{O}(\sqrt{T})$} regret in the stochastic setting when seller and buyer valuations are drawn i.i.d. from a fixed and possibly correlated unknown distribution. We further show that it is impossible to achieve sublinear regret in the non-stationary scenario where valuations are generated upfront by an adversary. Our ambitious benchmark for these results is the best incentive-compatible and individually rational mechanism.
    This separates us from previous works on efficiency maximization in bilateral trade, where the benchmark is a single number: the best fixed price in hindsight. 

    A particular challenge we face is that uniform convergence for all mechanisms' profits is impossible. We overcome this difficulty via a careful chaining analysis that proves convergence for a provably near-optimal mechanism at (essentially) optimal rate. 
    We further showcase the broader applicability of our techniques by providing nearly optimal results for the joint ads problem.
\end{abstract}

%% file: sections/10-introduction.tex
\section[Introduction]{Introduction 
{}}
\label{sec:introduction}
    
    Bilateral trade is a {fundamental} economic model where two agents---a seller and a buyer---are interested in trading a good {\citep[e.g.,][]{MyersonS83}.} Both agents are characterized by a private valuation for the item, and their goal is to maximize their own utility. 
    {Previous} work {in computer science} {\citep[e.g.,][]{ColiniBaldeschiKLT16,BrustleCWZ17,BabaioffGG20,BlumrosenD21,DuttingFLLR21,DengMSW21,kang22fixed,LiuR023,CaiW23,BernasconiCCF24}} has 
    mainly focused on maximizing economic efficiency, measured in terms of welfare or gains-from-trade. We consider this problem from the perspective of a broker whose goal is to maximize its profit {(payment from the buyer minus payment to the seller)}
    while ensuring that (i) agents behave according to their true preferences (\emph{dominant-strategy incentive compatibility}) and (ii) the utility for participating in the mechanism of each agent is non-negative (\emph{individual rationality}). 

    Profit maximization is a canonical objective in economics (already studied in the seminal paper by \citet{MyersonS83}), and it is a well-motivated objective in the applications (e.g., all e-commerce platforms that connect sellers to buyers). 
    It is also gaining attention from a theoretical computer science perspective (most notably in \citet{HajiaghayiHPS25}). The predominant approach considered in prior work is \emph{Bayesian}, as it assumes that both agents' valuations are drawn from known and typically independent distributions.  In many practical applications, however, the availability of such data is a non-trivial assumption, and valuations may not be independent. 
    
    In this work, we study profit maximization in repeated bilateral trade through the lens of online learning, with the goal of designing an algorithm that minimizes the {regret} with respect to the best mechanism in hindsight. Note {that unlike} prior research {on repeated bilateral trade}, which has focused on gains-from-trade \citep[e.g.,][]{Cesa-BianchiCCF24,Cesa-BianchiCCF24jmlr,AzarFF24,BernasconiCCF24} and has thus considered \emph{fixed-price mechanisms}, here we have the ambitious goal of converging to the profit of {a provably complex optimal mechanism}. In fact, while fixed-price mechanisms to some extent \emph{characterize} efficiency-maximizing mechanisms, they are far from optimal in our problem {(we refer to \Cref{app:efficiency} for further details)}. The nature of our benchmark {also} differentiates us {from} the vast majority of {the} online learning {literature} {on economic problems,} where the learning objective is intrinsically parametric, e.g., the best fixed price(s), \citep[e.g.,][]{BlumKRW03,KleinbergL03,GatmiryKSW24}, the best reserve price(s) \citep[e.g.,][]{Cesa-BianchiGM15,RoughgardenW19}, the best linear contract \citep{ZhuBYWYJ22,DuettinGSW23} or the best fixed bid \citep[e.g.,][]{FengPS18,CesaBianchiCCFS24}. 

    As it is traditional in online learning, we study two data generation models: the adversarial input model, where the agent's valuations are decided upfront by an oblivious adversary, and the stochastic i.i.d.~case, when the agent's valuations are drawn i.i.d. from a possibly correlated but fixed unknown distribution. Note that, in the stochastic i.i.d.~case, the benchmark coincides with the best Bayesian mechanism, but unlike in the typical ``one-shot'' Bayesian mechanism, the input distribution is not known in advance, can be correlated, and has to be learned on the fly. 

    \subsection[Our Results]{Our Results}
    \label{sec:results}

    \xhdr{Stochastic Setting.}
    {As our main result, we show that the stochastic i.i.d.~setting, with possibly correlated buyer and seller valuations, admits a (polynomial-time) learning algorithm with a nearly tight regret bound of $\tilde{O}(\sqrt{T})$}
    \footnote{We use $\tilde{O}(\cdot)$ to suppress poly-logarithmic terms.}.
    A particular challenge is that uniform convergence for all mechanisms' profits is impossible. We overcome this via a careful chaining analysis that proves convergence for a provably near-optimal mechanism at (essentially) optimal rate. 
    We provide a more detailed summary 
    of our analysis in \Cref{sec:challenges}.

        \medskip 
        \noindent
        \textbf{Key Result 1} (\Cref{thm:regret-bound}). 
        For the stochastic i.i.d.~setting, there is a (polynomial-time) learning algorithm that achieves a regret of $\tilde{O}(\sqrt{T})$ with respect to the best mechanism in hindsight. 
        \medskip

    We complement this result with a lower bound of $\Omega(\sqrt{T})$  on the regret of any learning algorithm (\Cref{thm:iid_lower_bound}), showing that our algorithm is optimal up to poly-logarithmic terms. This is accomplished by embedding the standard hard instance for prediction with experts \citep[e.g.,][]{nicolo06}.  
    Our analysis also (essentially) settles the offline sample complexity of learning an expected $ \eps $-optimal mechanism, in the spirit of the sample complexity literature \cite[e.g.,][]{ColeR14}. More precisely, the analysis of our $\tilde{O}(\sqrt{T})$-regret algorithm and the corresponding lower bound implies that to recover a mechanism whose expected profit is $\eps$-optimal, having access to approximately $\nicefrac{1}{\eps^2}$ is sufficient as well as necessary.
    
    \xhdr{Adversarial Setting.}
    Moving to the more challenging adversarial setting, we construct a family of hard instances that defuses any learning algorithm, thus implying that no-regret is unattainable in such regime. We actually prove a stronger result: in the adversarial setting, no learning algorithm can achieve sublinear regret with respect to a $\nicefrac {2}{3} + \eps$ fraction of the optimal profit. 

    \medskip
    \noindent
    \textbf{Key Result 2} (\Cref{thm:adv_lower_bound}). 
    For the adversarial setting, no learning algorithm can achieve sublinear regret with respect to a $\nicefrac{2}{3} + \eps$ fraction of the optimal profit. 
    \medskip

                Our results confirm an intriguing separation between the adversarial and the stochastic model, that has been recently uncovered in many other economically motivated scenarios {\cite[e.g.][]{Cesa-BianchiCCF24,CesaBianchiCCFS24,AggarwalBDF24}}. This difference is indeed not observable in many classical online learning problems, such as prediction with experts, bandits, feedback graphs, or partial monitoring \citep{BartokFPRS14,AlonCDK15}.          A natural open question is whether there are meaningful \emph{intermediate} data generation models beyond the stochastic setting where no-regret is achievable, as in, e.g., smoothed analysis, random order, adversarial corruption \citep[see, e.g.,][]{Roughgarden20}.  {A promising direction here is the notion of a $\sigma$-smooth adversary \citep{HaghtalabRS24}, that has been successfully adopted to move beyond adversarial lower bounds for e.g., online classification \citep{BlockDGR22}, bilateral trade \citep{Cesa-BianchiCCF24jmlr} and general auctions \citep{DurvasulaHZ23}.} We leave this question as an {exciting} 
                open problem. 

        \xhdr{Joint Ads Problem.} We finally exploit a {novel and intriguing} 
        correspondence between bilateral trade and non-excludable mechanism design to (essentially) close an open problem of \citet{AggarwalBDF24}, by 
        {reducing an $O(T^{3/4})$ to an}
        $\tilde{O}(\sqrt{T})$ regret bound 
        for the joint ads problem. 
        In this problem, two advertisers, e.g., a brand and a merchant, both derive a private possibly correlated value from an ad, and the mechanism's decision is to either show the ad or not, and define payments from either party to the mechanism, with the goal of maximizing revenue. 

    \subsection{Challenges}

        Consider a generic mechanism $M$ that belongs to the set $\M$ of all dominant-strategy incentive compatible (DSIC) and individually rational (IR) mechanisms. $M$ is specified by two functions---an allocation and a payment rule---mapping bids to allocation and payments. Standard mechanism design arguments (see \Cref{sec:model}) enable a simple geometric interpretation of the structure of $M$. Its allocation rule corresponds to a subdivision of the unit square $[0,1]^2$ into two regions: one in which the trade happens and one where this is not the case. This partition has to be ``north-west'' monotone so that increasing the buyer's bid or decreasing the seller's does not decrease the likelihood of observing a trade. The payments are then given by the ``south-east'' projection of the valuation pair on the boundary of the allocation region {(see \Cref{fig:myerson_prices}).} 
        
        This characterization simplifies our task from learning two functions to a single curve/region. However, this is far from trivializing the problem: (i) the action space is (still) non-parametric, and (ii) the $\prof$ function mapping mechanisms to their induced profit may be highly non-regular. These two features defuse the standard learning approaches, as we briefly detail in the following. 
        \begin{figure}[t!]
            \centering
            \begin{subfigure}{0.3\textwidth}
                \centering
                \includegraphics[width=\linewidth]{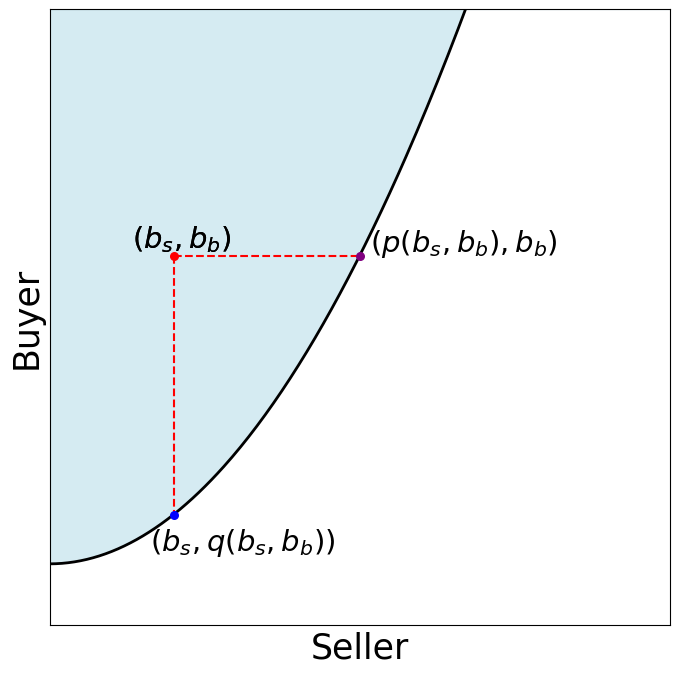}
                \caption{Myerson's Payments.}
                \label{fig:myerson_prices}
            \end{subfigure}
            \hfill
            \begin{subfigure}{0.3\textwidth}
                \centering
                \includegraphics[width=\linewidth]{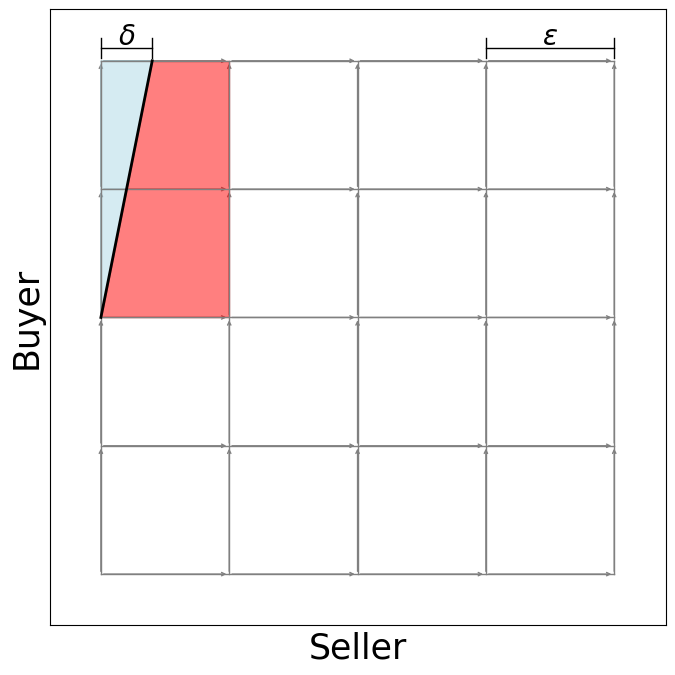}
                \caption{Visualization of \Cref{ex:fixed_grid}.}
                \label{fig:fixed_grid}
            \end{subfigure}
            \hfill
            \begin{subfigure}{0.3\textwidth}
                \centering
                \includegraphics[width=\linewidth]{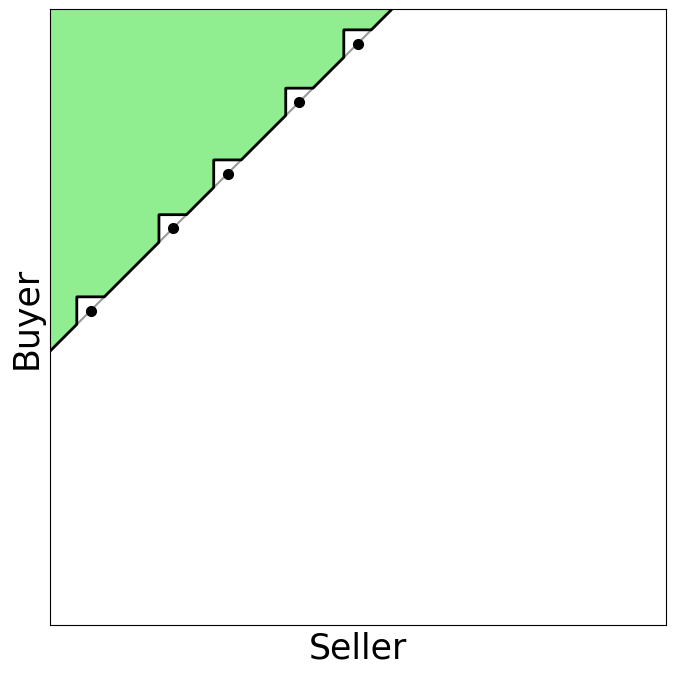}
                \caption{Hard Instance of \Cref{ex:hard_distribution}.}
                \label{fig:hard_distribution}
            \end{subfigure}
            \caption{Visualization of Myerson's Payment, and \Cref{ex:fixed_grid,ex:hard_distribution}.}
        \end{figure}
        
    \xhdr{Insufficiency of Fixed Discretization.}
    When faced with a large action space, the natural  approach in online learning is to discretize it (e.g., \citet{KleinbergSU19}). This method splits the problem in two: finding offline a finite subset of actions that well approximates the optimal action \emph{uniformly over all possible inputs} and then running online a \emph{discrete} learning algorithm on the discretized action space. The most natural discretization of our mechanism space is to cast a uniform $\eps$-grid on the $[0,1]^2$ square and consider all the $\binom{\nicefrac 2\eps}{\nicefrac 1\eps}$ mechanisms $\M^{\perp}_{\eps}$ with allocation regions that can be described as the union of the cells of the grid.\footnote{The family $\M^{\perp}_{\eps}$ is a uniform $\eps$-coverage of the mechanism space, for a suitable definition of distance between mechanisms. Similar ideas have been used successfully with some additional assumptions on the data generation model, e.g., smoothness of the underlying distributions \citep{DurvasulaHZ23,AggarwalBDF24}.}
    Unfortunately, as we show in the following example, this rich class of mechanisms does not provide a uniform approximation for our problem. This is due to the non-continuity (and thus non-Lipschitzness) of our objective function.  
    \begin{example}[A Fixed Discretization Does Not Work]
    \label{ex:fixed_grid}
        Let $\delta>0$ be an arbitrarily small parameter, and consider the uniform distribution $\D$ on the points  $(0,\nicefrac{1}{2})$ and $(\delta,1)$. The optimal mechanism allocates in the triangle of vertices $(0,\nicefrac{1}{2})$,$(0,1)$, and $(\delta,1)$, so that any valuations in the support of $\D$ result in a trade, with expected revenue approximately $\nicefrac 34$ (we consider $\delta$ negligible). Consider now any fixed discretization parameter $\eps > \delta$\footnote{The parameter $\delta$ can be arbitrarily small, while $\eps$ is fixed in advance by the learner. Note, for the same argument, one can adjust this construction to defuse \emph{any} fixed grid.}, the support of $\D$ is contained in the first column of the grid, so the profit extracted by any mechanism $M \in \M^{\perp}_{\eps}$ only depends (up to $O(\delta)$ payment to the seller) on the horizontal segment that crosses such column. The best possible profit corresponds to crossing the first column approximately at height $\nicefrac 12$, for an expected utility of $\nicefrac 12$, which is a constant factor away from the optimal one (see also \Cref{fig:fixed_grid}).
    \end{example}

    Fixed discretization of the action space is one of the main tools for positive results in the adversarial case when faced with a continuous action space \citep[e.g.,][]{KleinbergL03,KleinbergSU19}. This fundamental gap between the best DSIC and IR mechanism and the best-discretized mechanism is, in fact, at the heart of our $\Omega(T)$ adversarial lower bound (\Cref{thm:adv_lower_bound}), where we use a ``needle in a haystack'' construction to hide the optimal mechanism to the learner. 
    
    \xhdr{Failure of Local Search.} A second, natural, approach consists in implementing some sort of local search in the mechanism space, as in e.g., stochastic gradient ascent. Regardless of the metric or notion of concavity we could impose on the non-parametric mechanism space $\M$, the crucial ingredient we would need for this approach to succeed is the uniqueness or ``vicinity'' of optimal mechanisms. Unfortunately, this is not the case, as detailed in the following example.
    
    \begin{example}[Non-Uniqueness of Optimal Mechanisms]
        Consider the uniform distribution over valuations $(0,1),$ and $(\nicefrac{1}{4}, \nicefrac{3}{4})$. Under this distribution, there are two optimal mechanisms: the mechanism $M_1$ that only allocates in $\{(0,1)\}$, and $M_2$ that allocates in the $[0,\nicefrac 14]\times [\nicefrac 34,1]$ rectangle. These two mechanisms achieve expected utility $\nicefrac 12$ and are ``well-separated'' (for instance, as soon as we restrict the allocation region of $M_2$, we lose the trade opportunity given by $(\nicefrac 14,\nicefrac 34)$, and the utility drops by a constant to $\nicefrac 14$).
    \end{example}

    The possibility of having two well-separated maxima allows us, in \Cref{thm:iid_lower_bound}, to embed the standard hard instance for prediction with two experts \citep[e.g.,][]{nicolo06} in our setting, thus providing a $\Omega(\sqrt T)$ lower bound on the minimax regret rate for the stochastic setting.

    \xhdr{Impossibility of Uniform Learning.} The third standard strategy leverages the uniform learnability of the objective functions. Suppose for the sake of the argument that we were able to prove a statement of the following type: for any precision $\eps>0$, it is possible to reconstruct the expected profit of \emph{all} mechanisms at the same time (up to an $\eps$ precision error) using $n(\eps)$ i.i.d. samples from the underlying distribution $\D$. This would imply that a follow-the-leader strategy incurs an overall regret of order $\sum_t \eps_t$, where $\eps_t$ is the precision achievable with $t$ samples. The road-block to pursuing this strategy in our case is that the family of profit function is complex from the statistical point of view, as it exhibits unbounded pseudo-dimension and constant Rademacher complexity (even the range space of monotone allocation region has unbounded VC dimension, see \Cref{app:complexity} for further details). As it turns out, these negative clues are not a coincidence: uniform learnability \emph{is not possible} in our problem, as {we demonstrate in the following example (See also \Cref{thm:uniform-learning})}.

    \begin{example}[A Hard Distribution]
    \label{ex:hard_distribution} Consider the uniform distribution $\D$ on the segment joining $(0,\nicefrac{1}{2})$ and $(\nicefrac 12, 1)$. For any number $n$, and any realization of $n$ i.i.d. samples $S$ from $\D$, consider mechanism $M_S$ that allocates in the $(0,1)$, $(0,\nicefrac{1}{2})$ and $(\nicefrac 12, 1)$ triangle, save for puncturing out the points $S$, as in \Cref{fig:hard_distribution}. We have the following two properties: (i) the expected profit of $M_S$ under $\D$ is at least $\nicefrac 14$, but (ii) the empirical profit of $M_S$ on the samples is zero. This implies that uniform learning is not possible against $\D$.
    \end{example}

    \xhdr{Our Techniques.} We overcome these barriers by controlling a stochastic process that simulates a branch-and-bound strategy among sets of increasingly powerful mechanisms. The branch-and-bound strategy initially considers fixed-price mechanisms (rectangular allocation regions), then iteratively branches off multiple refined mechanisms by modifying the boundaries of simpler ones. As we refine these mechanisms, they become richer and exhibit potentially larger profits. Thus, we implicitly explore the entire search space for mechanisms, avoiding the barrier inherent in local improvements. The issue with considering richer mechanisms is that we no longer can show that their empirical profit is concentrated around the real expected profit, due to the impossibility of uniform learning. Our important insight is that the impossibility of uniform learning in some sense mainly affects the mechanisms with poor empirical profit, thus we can abstain from search paths that explore mechanisms with small profit and focus on refining profitable mechanisms. 
    
     We adapt our characterization of profitable mechanisms as we explore the search space. This addresses the lack of good fixed discretization strategies, as we are learning a (data and distribution dependent) discretization on the fly. For the profitable mechanisms, this branch-and-bound strategy mimics a chaining analysis, a powerful tool for bounding suprema of stochastic processes (see \cite{Talagrand14}). 
    Chaining has been used in a wide variety of applications including generalization bounds \citep{AudibertB07,Cohen-AddadLS25,LiLS00}, dimensionality reduction \citep{BourgainDN15,NarayananN19}, error correcting codes \citep{RudraW14},  graph sparsification \citep{JambulapatiLS23,Lee23}, data structures \citep{BravermanCINWW17}, and compression \citep{Cohen-AddadLSS22,Cohen-AddadD0SS25}; see \cite{Nelson16} for more applications and references.
    The chaining analysis gives us very tight control over the profit of good mechanisms, which allows us to learn a sequence of mechanisms with nearly optimal regret. We give a more technical overview of the key ingredients in our proof in \Cref{sec:challenges}.

\subsection{Further Related Work}
\label{sec:related}

    \xhdr{Efficiency and Profit in Bilateral Trade.} The celebrated result of \citet{MyersonS83} on bilateral trade (strengthening an earlier result of \citet{Vickrey61}) states that even for independent and regular valuation distributions, it is impossible to simultaneously achieve (a) incentive compatibility---even the weaker notion of Bayesian incentive compatibility (BIC)--- (b) perfect efficiency,
    ({i.e., trade whenever the buyer has a higher valuation than the seller}), (c) individual rationality, and (d) budget balance (requiring that the broker does not run a deficit). On the positive side, they also characterize the optimal BIC mechanisms (which turn out to be DSIC) for social welfare, gains from trade, and profit under some assumptions on the valuations' distributions.

    Motivated by this impossibility result, a fruitful line of work in computer science has aimed at mechanisms that are incentive compatible and budget balanced, while achieving a provable worst-case approximation guarantee --- with respect to social welfare \citep{LiuR023,CaiW23,BlumrosenD21,DuttingFLLR21,kang22fixed} or gains-from-trade \citep{Mcafee08,Fei22,CaiGMZ21,DengMSW21,BrustleCWZ17}. 
    While these results hold in the Bayesian setting, where the agents' valuations are drawn independently from some known distribution, some of these approximation mechanisms actually make use of only a few (sometimes just a single sample) from the underlying distributions to achieve constant-factor approximations \citep{BlumrosenD21,DuttingFLLR21,kang22fixed}.

    The recent work of \citet{HajiaghayiHPS25} studies the compatibility of profit and efficiency (both gains-from-trade and social welfare) in the Bayesian setting. While they demonstrate positive results for fixed-price mechanisms, they present strong impossibility results for general mechanisms.
    
    \xhdr{Sample Complexity.} Beyond the rich line of work on online learning, a related approach considers the \emph{offline} learning problem of economic design problems in the spirit of PAC-learning \citep[e.g.,][]{ColeR14,MorgensternR15,Devanur0P16,CaiD17,GonczarowskiW21}. While the standard online-to-offline reduction tells that sample-complexity results imply online regret bounds, we are not aware of any prior work yielding non-trivial results for our settings. For instance, \citet{ColeR14} consider the setting where bidders’ valuations are independent and exhibit some regularity; we consider a bilateral trade with arbitrary and possibly correlated distributions. \citet{MorgensternR15} provide results for finite pseudo-dimension; {in \Cref{app:complexity}} we prove that such measure is unbounded in our problem. \citet{GuoHGZ21} show general sample complexity bounds for product distributions. Their results imply sample complexity bounds for a range of applications, via appropriate discretizations of the value space. In contrast, a feature of our problem is the lack of a viable discretization (see \Cref{ex:fixed_grid}). They also provide results for \emph{monotone} problems, something our setting does not enjoy.

%% file: sections/20-model.tex
\section{The Learning Model and Preliminaries}
\label{sec:model}

    \xhdr{Profit in Bilateral Trade.} In the (static) bilateral trade problem, a seller and a buyer interact with each other. The seller holds a single good. The seller's and buyer's  (private) valuations for the good are $\vs \in [0,1]$ and $\vb \in [0,1]$, respectively. Seller and buyer submit a bid $\bs \in [0,1]$ and $\bb \in [0,1]$ (not necessarily truthful) to a broker running mechanism $M$. A mechanism is characterized by an allocation region $A \subseteq[0,1]^2$, and pricing rules $p ,q:[0,1]^2 \to [0,1]$. 
    The trade happens if and only if the bids $(\bs,\bb)$ belong to the allocation region $A$, while the payments are made according to $p$ and $q$. 
    Payments are from the broker to the seller, and from the buyer to the broker.
    That is, the seller receives $p(\bs,\bb)$, 
    while the buyer pays $q(\bs,\bb)$.
    We require that $p(\bs,\bb) = q(\bs,\bb) = 0$ whenever $(\bs,\bb) \not\in A$, i.e., there is no trade.\footnote{The requirement that payments are bounded in $[0,1]$ and zero when there is no trade is without loss of generality, given our focus on DSIC and IR mechanisms and the goal of maximizing the broker's profit. See \Cref{rem:wlog}.}

    The buyer's and seller's utility with valuations $\vb$ and $\vs$ 
    under bids $(\bb, \bs) \in [0,1]^2$ are: 
    \begin{equation*}
    \label{eq:utilities}
    \util{\bs}{\bb}{s} = \vs - \ind{(\bs,\bb)\in A} \cdot \vs + p(\bs, \bb), \util{\bs}{\bb}{b} = \ind{(\bs,\bb)\in A} \cdot \vb - q(\bs, \bb) \notag
    \end{equation*}
     We focus on dominant-strategy incentive compatible (DSIC) and individually rational (IR) mechanisms. There, each player maximizes his utility by a truthful bid regardless of the other player's bid, and the utility from participating in the mechanism is at least as high as that from not participating in the mechanism (i.e., it is at least $\vs$ for the seller and at least $0$ for the buyer).
    In formulas,
    \begin{align*}
        \text{DSIC:}\quad \util{\vs}{\bb}{s} &\geq \util{\bs}{\bb}{s} &&\forall \ \vs \in [0,1], (\bs, \bb) \in [0, 1]^2 \\
        \util{\bs}{\vb}{b} & \geq \util{\bs}{\bb}{b} &&\forall \ \vb \in [0,1], (\bs, \bb) \in [0, 1]^2 \nonumber \\
        \text{IR}:\quad\;\;\; \util{\vs}{\bb}{s} &\geq \vs, \,\, \util{\bs}{\vb}{b} \geq 0 &&\forall \ (\vs, \vb) \in [0, 1]^2, (\bs, \bb) \in [0, 1]^2 \nonumber
    \end{align*}

    We denote the class of all DSIC and IC mechanisms by $\M$. Since we consider DSIC and IR mechanisms, we will henceforth assume that the players' bids are equal to their respective valuations. 
    The profit earned by a broker running mechanism $M$ 
    is the difference between the payment the mechanism collects from the buyer and the payment that it makes to the seller. That is,
    $\prof(M,(\vs,\vb)) = q(\vs,\vb) - p(\vs,\vb)$.

    \xhdr{The Learning Protocol.} We study the repeated bilateral trade problem, defined as follows. At each time step $t$, a new seller-buyer pair arrives, with private valuations $v^t = (\vs^t,\vb^t) \in [0,1]^2$. The learner/broker proposes a mechanism $M^t$, and the agents bid. The trade occurs according to $M^t$ and the bids of the agents, which then leave forever. The broker earns the difference between the price paid \emph{by} the buyer and that paid \emph{to} the seller, i.e., $\prof(M^t, v^t)$.
    Formally, we consider the following learning protocol. 
    
        \begin{algorithm}[h!]
    \begin{algorithmic}[ht]
        \For{time $t=1,2,\ldots$}
            \State a new pair of agents arrives with private valuations $v^t=(\vs^t,\vb^t) \in [0,1]^2$
            \State the learner declares a mechanism $M^t$, and then the agents declare their bids
            \State the trade takes place according to mechanism $M^t$ and the bids
            \State the learner gains $\prof(M^t, v^t) \in [-1,1]$
        \EndFor
    \end{algorithmic}
    \caption*{\textbf{The Learning Protocol}}
    \end{algorithm}

    We aim to design a learning algorithm that minimizes the regret: 
    \[
        \sup_{M \in \M} \E{\sum_{t=1}^{T} \prof(M, v^t) - \sum_{t=1}^{T} \prof(M^t, v^t)}.
    \]
    The typical goal is to achieve sublinear worst-case regret; this implies that on any input, the algorithm's profit converges, on average, to that of the benchmark.

    We consider both the adversarial setting, where the sequence of valuations is generated up-front by an oblivious adversary, and the stochastic i.i.d. one, where the agents' valuations  $v^t = (\vs^t,\vb^t)$ are drawn i.i.d. from an unknown distribution. In the latter model, the $v^t$ are independent across time, but the seller and buyer valuations can be arbitrarily correlated.
    
    \xhdr{Characterization of Mechanisms.}    
    We next observe that the allocation regions and the pricing functions of DSIC and IR mechanisms have a distinctive structure: the allocation region needs to respect a certain monotonicity property, while the allocation region uniquely induces the payments. We formally introduce monotonicity and Myerson payments \citep{myerson81}.

    \begin{definition}[Monotone Allocation Region]
    \label{def:monotone_allocation}
        We define the partial ordering $\preceq$ on $[0,1]^2$ as follows: for any two points $x=(x_1,x_2)$ and $y=(y_1,y_2)$, we say that $x\preceq y$ ($y$ dominates $x$) if $x_1 \ge y_1$ and $x_2 \le y_2.$ 
        An allocation region $A\subseteq[0,1]^2$ is monotone (with respect to $\preceq$) if for any $x \in A$ and $y \in [0,1]^2$ with $x \preceq y$ it holds that $y \in A$.
    \end{definition}

    \begin{definition}[Myerson Payments]
    \label{def:payments}
        Let $A$ be a monotone allocation region, then the associated Myerson payments are defined as follows:
        \begin{align*}
            p(\bs,\bb) =& \ind{(\bs,\bb) \in A} \cdot \max\{x \in [0,1] \mid (x,\bb) \in A\} &&\forall\; (\bs,\bb) \in [0,1]^2\\
            q(\bs,\bb) = &\ind{(\bs,\bb) \in A}\cdot \min\{y \in [0,1] \mid (\bs,y) \in A\} &&\forall\; (\bs,\bb) \in [0,1]^2 
        \end{align*}
    \end{definition}

    In words, for an allocation region to be monotone it should be ``closed'' in a north-west direction. That is, if a point $(\bs,\bb)$ is in $A$ then any point $(\bs',\bb)$ with $\bs' \leq \bs$ and any point $(\bs,\bb')$ with $\bb' \geq \bb$ should be in $A$ as well. The payments of a point $(\bs,\bb) \in A$, in turn, correspond to the south-east projection to the allocation boundary. See \Cref{fig:myerson_prices} for illustrations.


    \begin{restatable}{proposition}{characterization}
    \label{thm:mechanisms}
        A mechanism $M$ for bilateral trade is dominant-strategy incentive compatible and individually rational if and only if its allocation region is monotone and the payments are Myerson payments.
    \end{restatable}

    The proof of \Cref{thm:mechanisms} follows by standard arguments \citep{myerson81,MyersonS83}. For completeness, we provide a proof in \Cref{app:mechanisms}. An important consequence of \Cref{thm:mechanisms} is that it reduces the learning problem to learning the monotone allocation region $A$, with the learner's profit given by Myerson payments.
    Note, though, that this is far from trivializing the problem, as we formally argue in \Cref{app:complexity}.

    \xhdr{Profit and Properties of Optimal Mechanism.} 
    We conclude the preliminaries with some observations regarding the tradeoffs involved in choosing different (monotone) allocation regions. The first observation is that by shrinking the allocation region we loose potentially profitable trades, while we increase the profit of those trades that still happen.

    \begin{observation}[Trades vs.~Profit Trade-off]
        \label{prop:trade-off}
        Consider any two mechanisms $M,M'$ in $\M$ with allocation regions $A, A'$ such that $A \subseteq A'$. 
        For any  $v \in [0,1]^2$ we have  one of three cases: 
        \begin{itemize}
            \item $v \notin A'$, and  $\prof(M,v)=\prof(M',v)=0$
            \item $v \in A'\setminus A$, and $\prof(M,v)=0$
            \item $v \in A$, and $\prof(M,v) \ge \prof(M',v)$
        \end{itemize}
    \end{observation}

    This trade-off allows us to derive a useful property of the optimal mechanism: its allocation region does not contain the lower-right triangle! Denote with $U$ the upper-left triangle, $U =\{(x,y) \in [0,1]^2\, | \, x\le y\}$, and with $\Mplus$ the subfamily of $\M$ containing the mechanisms whose allocation region is contained in $U$. We have the following lemma. 
    
    \begin{lemma}
    \label{lem:Mplus}
        For any mechanism $M\in \M$ there exists a mechanism $M^+ \in \Mplus$ such that, for any valuation $v$, the following inequality holds:
        \(
            {\prof(M^+,v)} \ge  {\prof(M,v)}.
        \)
    \end{lemma}
    \begin{proof}
        Fix any mechanism $M \in \M$ with allocation region $A$ and payment functions $p,q$, and consider the mechanism $M^+$ defined by the allocation region $A^+=A \cap U$, where $A$ is the allocation region of $M$. The set $A^+$ is monotone and thus is a valid allocation region. Fix any valuation pair $v = (\vs,\vb)$; we want to prove that $\prof(M^+,v) \ge \prof(M,v).$ From \Cref{prop:trade-off}, we know we only need to consider the case in which 
        $v \in A \setminus A^+$. If $v \in A \setminus A^+$, then it does not belong to $U$ and we have the following inequality: 
        \(  \prof(M,v) = (q(v) - p(v)) \le (\vb - \vs) \le 0 = \prof(M^+,v),
        \)
        where the first inequality holds by individual rationality, and the second because $v \notin U$. 
    \end{proof}

    
    \begin{remark}[The Allocation Regions Are Closed]
        For convenience, we always assume that the allocation regions are closed. This is without loss of generality; in fact, for any allocation region in $U$, if we consider its closure, we get a mechanism that extracts at least as much revenue. This assumption allows us to simplify and formalize some arguments; for instance, if $A$ is closed, then the $\max$ and $\min$ operations in the definition of the payments are well defined. 
    \end{remark}

%% file: sections/30-proof_sketch.tex
\section{Proof Ideas at a Glance}
\label{sec:challenges}

    In this section, we briefly summarize the ideas behind our $\tilde{O}(\sqrt T)$ result in the stochastic setting, that are formally laid down in \Cref{sec:algorithm}. In this setting, valuations $v^1, v^2, \dots$ are drawn i.i.d. according from a fixed but unknown distribution, and our goal is to approximate the best mechanism $\Mopt$ on the valuation distribution.\footnote{For simplicity, we assume here that the best mechanism exists.}
    {Consider what happens at the generic time step $t+1$, after receiving samples $v^1,\ldots,v^{t}$. If we could determine}
    a mechanism $M^{t+1}$ such that 
    \begin{equation}
     \label{eq:generalization} \mathbb{E}\left[\prof(\Mopt,v)\right] \leq \mathbb{E}\left[\prof(M^{t+1},v) \right] + t^{-\alpha}  
    \end{equation}
    for some exponent $\alpha \in (0,1)$, then summing up over all time step, we would get an overall regret $\sum_{t=1}^T t^{-\alpha} \in {O}(T^{1-\alpha})$. In particular, for $\alpha = \nicefrac 12$, we would obtain the optimal rate of $\sqrt{T}$. 
    {From \Cref{ex:fixed_grid} (and \Cref{thm:uniform-learning}), we know that past samples are not enough to estimate \emph{at the same time} all mechanisms' expected profit. It is, in fact, possible that the empirical profit of some mechanism can be arbitrarily far from its true expected profit. This rules out a Follow-The-Leader approach, which would update $M^t$ to be the empirical optimum after every sample.} 



    
    Since we are ultimately only interested in Equation \ref{eq:generalization}, a natural idea is to constrain the candidate mechanisms to some subset $\mathcal{M}^*\subset \mathcal{M}$ such that: (i) we can show that uniform convergence holds for all $M\in \mathcal{M}^*$ and (ii) the best mechanism in $\mathcal{M}^*$ has nearly the same profit as $\Mopt$.
    Indeed, {given an arbitrary mechanism $M$ and precision parameter $\eps$, we can construct an ``$\eps$-simple'' mechanism $M'$ that well approximates the profit of $M$. The simplicity of $M'$} is captured by the fact that the boundary of its allocation region is determined by at most $O(\nicefrac 1{\eps})$-axis parallel segments, while the approximation property is formalized by the following inequality:
    \begin{equation}
    \label{eq:weakguarantee}
        \mathbb{E}\left[\prof(M,v)\right] \leq \mathbb{E}\left[\prof(M',v)\right] + \varepsilon.
    \end{equation}
    Limiting the mechanisms to the set $\mathcal{M}_{\varepsilon}$ of such $\eps$-simple mechanisms turns out to be a fruitful idea as proving a uniform convergence for these mechanisms is indeed possible. 
    Our algorithm \simplify computes the mechanism $\Memp$ that maximizes the empirical profit and then posts the mechanism in $\Meps$ that approximates it. 
    {While the family of $\eps$-simple mechanisms is inspired by a similar construction in \citet{AggarwalBDF24}, their work only implies a convergence rate of $\alpha = \nicefrac 14$ (for a regret rate of order $T^{\nicefrac 34}$) via a standard VC dimension argument. As a warm-up, we next present a uniform convergence result on $\Meps$ ($\alpha = \nicefrac 13$), then we argue how to further refine our discretization argument to attain the optimal $\alpha = \nicefrac 12.$}
    
    \subsection[Warm Up: A T to the 2/3 Regret Bound.]{Warm Up: A $\Tilde{O}\left(T^{2/3}\right)$ Regret Bound.}
    There are infinite mechanisms that can be described by $O(\nicefrac 1 \eps)$ axis parallel segments, so we do not have a finite bound on the size of $\mathcal{M}_{\varepsilon}$. One might be tempted to impose a fixed set of such lines as candidate boundary segments. Unfortunately, this approach does not work (\Cref{ex:fixed_grid}) without prior knowledge of the distribution.

    However, we are not interested in learning \emph{all} $\eps$-simple mechanisms, but only those relevant to the realized samples. While the distribution itself does not have finite support, the empirical one induced by the samples necessarily does. We then restrict our attention to the $\eps$-grid induced by coordinates of the samples, leading to $\binom{t}{\nicefrac 1{\varepsilon}} \approx \exp(\nicefrac{\log t}\eps)$ many candidate mechanisms $\mathcal{M}_{\varepsilon}$. 
    Using the randomness of the sample to prove concentration for mechanisms induced by the sample introduces dependencies; e.g., a mechanism's profit is highly correlated with the $O(\nicefrac{1}{\eps})$ samples that determine its boundary.
    We address these correlations
    by removing these $O(\nicefrac{1}\eps)$ samples when proving a union bound, which still leaves us with $t-O(\nicefrac{1}\eps)$ fresh samples to prove concentration. With this setup, we can show via standard concentration inequalities that
    \begin{equation}
    \label{eq:unionboundbasic}
        \sup_{M\in \mathcal{M}_{\varepsilon}}\left\vert \frac{1}{t}\sum_{i=1}^{t}\prof(M, v^i) - \mathbb{E} \left[\prof(M, v)\right]\right\vert \lesssim \sqrt{\frac{\log|\mathcal{M}_{\varepsilon}|}t} \lesssim \sqrt{\frac{\log T}{\eps t}}
    \end{equation}

    The optimal mechanism $\Mopt$ is one single mechanism, so we can safely assume its empirical performance concentrates on past data. We have the following chain of inequalities:
    \begin{align}
        \E{\prof(\Mopt,v)}&\lesssim  \frac{1}{t} \sum_{i=1}^{t} \prof(\Memp,v^i) + \sqrt{\frac{\log T}{t}} \tag{$\Mopt$ concentrates and by design of $\Memp$}\\
        &\lesssim  \frac{1}{t} \sum_{i=1}^{t} \prof(M^{t+1},v^i) + \eps + \sqrt{\frac{\log T}{t}} \tag{\Cref{eq:weakguarantee} on empirical distribution}\\
        &\lesssim \E{\prof(M^{t+1},v)} + \eps + \sqrt{\frac{\log T}{t}} 
 + \sqrt{\frac{\log T}{\eps t}}, \label{eq:final_ineq}
    \end{align}
    where the last inequality is due to \Cref{eq:unionboundbasic}.
    The $\tilde O(T^{\nicefrac 23})$ bound follows by setting $\eps \approx t^{-\nicefrac 13}$. 



        \begin{figure}[t!]
            \centering
            \begin{subfigure}{0.3\textwidth}
                \centering
                \includegraphics[width=\linewidth]{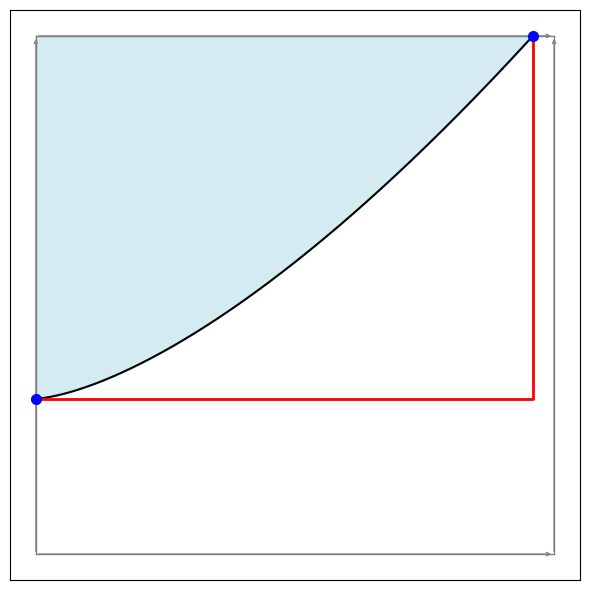}
                \caption{Boundary of $\app{0}{M}$.}
                \label{fig:approx_0}
            \end{subfigure}
            \hfill
            \begin{subfigure}{0.3\textwidth}
                \centering
                \includegraphics[width=\linewidth]{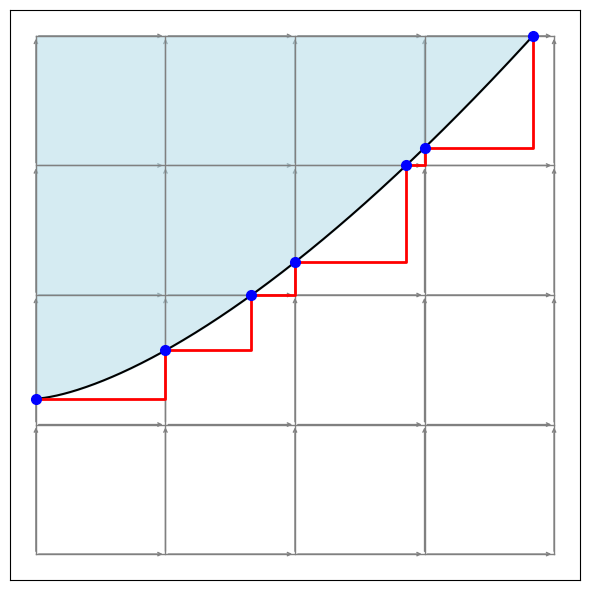}
                \caption{Boundary of $\app{2}{M}$.}
                \label{fig:approx_2}
            \end{subfigure}
            \hfill
            \begin{subfigure}{0.3\textwidth}
                \centering
                \includegraphics[width=\linewidth]{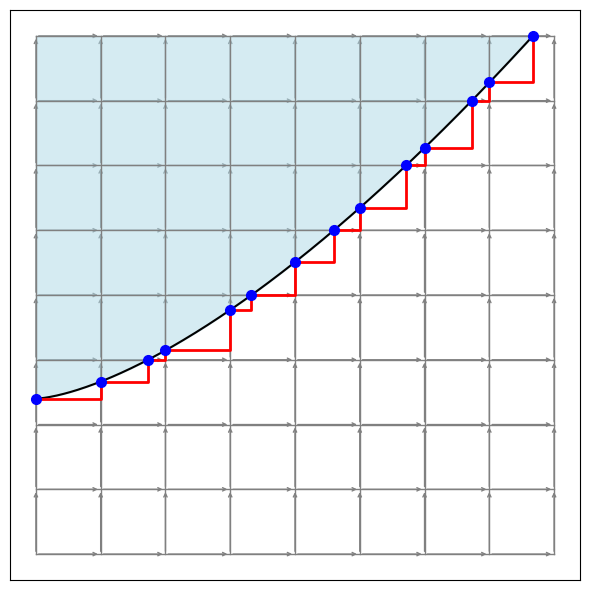}
                \caption{Boundary of $\app{3}{M}$.}
                \label{fig:approx_3}
            \end{subfigure}
            \caption{Visualization of approximating sequence (in red) of mechanism $M$ (blue allocation region).}
            \label{fig:approx_grid}
        \end{figure}
    \subsection[Main Result: sqrt T Regret]{Main Result: $\Tilde{O}(\sqrt{T})$ Regret}

    The $\Tilde{O}(T^{\nicefrac 23})$ provides the important insight that among the $\exp(\nicefrac{(\log t)}{\eps})$ many candidate mechanisms in $\Meps$ induced by the samples, there is an $M^t$ whose expected profit is with an additive error of $t^{-\nicefrac 13}$ of the expected profit of $\Mopt$. 
    To ensure that the empirical profit of $M^t$ is close to the expected profit of $M^t$, the previous approach consisted in simply performing a union bound over all such candidate $\exp(\nicefrac{(\log t)}{\eps})$ mechanisms.
    This approach is very natural, as it is difficult to characterize profitable from unprofitable mechanisms a priori.
    At the same time, the approach is lossy: We actually do not need strong concentration for unprofitable mechanisms in $\mathcal{M}_{\varepsilon}$.

    We now describe how our chaining approach, originally designed to perform very efficient union bounds for uniform convergence, can provide some intuition on how to identify profitable mechanisms on the samples.
    While we cannot guarantee that we estimate the profit of all mechanisms in $\mathcal{M}_{\varepsilon}$ accurately, we show that this is possible for profitable mechanisms, while unprofitable mechanisms are either identified as such, or have their profit underestimated. 
    Suppose we have a set of mechanisms $\mathcal{N}_{h}$, such that for every $M \in \M$, there exists $M'\in \mathcal{N}_{h}$ with 
    \begin{equation}
      \label{eq:desire}
      \left\vert\prof(M,v^i)-\prof(M',v^i)\right\vert \leq 2^{-h}\quad \forall v^1,v^2, \dots, v^{t}
    \end{equation}
    A crucial ingredient in chaining is to decompose the target profit $\prof(M,\cdot)$ via a telescopic argument that uses a sequence of approximating mechanisms as follows
    \begin{equation*}
        \prof(M,v^i) = \prof(\app{0}{M},v^i) + \sum_{h=0}^{\log \nicefrac{1}{\eps}} \prof(\app{h+1}{M},v^i) - \prof(\app{h}{M},v^i),
    \end{equation*}
    where each $\app{h}{M}$ belongs to the corresponding $\cN_h$. This sequence of approximating mechanisms is visualized in \Cref{fig:approx_grid}. Our goal is to learn the expected profit of $M$, which translates to learning separately the terms of the above formula. The zero-th term can be addressed with standard union bound techniques by noting that $\app{0}{M}$ is a fixed-price mechanism (see \Cref{fig:approx_0}), so we only focus on learning the increments.  
    The idea inherent in chaining is that instead of performing a union bound on all target mechanisms in one shot, we perform a union bound on the set of successive increments. Specifically, by combining basic union bounds similar to those used in \Cref{eq:unionboundbasic} and a reduced variance due to \Cref{eq:desire}, we have
    \begin{align*}
    \sup_{\substack{(M_{h\!+\!1},M_{h}) \\ \in\mathcal{N}_{h\!+\!1}\times  \mathcal{N}_{h}}}\! \left\vert \sum_{i=1}^t\frac{1}{t}\big(\prof(M_{h\!+\!1},v^i) \!-\! \prof(M_{h},v^i)\big)\!  - \!\E{\prof(M_{h\!+\!1},v) \!- \! \prof(M_{h},v)}\right\vert \!\lesssim\!\sqrt{\frac{\log|N_{h\!+\!1}| |\cN_{h}|}{t\cdot 2^{-2h}}}.
    \end{align*}
    At this point, we wish to find a careful tradeoff between the magnitude of increments and the size of the families $\cN_{h}$. 
    Optimistically, we would attempt to use $\mathcal{M}_{\varepsilon}$ as $\mathcal{N}_{h}$, for appropriate values of $\varepsilon$, which would bound the right-hand-side of the above inequality with 
    \begin{equation}
        \label{eq:increments}
        \frac{1}{2^h}\sqrt{\frac{\log|\cN_{h+1}| |\cN_{h}|}{t}} \lesssim  \frac{1}{2^h}\sqrt{\frac{2^{h}\log t}{t}}\lesssim \sqrt{\frac{\log t}{t\cdot 2^h}}.
    \end{equation}
    Summing these terms for all increments would then yield
    \begin{align*}
        &\sup_{M \in \M_{\eps}} \left\vert \frac{1}{t}\sum_{i=1}^t \prof(M,v^i) - \mathbb{E}\left[\prof(M,v)\right] \right\vert \lesssim \sum_{h\geq 0} \sqrt{\frac{\log t}{t\cdot 2^h}} \lesssim \sqrt{\frac{\log T}{t}},
    \end{align*}
    achieving the desired regret of $\Tilde{O}\left(\sqrt{T}\right)$ (via an analysis similar to the one in \Cref{eq:final_ineq}).
    
    Unfortunately, there is a major issue with this (ideal) analysis: mechanisms in $\mathcal{M}_{\varepsilon}$ do not satisfy \Cref{eq:desire}, as they only provide a \emph{one-sided} bound!
    This requires several major adjustments to the chaining analysis; overcoming them constitutes the core technical contribution of our work.
    Our first important ingredient is to show that for any mechanism $M$, there exists $M'\in\mathcal{M}_{\varepsilon}$ such that 
    \begin{equation}
        \label{eq:desire2}
        \prof(M,v^i)-\prof(M',v^i) \leq \varepsilon \quad \forall i\in \{1,\ldots t\}.
    \end{equation}
    Notice the absence of absolute values when comparing this guarantee with \Cref{eq:desire}. 
    By itself, \Cref{eq:desire2} does not facilitate the chaining analysis. 
    However, for any telescoping sequence in which the increments are strictly increasing, that is $\prof(\app{h+1}{M},v^i) \geq \prof(\app{h}{M},v^i)$ for all $v^i$, it turns out that \Cref{eq:desire2} is sufficient to prove concentration.

    Let us interpret when this happens and more importantly when this does not happen. If $\prof(\app{h+1}{M},v^i) \geq \prof(\app{h}{M},v^i)$ for all $v^i$, then $\app{h+1}{M}$ extracts slightly more profit than $\app{h}{M}$ from the distribution. In this case, $\app{h+1}{M}$ is clearly preferable to $\app{h}{M}$ on the sample. Fortunately, \Cref{eq:desire2} implies $|\prof(\app{h+1}{M},v^i) - \prof(\app{h}{M},v^i)|\leq 2^{-h}$, which implies with the optimistic chaining analysis implies that the empirical profit of $\app{h+1}{M}$ is concentrated around its expectation.

    Next, let us consider the more challenging case where $\prof(\app{h+1}{M},v^i) \leq  \prof(\app{h}{M},v^i)$. In this case, a possible configuration is $\prof(\app{h+1}{M},v^i)\leq 0$ and $\prof(\app{h}{M},v^i) \geq \nicefrac 12$ for a few samples\footnote{These numbers are possible, but chosen arbitrarily for the sake of the argument. The main point is to emphasize that the absolute value of differences can be a fairly large constant.}. Based on a single sample $v^i$, it is typically not possible to determine whether $\app{h+1}{M}$ is more profitable than $\app{h}{M}$ as it can exclude the potential profit of $v^i$ for more profit on other candidate samples.
    Unfortunately, in this case, we also have
    \begin{equation}
        \label{eq:badcase}
        \left\vert \prof(\app{h+1}{M},v^i) -\prof(\app{h}{M},v^i)\right\vert \geq \tfrac 12.
    \end{equation}
    This is fatal for a chaining analysis: If this happens too many times, then the increments (i.e. the right-hand-side of \Cref{eq:increments}) can only be bounded by $\sqrt{\nicefrac{2^{h}}{t}\cdot \log t}$, and summing up over all increments yields no concentration whatsoever.

    Our main insight is that if \Cref{eq:badcase} happens for too many samples $v^i$, then $\app{h+1}{M}$ \emph{has less profit than $M_h$ and cannot be the empirical optimum!} Moreover, the same also holds for any mechanism that refines $\app{h+1}{M}$ via \Cref{eq:weakguarantee}, allowing us to remove all of these mechanisms for the search space of potentially optimal mechanisms (i.e., we are breaking the chain).
    Thus, we can show that a (slightly weaker) form of \Cref{eq:increments} holds for mechanisms with sufficiently large profit, which still allows us to obtain the desired $\Tilde{O}(\sqrt{T})$ regret bound.

    In a nutshell, we have the following dichotomy working in our favor (formalized in \Cref{prop:local}): 
    \begin{itemize}
        \item Either all but a very tiny fraction of the increments are positive. In this case, we can continue the chaining analysis.
        \item Or sufficiently many increments have a large negative value. In this case, the mechanism is suboptimal, and we no longer have to prove concentration.
    \end{itemize}
    
    While the remaining analysis still has substantial technical challenges, the on-the-fly construction and truncation of chains ultimately leading towards optimal or nearly optimal mechanisms is our main novelty.  We therefore believe that our work can be a starting point for proving low regret bounds for other problems lacking strong characterizations of optima. Indeed, as a consequence of our work, we also manage to extend our analysis to the Joint-Ads problem.

%% file: sections/40-stochastic_upper_bound.tex
\section{A \texorpdfstring{$\tilde O(\sqrt{T})$}{} Learning Algorithm for the Stochastic Setting}
\label{sec:algorithm}

    In this section, we present the learning algorithm \simplify and prove that it enjoys a $\tilde O(\sqrt{T})$ regret rate in the stochastic setting.
    
    \subsection[Simplify-the-Best-Mechanism]{\simplify}
    \label{subsec:simplify}
        For each time step $t$, the learning algorithm \simplify behaves in two phases. First, it computes the best mechanism $M^\star_t$ on the past data; this operation can be reduced efficiently to computing the shortest path on a suitably defined directed graph. Second, it finds a simplified version of the optimal mechanism $M^\star_t$ by adapting its allocation region to a uniform grid of step-size $\eps_t \approx t^{-\nicefrac 12}$. We refer to the pseudocode for further details, while we present separately in \Cref{subsec:maximizer} and \Cref{subsubsec:simplify} the implementation details on the two phases of the algorithm. 

        \begin{algorithm}[ht]
        \begin{algorithmic}[ht]
            \State \textbf{Environment}: Valuations $\{(v_1^t,v_2^t)\}$ are drawn i.i.d.    
            \State  Propose any mechanism $M^1$ and observe $(v^1_1,v_2^1)$ 
            \For{time $t=2,\ldots,T$}
                \State Set precision parameter $\eps_t = {200 \sqrt{\nicefrac{\log^3 T}{t-1}}}$
                \State Let $M^\star_t$ be the profit-maximizing mechanism over valuations $(v^1_1,v_2^1), \dots, (v^{t-1}_1,v_2^{t-1})$
                \State Let $M^t$ be the simplified mechanism with precision $\eps_t$ associated with $M^\star_t$
                \State Propose mechanism $M^t$ and observe $(v^t_1,v_2^t)$
            \EndFor
        \end{algorithmic}
        \caption*{\textbf{\simplify}}
        \label{alg:adaptive}
        \end{algorithm}

    \subsubsection{Finding the Best Mechanism}
    \label{subsec:maximizer}

        The task of finding the best mechanism is inherently offline: given a set $S$ of $n$ points in the $[0,1]^2$ square, how can we compute the empirical best mechanism with respect to $S$? Surprisingly, this problem can be reduced to finding the shortest path on a suitable directed graph. We have the following definition: 
        \begin{definition}[Point-Based Grid]
        \label{def:grid}
            Given a set $S = \{(v_s^i,v_b^i)\}$ of $n$ points, we introduce the projections over the $x$ and $y$ axis. $S_x$ (respectively $S_y$) contains in increasing order the values $0$, $1$, and $v_s^i$ (respectively $v_b^i$) for $i=1, \dots, n.$ The grid $G_S$ associated to $S$ is a directed graph defined as follows: 
            \begin{itemize}
                \item The nodes of $G_S$ are represented by the points in the cartesian product $S_x \times S_y.$ 
                \item Horizontal and vertical segments represent the edges of $G_S$: 
                \begin{itemize}
                    \item $(x_i,y) \to (x_{i+1},y)$ for any $y \in S_y$ and $x_i\in S_x\setminus \{1\}$
                    \item $(x,y_j) \to (x,y_{j+1})$ for any $x \in S_x$ and $y_j\in S_y\setminus \{1\}$
                \end{itemize}
            \end{itemize}
        \end{definition}

        Beyond the formal definition, the construction of the directed graph $G_S$ is pretty simple: it is just an axis-aligned orthogonal grid that contains all the points in $S$, the origin and $(1,1)$, and joins them via horizontal (pointing right) and vertical (pointing up) edges. We refer to \Cref{fig:point_based_grid}.

        \begin{remark}
            The grid $G_S$ --- as well as the other grids used in the rest of the paper --- have a double nature: they are directed graphs, but they live embedded in the $[0,1]^2$ square. Therefore, it is natural to treat a node in $G_S$ as a point in the square and an edge as a segment. So, we refer to the geometric properties of nodes and edges via this natural embedding. 
        \end{remark}

        We introduce grids to discretize the space of allocation regions. In particular, it is natural to associate mechanisms to paths on grids, as formalized by the following definition. 
        
        \begin{definition}[Complete Path and Associated Mechanism]
        \label{def:complete_path}
            A path $\pi$ in a point-based grid $G$ is called complete if it starts in $(0,x_1)$ and terminates in $(x_2,1),$ for some $x_1,x_2 \in [0,1]$. Any complete path $\pi$ univocally identifies a mechanism $M_{\pi}$ whose allocation region is defined as follows: 
            \[
                A_{\pi} = \{v \in [0,1]^2 \text{ such that } \exists v' \in \pi \text{ with } v\succeq v'\}\footnote{\text{Recall, $(v_1,v_2)\succeq (v_1',v_2')$ if $v_1 \le v_1'$ and $v_2 \ge v_2'.$}}.
            \]
        \end{definition}

        Complete paths on a point based grid $G_S$ have a simple structure: they start from a point on the left side of the $[0,1]^2$ square, then move either right or up in $G_S$, and terminate in some node on the upper side of the square. From a combinatorial perspective, it is fairly easy to \emph{count} how many complete paths are on a point based grid. 

        \begin{lemma}[Number of Complete Paths in $G_S$]
            \label{lem:n_of_paths}
            Let $G_S$ be a point induced grid with $|S| = n\ge 5$, then there are at most $(2e)^{2n}$ distinct complete paths in $G_S$.
        \end{lemma}
        \begin{proof}
            Consider the nodes of $G_S$. They form an orthogonal grid in the $[0,1]^2$ square composed by $(n+2)^2 \le 2n^2$ nodes, where in the previous inequality we used that $n\ge 5$. There are at most $2n^2$ ways of choosing the starting and the terminal node (at most $n+2$ nodes on the left side and $n+2$ nodes of the upper side of the square) and, once these two points are fixed, a complete path is characterized by at most $\binom{2(n+1)}{n+1}$ ways in which it can arise, since its length is at most $2(n+1)$ (from $(0, 0)$ to $(1, 1)$) and a $2(n+1)$-long complete path must go ``up'' $n+1$ times and ``right'' $n+1$ times. All in all, the number of complete paths is
            at most:
            \[
            2n^2 \cdot \binom{2(n+1)}{n+1} \le 4e n^2 (2e)^n \leq (2e)^{2n},
            \]
            where we use the inequalities $\binom{n}{k} \le (e\cdot \nicefrac{n}{k})^k$ and $2n^2 \le (2e)^n$ for $n\ge 5$.
        \end{proof}
        
        As a further step, we introduce the notion of edge weights. 
        \begin{definition}[Edge-Weights]
        \label{def:edge_weights}
            Given any point-induced grid, we define the weight $w_e$ and the influence region $A_e$ of the generic edge $e = (u;v)$\footnote{To avoid confusion, we denote with $(a,b)$ a point with coordinates $a$ and $b$, while $(a;b)$ represents the edge that goes from node $a$ to node $b.$} as follows:
            \[
            \begin{cases}
                A_e = [0,x] \times [u_2,v_2) \text{ and } w_e =-x \text{ if $u_1 = v_1 = x$}  \\
                A_e = [u_1,v_1) \times [y,1]  \text{ and } w_e  =y \text{ if $u_2 = v_2 = y$} 
            \end{cases}
            \]
            where we adopt the convention that if $v_2 = 1$ in the first case or $v_1 = 1$ in the second case, then the corresponding interval extreme is closed. 
        \end{definition}

        \begin{figure}[t!]
            \centering
            \includegraphics[width=0.3\linewidth]{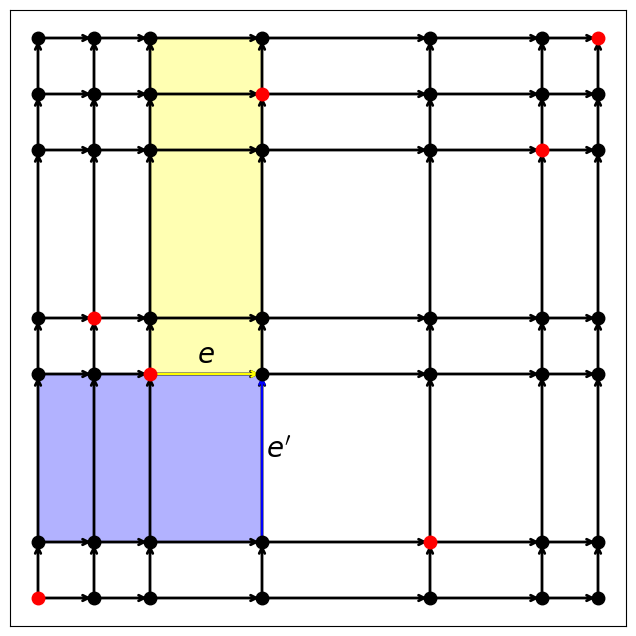}
            \caption{An instance of a point-based grid, as defined in \Cref{def:grid}. The points in $S$ are highlighted in red. The area filled with yellow is the influence region $A_e$ for the yellow horizontal edge $e$; the one filled with blue is the influence region $A_{e'}$ for the blue vertical edge $e'$.}
            \label{fig:point_based_grid}
        \end{figure}
        The edge weights are crucial in establishing a connection between mechanisms and paths.  

        \begin{lemma}
        \label{lem:decomposition}
            For any set $S$ of points and complete path $\pi$ on $G_S$, the following equality holds true:
            \[                
                \sum_{v \in S}\prof(M_{\pi},v) = \sum_{e \in \pi} w_e \cdot |A_e \cap S|.
            \]
        \end{lemma}
        \begin{proof}
            Consider any complete path $\pi$ and denote with $E_v$ its vertical edges and with $E_h$ its horizontal ones. By definition of influence regions $A_e$, it is immediate to verify that both $\{A_e\}_{e \in E_h}$ and $\{A_e\}_{e \in E_v}$ partition the allocation region $A_{\pi}$. Consider then any point $v \in S$. If $v\notin A_{\pi}$, then $M$ does not extract any profit from it; similarly, $v$ does not belong to any $A_e$ for $e$ edge of $\pi$, so that $v$ does not contribute to the right-hand side of the statement.  
            
            Consider now the other case: if $v$ belongs to $A_{\pi}$, then there exists exactly two edges --- a vertical edge $e$ and a vertical one $e'$ --- such that $v \in A_{e} \cap A_{e'}$. By definition of payments, the profit induced by $v$ according to $M_{\pi}$ is exactly $w_e + w_{e'}$. Summing up over all the points in $S$ yields the statement. 
        \end{proof}

        This structural lemma allows us to argue that the optimal mechanism on a fixed and finite set $S$ of points lies in the augmented grid $G_S$.

        \begin{theorem}
        \label{thm:complete_optimality}
            For any finite set of points $S$, the following equality holds true:
            \[
                \max_{M \in \M} \sum_{v \in S} \prof(M,v) = \sum_{v \in S} \prof(M_{\pi},v),
            \]
            for some complete path $\pi$ in $G_S.$
        \end{theorem}
        \begin{proof}
            For any set $S$ of points and mechanism $M \in \M$, we prove that there exists a complete path $\pi$ in $G_S$ such that the profit of $M_{\pi}$ on $S$ is an upper bound to the profit of $M$ on the same points. 

            Fix then any set of points $S$ and mechanism $M$ with allocation region $A$; we introduce the following set $A_S$:
            \[
                A_S =  \{v \in [0,1]^2|\, \exists\, v' \in A \cap S \text{ such that } v' \preceq v\}.
            \]
            The set $A_S$ is a monotone allocation region: consider in fact any $v \in A_S$ and any $\hat v \in [0,1]^2$ such that $v \preceq \hat v$. By definition of $A_S$ there exists $v' \in A \cap S$ such that $v' \preceq v$, which implies that $v' \preceq \hat v$ (by the transitive property); therefore $\hat v \in A_S$ as well.

            For what follows, we need to prove that there exists a complete path $\pi$ in $G_S$ such that $A_S = A_{\pi}$. For any $v=(v_1,v_2) \in [0,1]^2$, we define the monotone rectangle $R_v$ as $R_v = [0,v_1] \times [v_2,1]$; these rectangles are useful to characterize $A_S$. More precisely, the set $A_S$ is the union for $v \in S \cap A$ of $R_v$ (see \Cref{fig:complexity_plots} and the constructions in Appendix $\ref{app:complexity}$). This means that $A_S$ is a monotone region whose boundary is composed by the union of axis-oriented orthogonal segments. These segments - with rightward and upward orientations for horizontal and vertical edges respectively - all belong to the grid $G_S$, being included in the sides of the rectangles $R_v$ for some $v \in S$. Therefore, there exists a complete path $\pi$ so that $A_S$ is bounded by the $(0,0)$-$(0,1)$ segment, the $(0,1)$-$(1,1)$ segment, and by $\pi.$
            
            Finally, $A_\pi = A_S \subseteq A$, therefore $\prof(M,v) \le \prof(M_{\pi},v)$ for all $v \in A \cap S$ (by \Cref{prop:trade-off}). On the other hand, both $M$ and $M_{\pi}$ do not extract any profit from the points in $S \setminus A$, thus concluding the proof. 
        \end{proof}

        The above theorem implies that --- to find the profit maximizing mechanism on the observed data --- it is enough to solve a shortest path problem. 

        \begin{corollary}[Finding the Optimal Mechanism]
        \label{cor:empirical_maximizer}
            Given a set $S$ of points, it is possible to find efficiently a mechanism $M^{\star} \in \M^+$ supported on $G_S$ such that 
            \[
                \sum_{v \in S} \prof(M^{\star},v) = \max_{M \in \M} \sum_{v \in S} \prof(M,v). 
            \]
        \end{corollary}
        \begin{proof}
            The algorithm is simple: construct the point-induced graph $G_S$ and find a shortest complete path $\hat \pi$ in $G_S$, where the edge weights are defined as follows: 
            \begin{equation}
                \label{eq:new_weights}
                w^S_e = - w_e \cdot |A_e \cap S|.
            \end{equation}
            Note, the graph $G_S$ is acyclic, so it is possible to find efficiently the shortest path using the Bellman-Ford algorithm. Moreover, to restrict the optimization to complete paths, it is enough to add a dummy source node $s$ (with outgoing zero-weight edges pointing to all the nodes on the $(0,0)$-$(0,1)$ segment) and a dummy target node $t$ (with incoming zero-weight edges from all the nodes on the $(0,1)$-$(1,1)$ segment). To force the path not to contain any point outside $U$, the upper-left triangle, then it is enough to remove all nodes in $G_S$ that are not in $U$. As already argued in \Cref{lem:Mplus}, this is without loss of generality. 

            Let  $\pi^\star$ be the complete path satisfying the equality in \Cref{thm:complete_optimality}.
            We have the following chain of inequalities: 
            \begin{align*}
                \sum_{v \in S} \prof(M_{\hat \pi},v) &= \sum_{e \in \hat \pi} w_e|A_e \cap S| \tag{By \Cref{lem:decomposition}}     \\
                &=- \sum_{e \in \hat \pi} w_e^S \tag{By weight definition in \Cref{eq:new_weights}}\\
                &\ge -  \sum_{e \in \pi^\star} w_e^S \tag{$\hat \pi$ shortest path}\\
                &= \sum_{e \in \pi^\star} w_e|A_e \cap S| \tag{By weight definition in \Cref{eq:new_weights}}\\
                &=\sum_{v \in S} \prof(M_{\pi^\star},v)  \tag{By \Cref{lem:decomposition}}\\
                &= \max_{M \in \M}\sum_{v \in S} \prof(M_{\pi^\star},v),
            \end{align*}
            where the last equality holds by the choice of $\pi^\star$. This concludes the proof.
            
        \end{proof}

    \subsubsection{The Simplification Procedure}
    \label{subsubsec:simplify}

        We introduce a procedure that allows us to associate to each mechanism $M$ a ``simplified'' version that is adherent to a certain uniform grid.
        \begin{definition}[Boundary]
        \label{def:boundary}
            For any mechanism $M$, we denote with $\partial M$ its boundary:
            \[
                \partial M = \left\{v \in A_M |\, \forall \eps > 0, \left([0,1]^2\setminus A_M\right) \cap B_{\eps}(v) \neq \emptyset\right\}    
            \]
        \end{definition}

        In words, each mechanism is characterized by the curve $\partial M$ that starts from the $(0,0)$-$(0,1)$ vertical segment and terminates in the $(0,1)$-$(1,1)$ horizontal one. In particular, for point-based grids, the complete path that induces a mechanism is its boundary. From a different perspective, $\partial M$ is the (topological) boundary $\partial A$ of the corresponding allocation region $A$, when considering the subspace topology on $[0, 1]^2$.
        
        \begin{observation}
            The boundary $\partial M$ of a mechanism $M$ characterizes the allocation $A$ as any point in $A$ is dominated by at least a point on $\partial M$.
        \end{observation}

        We refer to points on the curve sorted according to the natural left-to-right and bottom-up parametrization: $v \in \partial M$ comes before $v'\in \partial M$ if $v_i \le v_i'$ for $i=1,2$, where this total ordering is well posed because of the monotonicity of the allocation region. In general, $\partial M$ may be utterly complicated (see also \Cref{app:complexity}); for instance, if we look at the profit-maximizing mechanism over $n$ points, the corresponding $\partial M$ may have $O(n)$ ``corners''. We present a generic procedure to simplify allocation regions without affecting too much the corresponding profit.
        
        For any $\eps>0$, we denote with $G_{\eps}$ the grid based on the points $(i\eps, i\eps)$ for $i=1, \dots, \lfloor \nicefrac 1{\eps}\rfloor$, according to \Cref{def:grid}. For any mechanism $M$ and precision $\eps$, we denote with $S_M^{\eps}$ the set of points defined as the union between the points of the form $(i\eps, i\eps)$ for $i=1, \dots, \lfloor \nicefrac 1{\eps}\rfloor$ and the points at the intersections of $G_{\eps}$ and $\partial M$. Note, we are intersecting a curve $\partial M$ with the union of segments/edges in $G_{\eps}$, and we adopt the convention that if $\partial M$ and $G_{\eps}$ share a segment, then the intersection only contains the first and last shared point of such segment (where first and last are defined according to the edge direction or equivalently to the ordering previously specified for the curve $\partial M$). We denote with $G_M^{\eps}$ the grid induced by the points $S_{M}^{\eps}$, again according to \Cref{def:grid}.

        \begin{figure}[t!]
            \centering
            \includegraphics[width=0.3\linewidth]{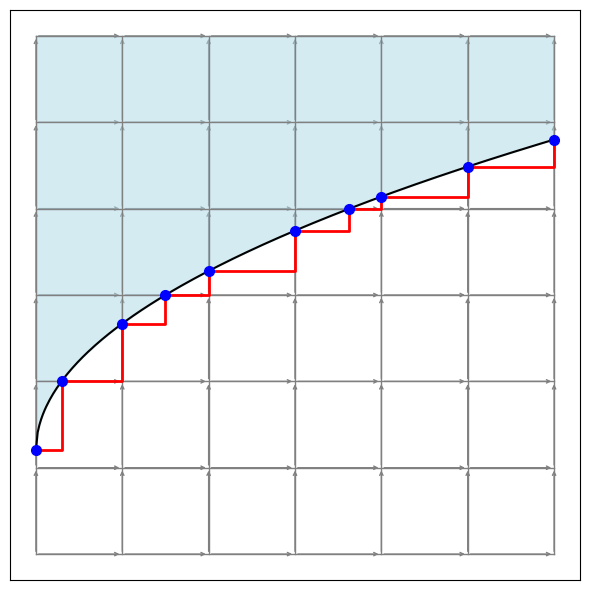}
            \caption{Visualization of the approximation (red) of a mechanism $M$ (blue allocation region), with precision $\eps = \nicefrac{1}{6}$. Note, the blue points are $\partial M \cap G_\eps$.}
            \label{fig:curve_aproximation}
        \end{figure}
        
        \paragraph{Constructing the Simplified Mechanism.} We now explain how to construct the complete path $\pi_{M}^{\eps}$ in $G_M^{\eps}$ that defines the simplified mechanism $M_{\eps}$ (we refer to \Cref{fig:curve_aproximation} for a visualization). $\pi_{M}^{\eps}$ starts from the intersection of $\partial M $ with the $(0,0)$-$(0,1)$ segment and terminates at the intersection between $\partial M$ and the $(0,1)$-$(1,1)$ segment. It contains all the points in the intersection of $\partial M$ and $G_\eps$, so we only need to specify how to interpolate between two consecutive  points there. Based on the specified total ordering for the boundary, let $v=(v_1,v_2)$ and $v'=(v'_1,v_2')$ be two consecutive points in the intersection of $\partial M$ and $G_\eps$; we have two cases: 
        \begin{itemize}
            \item If there exists an edge $G_M^{\eps}$ that connects them, then add this edge to the path.
            \item If no such edge exists, consider the node $u$ of coordinates $(v_1',v_2)$, which belongs to $G_M^{\eps}$. The path $\pi_M^{\eps}$ contains the directed edges $v$-$u$ and $u$-$v'$. 
        \end{itemize}

        \begin{definition}[Simplified Mechanism]
        \label{def:simplification}
            For any mechanism $M \in \M$ and precision $\eps$, we denote with $M_{\eps}$ the simplified mechanism induced by $\pi_M^{\eps}$. 
            
        \end{definition}

        \paragraph{Algorithmic Aspects.} We have presented a procedure to compute the path $\pi_M^{\eps}$, and we observe that this can be performed efficiently by our algorithm. From an algorithmic perspective, the only unclear passage is how to find intersections of $\partial M^\star$ and $G_{\eps}$, where $M^\star$ is the optimal mechanism on the points observed so far, from Corollary \ref{cor:empirical_maximizer}. This can be done efficiently as $\partial M^\star$ and $G_\eps$ are (finite) unions of horizontal and vertical segments/edges of known coordinates. In particular, $\partial M^\star$ contains at most $O(n)$ segments, while $G_{\eps}$ has $O(\nicefrac{1}{\eps^2})$ segments.

        \paragraph{Properties of the Simplified Mechanism.} Before moving on, we present some simple but crucial properties of simplified mechanisms. 
        \begin{proposition}[Properties of Simplified Mechanisms]
        \label{prop:simplified_mechanisms}
            Consider any mechanism $M \in \M$ and precision $\eps$, then the following properties are verified:
            \begin{itemize}
                \item[(i)] $A_{M} \subseteq A_{M_{\eps}}$
                \item[(ii)] If $v \in A_M$, then $\prof(M,v) - \prof(M_{\eps},v) \le 2\eps$
                \item[(iii)] $S^\eps_{M_\eps} = S^\eps_{M}$
            \end{itemize}
        \end{proposition}
        \begin{proof}
            Consider any point $v \in A_M$: it is dominated by a point $v'$ on the boundary of $M$. If $v'$ belongs to $S^\eps_M$ then there is nothing else to prove, as $v' \in A_{M_{\eps}}$; otherwise, there exists two consecutive points $u=(u_1,u_2)$ and $w=(w_1,w_2)$ in $S^\eps_M$ such that $v'$ belongs to the section of $\partial M$ between $u$ and $w$. If $u_1\neq w_1$ and $u_2 \neq w_2$, this means that $v'$ belongs to the $[u_1,w_1] \times [u_2,w_2]$ rectangle, which dominates (in the sense of the relation $\preceq$) the point $(u_2, w_1)$, which is in the boundary of $M_{\eps}$ (by construction). This proves that $v$ belongs to $A_{M_{\eps}}$ when the coordinates of $v'$ are different. In particular, this argument tells us that each such point on the boundary of $M$ belongs to a rectangle in $A_{M_{\eps}}$, such that this rectangle has sides of length at most $\eps$ and is dominated by a point on the boundary of $M_{\eps}$. Also, the right and lower side of this rectangle belong to $\partial M_\eps$, by construction. The same holds - with a segment of length at most $\eps$ instead of a rectangle - when $u$ and $w$ coincide for one coordinate.

            Consider now point (ii), and fix any point $v = (u_1, w_2) \in A_M$. We have $\prof(M,v) = q-p$, as defined in \Cref{def:payments}. This means that there exists two points $u = (u_1,q)$ and $w = (p,w_2)$ in $\partial M$. By point (i), we know that $v\in A_{M_{\eps}}$, and that $\prof(M_{\eps},v) = q_{\eps} - p_{\eps}$, induced by the points $u_{\eps} = (u_1,q_{\eps})$ and $w_{\eps} = (p_\eps,w_2)$ in $\partial M_{\eps}$. By the construction used in the previous point, we have that $q_\eps$ is at most $q + \eps$, while $p_{\eps}$ is at least $p-\eps$. This implies point (ii). 
            
            We finally address point (iii). This is equivalent to proving that $\partial M_\eps$ and $\partial M$ have the same intersection with $G_\eps$. This is true directly by construction of $M_\eps$.
        \end{proof}

    \subsection{Approximating Sequence of Mechanisms}

        In this section, we show how to iteratively apply the simplification procedure to construct a sequence of approximating mechanisms, and analyze its properties. 

        \begin{definition}[Approximating Sequence] 
        \label{def:approximating_sequence}
            Let $M \in \M$ be any mechanism, then the sequence of approximating mechanism $\app{h}{M}$ is defined as $\app{h}{M}=M_{2^{-h}}$, for $h \in \mathbb N$.
        \end{definition}

        Stated differently, the sequence $\app{h}{M}$ approximates the target mechanism $M$ via the simplification procedure (see \Cref{def:simplification}) for geometrically decreasing precisions $2^{-h}$. We have some immediate properties that stem from the definition of $\app{h}{M}$.

        \begin{proposition}
        \label{prop:simple_properties}
            For any mechanism $M \in \M$, and indices $h'\ge h$, it holds that 
            \begin{itemize}
                 \item{(Coarser Composition Rule)} $\app{h}{M} = \app{h}{\app{h'}{M}}$
                \item{(Monotonicity)} $A_{\app{h'}{M}} \subseteq A_{\app{h}{M}}$
            \end{itemize}
        \end{proposition}
        \begin{proof}
            From point (iii) of \Cref{prop:simplified_mechanisms}, we know that $\partial M \cap G_{2^{-h'}} = \partial \app{h'}{M} \cap G_{2^{-h'}}$. Since $h'\ge h$, we have that $G_{2^{-h}} \subseteq G_{2^{-h'}}$, which, in turn, implies\footnote{$A \cap C = B\cap C$ implies $A\cap C' = B \cap C'$ if $C' \subseteq C$.} that $\partial M \cap G_{2^{-h}} = \partial \app{h'}{M} \cap G_{2^{-h}}$; since the intersection of the boundary with the grid characterizes the mechanism, this  concludes the proof of the first property: $\app{h}{M} = \app{h}{\app{h'}{M}}$. 

            We move to the second property. We have that $A_{\app{h}{M}} = A_{\app{h}{\app{h'}{M}}}$ as the two mechanisms are the same, while point (i) of \Cref{prop:simplified_mechanisms} implies that that $A_{\app{h}{\app{h'}{M}}} \supseteq A_{\app{h'}{M}}.$ All in all, we can conclude that $A_{\app{h'}{M}} \subseteq A_{\app{h}{M}}.$
        \end{proof}

        In the following proposition, we summarize the properties of approximating sequences for $\M^+$, the subset of $\M$ whose mechanisms have allocation region contained in the upper-left triangle of $[0, 1]^2$.

        \begin{proposition}[Properties of the Approximating Sequence]
        \label{prop:approximating_seq}
            Let $M$ be any mechanism in $\M^+$, the following properties holds true for any point $v \in [0,1]^2$:
            \begin{itemize}
                \item[\textnormal{(i)}] $\prof(\app{h'}{M},v) - \prof(\app{h}{M},v) \le 2 \cdot 2^{-h}$ for any $h'\ge h.$
                \item[\textnormal{(ii)}] $\prof(M,v) - \prof(\app{h}{M},v) \le 2 \cdot 2^{-h}$ for any $h$
                \item[\textnormal{(iii)}] If $v \notin A_{\app{h}{M}}$, then $\prof(\app{h'}{M},v) = 0$ for all $h' \ge h$
                \item[\textnormal{(iv)}]  If $v \in A_{\app{h}{M}}$, then $\prof(\app{h'}{M},v)$ is non-decreasing for any $h' \le h$
    
            \end{itemize}
        \end{proposition}
        \begin{proof}
            We start proving point (ii), which implies, by the Coarser Composition Rule (\Cref{prop:simple_properties}) point (i). Consider any point $v \in [0,1]^2$. We have three cases:
        \begin{itemize}
            \item If $v$ does not belong to the allocation region of $\app{h}{M}$, then it does not belong to that of $M$ either (by point (i) of \Cref{prop:simplified_mechanisms} we know that $A_M \subseteq A_{\app{h}{M}}$). Therefore $\prof(M,v)=\prof(\app{h}{M},v)=0.$
            \item If $v$ belongs to the allocation region of $M$, then it also belongs to that of $\app{h}{M}$ (once again, by point (i) of \Cref{prop:simplified_mechanisms}), moreover, by point (ii) of \Cref{prop:simplified_mechanisms}, we have that: 
            \[
                \prof(M,v) - \prof(\app{h}{M},v) \le 2\cdot 2^{-h}. 
            \]
            \item Consider now the situation in which $v$ belongs to the allocation region of $\app{h}{M}$, but not of $M$. By construction, $v$ belongs to a rectangle with sides of length at most $2^{-h}$, such that the right and lower sides of the rectangle belong to $\partial \app{h}{M}$. Therefore, under $\app{h}{M}$, the price paid to the seller is at most $v_1+2^{-h}$, and that paid by the buyer is at least $v_2-2^{-h}$, for a revenue of at least $v_2-v_1-2 \cdot 2^{-h}$. If $v \in U$, the upper-left corner of $[0, 1]^2$, then $v_2-v_1 \ge 0$, therefore $\prof(\app{h}{M}, v) \ge -2\cdot 2^{-h}$, which verifies the inequality since $\prof(M, v) = 0$ in this case. If $v$ is not in $U$, then it cannot be too far from $U$, as it is contained in the allocation region of $\app{h}{M}$, and the allocation region of $M$ is contained in $U$. In particular, since $M\in \M^+$, the lower-left and upper-right vertexes of the rectangle described above need to belong to $U$, since by construction the union of the lower side and the right side connects two points in $A_M\subseteq U$. This means that the lower-right vertex $u$, whose coordinates give the Myerson's prices $q = u_2$ and $p = u_1$ for $v$, is such that $q-p$ is at most $-2^{-h}$. This follows by considering the vertical and horizontal distance between $u$ and the main diagonal characterizing $U$.
            \end{itemize}

            Consider now the points (iii) and (iv). We know that the set sequence $A_{\app{h}{M}}$ is monotonically decreasing with respect to the inclusion (\Cref{prop:simple_properties}), therefore $A_{\app{h'}{M}} \subseteq A_{\app{h}{M}}$ for any $h'\ge h$. This means that $v \notin A_{\app{h}{M}}$ implies $v \notin A_{\app{h'}{M}}$, and thus $\prof(A_{\app{h'}{M}},v) =0$. Conversely, if $v \in A_{\app{h'}{M}}$, then it also belongs to $A_{\app{h}{M}}$ and, by the monotonicity of the profit function (\Cref{prop:trade-off}), we have that $\prof(\app{h'}{M},v) \ge \prof(\app{h}{M},v).$
        \end{proof}
    
        When the valuation pair is sampled from some distribution $\D$, we may use $\E{\prof(M)}$ instead of $\Esub{V \sim \D}{\prof(M,V)}$ when there is no ambiguity on the distribution we are referring to. We have the following, crucial structural lemma on approximating sequences. 
    
        \begin{proposition}[Local Dichotomy of Approximating Sequences]
        \label{prop:local}
            Let $M$ be any mechanism in $\M^+$, and $h$ be any integer. Then for any distribution at least one of the following two properties holds:
            \begin{itemize}
                \item[\text{(i)}] $\E{\prof(\app{h}{M}) - \prof(\app{h+1}{M})} \ge 31 \cdot 2^{-h}$
                \item[\text{(ii)}] $\E{(\prof(\app{h+1}{M}) - \prof(\app{h}{M}))^2} \le 136 \cdot 2^{-h}$
            \end{itemize}
        \end{proposition}
        \begin{proof}
            We partition the $[0,1]^2$ square according to the revenue that the various points induce. In particular, we denote with $B$ the points that induce significant different revenues in the two mechanisms:
            \[
                B=\left\{v \in [0,1]^2 \mid \prof(\app{h}{M},v) - \prof(\app{h+1}{M},v) > \tfrac{2}{2^h}\right\}
            \]
            Moreover, for $i=1, \cdots, h-1$ we define:
            \[
                B_i=\left\{v \in B \mid \tfrac{2^{i}}{2^h} < {\prof(\app{h}{M},v)} \le \tfrac{2^{i+1}}{2^h}\right\}.
            \]
            \begin{claim}
            \label{claim:partition_b}
                The sets $B_i$ are a partition of $B$. Moreover for $v \in B$ the following two formulae hold: $\prof(\app{h+1}{M},v) = 0$ and $\prof(\app{h}{M},v)>\nicefrac{2}{2^h}$.
            \end{claim}
            \begin{proof}[Proof of the Claim]
                 The revenue function is at most $1$. Therefore, to prove that $\{B_i\}_{i=1}^{h-1}$ is a partition of $B$, we only need to prove that no $v \in B$ is such that $\prof(\app{h}{M},v)\le \nicefrac{2}{2^h}$. Clearly, if $v$ is in $B$, then it cannot belong to the allocation region of both mechanisms (otherwise \Cref{prop:approximating_seq} would imply that $\prof(\app{h+1}{M},v)\ge \prof(\app{h}{M},v))$. Similarly, if $v$ does not belong to the allocation region of $\app{h}{M}$, then it induces zero revenue for both $\app{h}{M}$ and $\app{h+1}{M}$, which is incompatible with being in $B$. The only possibility is that $v$ belongs to the allocation region of $\app{h}{M}$ but not of $\app{h+1}{M}$. This implies that $\prof(\app{h+1}{M},v) = 0$ and $\prof(\app{h}{M},v)>\nicefrac{2}{2^h}$ by definition of $B$. 
            \end{proof}
            If we focus on the fraction of revenue generated by the points in $B$, we get the following:
            \begin{align}
            \nonumber
                \Esub{V \sim \D}{\prof(\app{h}{M},V) \cdot \ind{V \in B}} &= \sum_{i=1}^{h} \Esub{V \sim \D}{\prof(\app{h}{M},V) \cdot \ind{V \in B_i}} \tag{By \Cref{claim:partition_b}}\\
                &\ge \sum_{i=1}^{h}\tfrac{2^{i}}{2^h} \cdot\Psub{V \sim \D}{V \in B_i} \tag{By def. of $B_i$}\\
            \label{eq:def_hashtag}
                &\ge \underbrace{\sum_{i=1}^{h}\tfrac{2^{2i}}{2^{2h}} \cdot \Psub{V \sim \D}{V \in B_i}}_{(\#)},
            \end{align}
            where, in the last inequality we used that $h \ge i$, and thus $2^{i-h} \ge 2^{2(i-h)}$.
            We now divide the analysis into two cases, according to whether the right-hand side of \Cref{eq:def_hashtag} (let's denote it with $(\#)$) is larger or smaller than $33 \cdot 2^{-h}$.
    
            \paragraph{First Case.} If $(\#) \ge 33 \cdot 2^{-h}$, then from \Cref{eq:def_hashtag} we can immediately conclude that $\E{\prof(\app{h}{M},V) \cdot \ind{V \in B}} \ge 33 \cdot 2^{-h}$, and we can write the following:
            \begin{align*}
                \E{\prof(\app{h}{M})} &= \E{\prof(\app{h}{M},V)\cdot \ind{V \in B}}+ \E{\prof(\app{h}{M},V)\cdot \ind{V \notin B}}\\
                &\ge 33 \cdot 2^{-h} + \E{\prof(\app{h}{M},V)\cdot \ind{V \notin B}} \tag{By bound on $(\#)$}\\
                &\ge 33 \cdot 2^{-h} + \E{(\prof(\app{h+1}{M},V) - \tfrac{2}{2^h})\cdot \ind{V \notin B}}\\
                &\ge \E{\prof(\app{h+1}{M})} + 31 \cdot 2^{-h}.
            \end{align*}
            Note, in the second to last inequality we applied point (i) of \Cref{prop:approximating_seq}, while in the last passage we used the observation that $\prof(\app{h+1}{M},v) = 0$ for $v \in B.$
    
            \paragraph{Second case.} If $(\#) < 33 \cdot 2^{-h}$, then we can look at the second moment of the difference:
            \begin{align*}
                &\E{(\prof(\app{h+1}{M}) - \prof(\app{h}{M})^2}\\
                &= \E{\prof(\app{h}{M},V)^2 \ind{V \in B} } + \E{(\prof(\app{h+1}{M}) - \prof(\app{h}{M})^2\ind{V \notin B} }\\
                &\le \E{\prof(\app{h}{M},V)^2 \ind{V \in B} } + \tfrac{4}{2^{2h}}\P{V \notin B}\\
                &\le 4 \cdot \sum_{i=1}^{h} \P{V \in B_i}  \tfrac{2^{2i}}{2^{2h}} + \tfrac{4}{2^{2h}} \le 136 \cdot 2^{-h}. \tag{By the case analysis and \Cref{claim:partition_b}}
            \end{align*}
            Note, in the first equality we use once again that $\prof(\app{h+1}{M},v) = 0$ for $v \in B$. In the first inequality we use that for $v \notin B$ it holds that $\prof(\app{h}{M},v) - \prof(\app{h+1}{M},v) \le 2\cdot 2^{-h}$ which, combined with Property (i) of \Cref{prop:approximating_seq} implies that $|\prof(\app{h}{M},v) - \prof(\app{h+1}{M},v)| \le 2\cdot 2^{-h}.$
        \end{proof}

    \subsection{Target Net and Regret Bound}
        The approximating sequences and their properties play a crucial role in our results, as they allow us to control the learning guarantees of the simplified mechanism $M^t$ from \Cref{subsec:simplify} that is played by the algorithm \simplify. In the rest of the section we forget about the online nature of the problem at hand and study the associated \emph{sample-complexity} problem, then concluding our argument via the standard offline-to-online reduction for online learning with full feedback. Fix then a distribution $\D$ and consider a sample (multi-)set $S$ of $n$ points drawn i.i.d. from $\D$. For any set $I$ of indices in $[n] = \{1,2,\dots, n\}$, we denote with $S(I)$ the projection of $S$ along the indices in $I$, i.e., all the samples corresponding to the indices in $I$. Note that $S(I)$ could be a multi-set. Since we consider sets of indices of cardinality $\nicefrac{4}{\eps}$, we make the assumption that $n\ge \nicefrac 4\eps$.
        
        \begin{definition}[Data-Dependent Net]
        \label{def:data_dependent_net}
            Fix any precision $\varepsilon = 2^{-H}$ for some $H \in \mathbb N$. Consider $n$ i.i.d. samples $S$ from $\D$, and let $I$ be any subset of $\nicefrac{4}{\eps}$ indices in $[n]$. We denote with $G^{\eps}_{S(I)}$ the grid induced by the points $S(I)$ and by the uniform $\eps$-grid. Furthermore, we denote with $\M^{\eps}_{S(I)}$ the set of all mechanisms induced by complete paths on $G^{\eps}_{S(I)}$.
        \end{definition}

        Recall, the grid is defined as in \Cref{def:grid}, while for complete paths we refer to \Cref{def:complete_path}. We complement this definition by providing an upper bound on the cardinality of $\cup_I \M^{\eps}_{S(I)}$, when the set of indices $I$ span over all the subset of $\nicefrac{4}{\eps}$ indices in $[n]$.
        
        \begin{proposition}
        \label{prop:cardinality_net}
            Let $n > 16$, then the following holds: 
            \[
                 \log \Big(\Bigl\lvert\bigcup_{\substack{
                 I \subseteq [n]\\|I| = \nicefrac{4}{\eps}}} \M^{\eps}_{S(I)}\Bigr\rvert \Big) \le  \frac{10}{\eps} \log n.     
            \] 
        \end{proposition}
    \begin{proof}
        There are $\binom{n}{\nicefrac{4}{\eps}}$ ways of choosing the indices $I$, and, for each such choice, each grid $G^{\eps}_S(I)$ is induced by $\nicefrac {5} {\eps}$ points. We have then the following bound (by \Cref{lem:n_of_paths}):
        \[
            \Bigl\lvert\bigcup_{\substack{
                 I \subseteq [n]\\|I| = \nicefrac{4}{\eps}}} \M^{\eps}_{S(I)}\Bigr\rvert \le (2e)^{\nicefrac{10}{\eps}}\binom{n}{\nicefrac{4}{\eps}} \le 4^{\nicefrac{1}{\eps}}e^{\nicefrac{14}{\eps}}n^{\nicefrac{4}{\eps}} \le n^{\nicefrac {10}{\eps}},
        \]
        where we use the inequality $\binom{n}{k} \le (e\cdot \nicefrac{n}{k})^k,$ and the assumption that $n > 16$. Applying the logarithm yields the desired bound. 
    \end{proof}

        Consider the optimal mechanism $M^\star$ on the sample $S$, as constructed in \Cref{subsec:maximizer} and \Cref{cor:empirical_maximizer}; it undergoes the simplification procedure presented in \Cref{subsec:simplify} to generate the simplified mechanism $\app{H}{M^\star}$ (which is the mechanism actually played by the algorithm). We want to argue that, with high-probability, the profit generated by $\app{H}{M^\star}$ \emph{on the sample $S$} is a good approximation of its expected profit with respect to a new sample from $\D$. The mechanism $\app{H}{M^\star}$ is induced by a complete path on the point induced grid $G_{M^\star}^\eps$, where $\eps = 2^{-H}$; moreover, it is characterized by at most $\nicefrac{4}{\eps}$ points in $G_S$, and this clearly also works as an upper bound on the number of characterizing elements from $S$. This means that, for some set of indices $I$ of cardinality at most $\nicefrac{4}{\eps}$, the boundary of $M^\star$ is a complete path in $G_{S(I)}^{\eps}$. We thus have the following lemma. 

    \begin{lemma}
    \label{lem:few_corners}
        Fix any precision $\eps = 2^{-H}$ for some $H \in \mathbb N$, and consider $n$ i.i.d. samples $S$ from distribution $\D$. Then the simplified mechanism $\app{H}{M^\star}$ belongs to $\M_{S(I)}^{\eps}$ for some set $I$ of at most $\nicefrac{4}{\eps}$ indices.
    \end{lemma}
    \begin{proof}
         By construction from \Cref{subsec:simplify}, the mechanism $\app{H}{M^\star}$ is characterized by $\partial M^\star \cap G_\eps$. Additionally, again by construction, this time from \Cref{cor:empirical_maximizer}, we have that each element of the intersection depends on two points of $G_S$, since the boundary of $M^\star$ is the union of segments from $G_S$, with each element of the union thus connecting two points in the sample grid. Therefore, proving the statement is equivalent to stating that the amount of intersecting elements between boundary and $\eps$-grid is bounded by $\nicefrac{2}{\eps}$.
       
         To prove this, notice that a complete path must go from the segment $(0, 0)$-$(0, 1)$ to the segment $(0, 1)$-$(1, 1)$; similarly to the proof of \Cref{lem:n_of_paths}, by monotonicity of the region induced by a complete path, its boundary can only enter up to $\nicefrac{2}{\eps}$ cells in the $G_\eps$ grid ($\nicefrac{1}{\eps}$ ``right" movements and $\nicefrac{1}{\eps}$ ``up" movements), with each cell visited at most once. Also, we can add an element to the intersection if and only if the path enters one of the cells by its lower or left sides, thus meaning that the cardinality of the intersection is bounded by the amount of cells that can be possibly explored, i.e. $\nicefrac{2}{\eps}$.
    \end{proof}

    Unfortunately, it is unclear whether a fast enough uniform learning bound on the whole class $\cup_I\M_{S(I)}^{\eps}$ is achievable. In particular, a naive union bound (using \Cref{prop:cardinality_net}) would imply a suboptimal sample complexity of order $\approx \nicefrac{1}{\eps^3}.$  To further restrict the class of mechanism we want to learn, we introduce the family of candidate mechanisms.

        \begin{definition}[Candidate Mechanisms]
        \label{def:candidates}
            Fix any precision $\varepsilon = 2^{-H}$ for some level $H \in \mathbb N$. We say that a mechanism $M \in \M^+$ is a candidate if the following condition holds for all $h = 0, 1, \dots, H-1$:
            \begin{equation}
            \label{eq:good_candidate}            \E{\prof(\app{h}{M}) - \prof(\app{h+1}{M})} < 4 \cdot 2^{-h}.
            \end{equation}
            We denote with $\C$ the set of all candidate mechanisms.
        \end{definition}

        The idea underlying this definition is that a mechanism whose approximating sequence does not uniformly respect \Cref{eq:good_candidate} is clearly suboptimal (with respect to distribution $\D$), and we do not need to learn it with high precision. We are now ready to present our target net. 
        
        \begin{definition}[Target Net]
        \label{def:target_net}
            For any distribution $\D$ and precision parameter $\varepsilon = 2^{-H}$, consider $n$ i.i.d. samples $S$ from $\D$. We construct a target net $\M^\star(S)$ as follows:
            \[
                \M^\star(S) =  \bigcup_{\substack{
                 I \subseteq [n]\\|I| = \nicefrac{4}{\eps}}} {\M}^{\eps}_{S(I)}\cap \C.
            \]
        \end{definition}

        In English, the target net is composed by all the candidate mechanisms in that correspond to some complete path in $G_{S(I)}^{\eps}$, for some $I$ of cardinality $\nicefrac{4}{\eps}$.

        \begin{observation}[The Structure of Our Net]
        \label{obs:structure_net}
            Before moving on, we observe that our target net $\M^\star(S)$ is both data and distribution dependent! While the family of candidate mechanisms $\C$ is distribution dependent, the families $\M^{\eps}_{S(I)}$ are indeed data-dependent, as they are induced by the realized samples. In particular, each mechanism in our net only depends on a small number of samples (i.e., $\nicefrac {4}{\eps}$), so that the remaining samples are enough to attain concentration.
        \end{observation}

    We now want to argue that the mechanism computed by the algorithm, i.e. $\app{H}{M^\star}$, where $M^\star$ is the empirical risk minimizer, belongs to the family $\M^\star(S)$ where we want to prove uniform learnability. 
    \begin{theorem}
    \label{thm:target}
        Fix a precision $\eps = 2^{-H}$ and a probability parameter $\delta$. Consider $n$ i.i.d. samples $S$, then $\app{H}{M^\star} \in \M^\star(S)$ with probability at least $(1-\delta)$, as long as $n \ge \tfrac{25^2}{\eps^2} \cdot \left(\log \left(\nicefrac{20}{\delta}\right) + \log \left(\nicefrac{20}{\eps} \right) \right)$.
    \end{theorem}
    \begin{proof}
        Let $M^\star \in \M^+$ be the empirical risk minimizer on the samples $S$ computed as in \Cref{cor:empirical_maximizer}, i.e. the mechanism supported in $G_S$ such that:
        \begin{equation}
            \label{eq:risk_minimizer}
                \sum_{v \in S}\prof(M^\star,v) \ge  \sum_{v \in S} \prof(M,v), \quad \forall M \in \M.
        \end{equation}
        We have already argued in \Cref{lem:few_corners} that $\app{H}{M^\star} \in \M_{S(I^\star)}^{\eps}$ for some set of indices $I^\star$ of cardinality $\nicefrac{4}{\eps}$, so we are left with proving that, with high probability, $\app{H}{M^\star}$ also belongs to $\C$. 
        
        If we focus on the samples in $S\setminus S(I^\star)$, then we get a weaker version of \Cref{eq:risk_minimizer}, namely\footnote{We introduce the following notation: $\prof(M,T)$, where $T$ is a set of points, denote the average profit of mechanism $M$ on the points in $T$.}:
        \begin{equation}
        \label{eq:normalized_risk_minimizer}
            \prof(M^\star,S\setminus S(I^\star)) \ge \prof(M,S\setminus S(I^\star)) - \tfrac{\eps}2 , \quad \forall M \in \M
        \end{equation}
        To justify this, we have the following chain of inequalities, for all $M \in \M$:
        \begin{align*}
            \sum_{p \in S\setminus S(I^\star)} \prof(M^\star, p) + \tfrac{4}{\eps} &\ge \sum_{p \in S}\prof(M^\star,p) \tag{$\prof \le 1$ and $|I^\star| = \nicefrac{4}{\eps}$}\\
            &\ge \sum_{p \in S} \prof(M,p) \tag{By \Cref{eq:risk_minimizer}}\\
            &\ge \sum_{p \in S\setminus S(I^\star)} \prof(M,p) - \tfrac{4}{\eps} \tag{$\prof \ge -1$ and $|I^\star| = \nicefrac{4}{\eps}$}.
        \end{align*}
        By assumption on the sample cardinality, we have that $n \ge \nicefrac{32}{\eps^2}$, which implies that the cardinality of the multi-set $S \setminus S(I)$ is at least $\nicefrac n2.$ Dividing both terms of the above chain of inequalities by $|S \setminus S(I)|$, we get \Cref{eq:normalized_risk_minimizer}.

        As a second step, we argue, that for all levels $h = 0, \dots, H-1$ the following inequality needs to hold:
        \begin{equation}
        \label{eq:property}
            \prof(\phi_h(M^\star),S\setminus S(I^\star)) - \prof(\phi_{h+1}(M^\star),S\setminus S(I^\star)) < 2 \cdot 2^{-h}.
        \end{equation}
        To see this, assume towards contradiction that there exists $h$ such that \Cref{eq:property} is not verified, then we would have the following:
        \begin{align*}
            \prof(\phi_h(M^\star),S\setminus S(I^\star)) &\ge \nicefrac{2}{2^{h}} + \prof(\phi_{h+1}(M^\star),S\setminus S(I^\star))\\
            &\ge \nicefrac{1}{2^{h}} + \prof(M^\star,S\setminus S(I^\star)) \tag{Point (ii) of Prop.~\ref{prop:approximating_seq}}\\
            &\ge \nicefrac{1}{2^{h}} - \nicefrac{\eps}2 + \prof(\phi_h(M^\star),S\setminus S(I^\star)) \tag{By Eq.~\ref{eq:normalized_risk_minimizer}}.
        \end{align*}
        The last inequality yields a contradiction, as $\nicefrac{1}{2^h} \ge \nicefrac{1}{2^H} = \eps > \nicefrac{\eps}2.$

        The last step in this proof consists in proving that, with high probability, all mechanisms in $\M^\eps_{S(I)}$ that exhibit the inequality in \Cref{eq:property} at all levels $h$ and whose allocation regions are included in $U$, the upper-left corner, belong to $\M^\star(S).$ 
        
        Consider any mechanism $M \in \left(\M^\eps_{S(I)} \cap \M^+\right) \setminus \C$, for some set of $\nicefrac {4} {\eps}$ indices $I$. If $M$ is not a candidate, it means that there exists a level $h$ such that $\E{\prof(\app{h}{M}) - \prof(\app{h+1}{M})} \ge 4 \cdot 2^{-h}.$ We want to apply concentration on $X$, the sum for $v \in S \setminus S(I)$ of the random variables of the form:  
        \[
            X_v = \tfrac 14 [\prof(\app{h}{M},v) - \prof(\app{h+1}{M},v) + 2\cdot 2^{-h}].
        \]
        The $X_v$ lie in $[0,1]$ by property (i) of \Cref{prop:approximating_seq}
        and because the revenue is at most $1$.
        Moreover, we are assuming that $\E{X_v}\ge \nicefrac 32 \cdot 2^{-h}.$ We also introduce the event $\cE_{M,h}$ that is realized when $X$ is way smaller than expected:
        \[
            \cE_{M,h} = \left\{\prof(\app{h}{M},S\setminus S(I)) - \prof(\app{h+1}{M},S\setminus S(I)) < 2 \cdot 2^{-h}\right\}.
        \]
        We have the following, where $n' = |S\setminus S(I)|$:
        \begin{align}
        \nonumber
            \P{\cE_{M,h}} &= \P{X < 2^{-h} \cdot n'}\\
        \nonumber
            &= \P{X - \E{X} < 2^{-h} \cdot n' - \E{X}}\\
        \nonumber
            &\le \P{X - \E{X} < - \nicefrac{1}{3} \cdot \E{X}}\\
            &\le e^{-\tfrac{1}{18} \E{X}} \tag{By (i) of \Cref{thm:chernoff}}\\
        \label{eq:concentration_bad}
            &\le e^{-\tfrac{\eps}{24}n}
        \end{align}
        where in the first and last inequalities we used monotonicity of probability and that $\E{X} \ge n' \cdot \nicefrac 32 \cdot {2^{-h}} \ge \nicefrac{3}{4}\cdot \eps n$, where this last inequality comes again from the observation that $n' > \nicefrac{1}{2}\cdot n$ under our assumption on $n$.
        
        We now need to union bound over all possible bad mechanisms in $\M_{S(I)}^{\eps}$ and levels $h$ for which $\E{\prof(\app{h}{M}) - \prof(\app{h+1}{M})} \ge 4 \cdot 2^{-h}$: thus, by \Cref{{prop:cardinality_net}} (and by noting that, since $\eps = 2^{-H}$, there are at most $2 \log \nicefrac{1}{\eps}$ levels) 
        we have that the probability that $\app{H}{M^\star}$ belongs to $\M^\star(S)$ is at least:
        \[
            2 e^{-\tfrac{\eps}{24}n} \cdot e^{\nicefrac{10}{\eps}\log n} \cdot \log \tfrac{1}{\eps} \le \delta,
        \]
        by how we set the number of samples $n.$
    \end{proof}

    The next step consists in arguing that it is possible to learn the mechanisms in the target net $\M^\star(S)$ fast enough in expectation. The proof of this statement is postponed until \Cref{subsec:target_net}.
    
    \begin{restatable}[Uniform Learning in $\M^\star(S)$]{theorem}{expectationconcentration}
    \label{thm:expectation_concentration}
        Fix any precision $\eps = 2^{-H}$ for some $H \in \mathbb N$. Let $S$ be an i.i.d. sample of size $n \ge \tfrac{32}{\eps^2} \log^3(\nicefrac{1}{\eps})$. The following inequality holds:
    \[
        \Esub{S}{\max_{M \in \M^\star(S)} \Bigl\lvert\prof(M,S) - \E{\prof(M)}\Bigr\rvert} \le 794 \eps.
    \]
    \end{restatable}

    We have now all the ingredients to prove our theorem: the algorithm \simplify achieves a nearly tight $\tilde{O}(\sqrt{T})$ regret. 

    \begin{theorem}[Regret Bound]
    \label{thm:regret-bound}
        Consider the profit maximization problem in bilateral trade against the stochastic i.i.d. adversary. The regret of $\simplify$ is at most $C \sqrt{T \log^3 T}$, for some constant $C$. 
    \end{theorem}
    \begin{proof}
        We first study the instantaneous regret of the generic time step $t+1$. The algorithm has access to the first $t$ samples $S^t$ received, computes the profit-maximizing mechanism $M^\star_{t+1}\in \M^+$ on $S^t$ and posts $M^{t+1}$, the version of $M^\star_{t+1}$ adapted to the grid with precision $\eps_{t+1}.$ We denote with $\cE_1$ the event that $M^{t+1}$ belongs to $\M^\star(S^t)$. Furthermore, let $\tilde M$ be any mechanism that verifies
        \begin{equation}
        \label{eq:sup} 
            \sup_{M \in \M} \E{\prof(M)} \le \E{\prof(\tilde M)} + \eps_{t+1},
        \end{equation}
        such mechanism exists by definition of $\sup$. Denote with $\cE_2$ the event that $\E{\prof(\tilde M)}$ is $\eps_{t+1}$ approximated on the $t$ samples. We introduce the clean event $\cE=\cE_1 \cap \cE_2$. By union bound, our choice
        of $\eps_{t+1} = {200 \sqrt{\nicefrac{\log^3 T}{t}}}$, standard Hoeffding concentrations (for $\cE_2$) and \Cref{thm:target} it holds that $\P{\cE} \ge 1- \nicefrac{2}{T}$. 
        
        Conditioning on $\cE$ on the samples $S^t$, we have the following chain of inequalities:
        \begin{align}
            \sup_{M \in \M}\E{\prof(M)} &\le \E{\frac{1}{t}\sum_{v \in S^t} \prof(\tilde M,v) \Bigr | \, \cE} + 2\eps_{t+1} \tag{By event $\cE_2$ and \Cref{eq:sup}}\\
            &\le \E{\frac{1}{t}\sum_{v \in S^t} \prof(M_{t+1}^\star,v) \Bigr | \, \cE} + 2\eps_{t+1} \tag{$M_{t+1}^\star$ empirical maximizer on $S_t$}\\
            &\le \E{\frac{1}{t}\sum_{v \in S_t} \prof(M^{t+1},v)\Bigr | \, \cE} + 4\eps_{t+1} \tag{By Property (ii) of \Cref{prop:approximating_seq}} \\
            &\le \E{\prof(M^{t+1})\Bigr | \, \cE} + 1590\eps_{t+1} \label{eq:t+1}.
        \end{align}
        Note, the expectation in the left-hand-side term is with respect to a fresh sample, while that on the right-hand-side terms is with respect to the first $t$ samples observed, conditioning on the clean event. Moreover, the last step of \Cref{eq:t+1} descends from \Cref{thm:expectation_concentration} noting that $\P{\cE} \ge \nicefrac 12$:
        \begin{align*}
            \E{\Bigl\lvert\prof(M^{t+1},S_t) \!-\! \E{\prof(M^{t+1})}\Bigr\rvert\! \cdot\! \mathbbm{1}_{\cE}\,\!}\!\!\le\, \! \E{\max_{M \in \M^\star(S_t)} \Bigl\lvert\prof(M,S_t)\! - \!\E{\prof(M)}\Bigr\rvert} \!\!\le \!794\eps_{t+1}
        \end{align*}
        Note, to apply the concentration result for $\M^{\star}(S)$ of the actual mechanism $M^{t+1}$ we need to condition on the high-probability event that $M^{t+1}$ does belong to $\M^{\star}(S)$. 
        
        \Cref{eq:t+1} gives us the crucial ingredient to provide a bound on the expected instantaneous regret.        By the law of total probability, and by the fact that $\cE^C$ has probability at most $\nicefrac{2}{T}$ we get:
        \[
            \sup_{M \in \M}\E{\prof(M)} -\E{\prof(M^{t+1})} \le 1590\eps_{t+1} + \tfrac 2T 
        \]
        Summing up the instantaneous regret for $t=0, \dots, T-1$ yields the desired regret bound. 
    \end{proof}

    \subsection{The Sample Complexity of Good Mechanisms} 
    \label{subsec:target_net}

    This section is devoted to proving \Cref{thm:expectation_concentration}, that we restate here. 

    \expectationconcentration*

    We decompose the proof of the theorem in many lemmata. {For the coupling argument of this section, let $S'$ be a (multi-)set of $n$ new i.i.d. samples from $\D$, independent of $S$. Also, let $I_M \subset [n]$ be a set of indexes satisfying $M \in \M^\eps_{S(I_M)}$. We have the following lemma.}
    
    \begin{lemma}[Adding the Ghost Samples $S'$]
    \label{lem:ghost}
        The following inequality holds:
        \begin{align*}
            &\Esub{S}{\max_{M \in  \M^{\star}(S)} \Bigl\lvert\prof(M,S) - \E{\prof(M)}\Bigr\rvert} \\
            &\le \Esub{S,S'}{\max_{M \in  \M^{\star}(S)} \Bigl\lvert\prof(M,S\setminus S(I_M)) - \prof(M,S'\setminus S'(I_M))\Bigr\rvert} + \eps
        \end{align*}
    \end{lemma}
    \begin{proof}
         {Let $M \in \M^\star(S)$. We start by rewriting $\E{\prof(M)}$ as the expectation of $\prof(M, S')$ over the distribution of $S'$}. Fix any realization of the samples $S$, we have the following:
        \begin{align*}
            |&\prof(M,S) - \Esub{S'}{\prof(M,S')}|\\ &\le \frac{1}{|S|} \left[\sum_{i \in I_M} |\prof(M,p_i) - \Esub{p_i'}{\prof(M,p'_i)}| + \Esub{S'}{\Bigl\lvert\sum_{i \notin I} \left(\prof(M,p_i) - \prof(M,p'_i) \right)\Bigr\rvert}\right] \\
            &\le \frac{{2}|I_M|}{|S|} + \frac{|S| - |S(I_M)|}{|S|}\Esub{S'}{\Bigl\lvert\prof(M,S\setminus S(I_M)) - \prof(M,S'\setminus S'(I_M))\Bigr\rvert}\\
            &\le \frac{{8}}{\eps n} +  \Esub{S'}{\Bigl\lvert\prof(M,S\setminus S(I_M)) - \prof(M,S'\setminus S'(I_M))\Bigr\rvert},
        \end{align*}
        where the second inequality holds because the revenue function's image is contained in $[-1,1]$. {The third one follows from \Cref{def:target_net}: every mechanims $M$ in $\M^\star(S)$ is characterized by a set of indices $I_M$ of size $\nicefrac{4}{\eps}$.} So far we have considered any mechanism in $\M^\star(S)$, we can thus apply the $\max$:
        \begin{align*}
            \max_{M \in \M^\star(S)}|&\prof(M,S) - \Esub{S'}{\prof(M,S')}|\\ 
           &\le \frac{{8}}{\eps n} +  \max_{M \in \M^\star(S)}\Esub{S'}{\Bigl\lvert\prof(M,S\setminus S(I)) - \prof(M,S'\setminus S'(I))\Bigr\rvert}\\
           &\le \frac{{8}}{\eps n} +  \Esub{S'}{\max_{M \in \M^\star(S)}\Bigl\lvert\prof(M,S\setminus S(I)) - \prof(M,S'\setminus S'(I))\Bigr\rvert}
        \end{align*}
        {The last inequality is Jensen's}. The result then follows by taking the expectation with respect to $S$, and by the assumption on $n$, {which implies $n\ge \nicefrac{8}{\eps^2}$}.
    \end{proof}

    The next step consists in symmetrizing the random variables we are trying to learn. In particular, for this step to work, we need that the random variables we are taking the max over to be independent of the family of functions that are contained in the max. Stated differently, 
    the family $\M^\star(S)$ depends on the samples $S(I)$, therefore, we cannot use those samples in the symmetrization argument! Moreover, to maintain symmetry, we cannot restrict our attention to the sole mechanisms generated by the samples $S$. We introduce the symmetric family $\M^\star(S,S')$ defined as follows:
    \[
        \M^\star(S,S') = \bigcup_{\substack{
                 I \subseteq [n]\\|I| = \nicefrac{4}{\eps}}}\bigcup_{J\subseteq I} \M^{\eps}_{S(I\setminus J) \cup S'(J)}.
    \]
    Stated differently, a mechanism in $\M^\star(S,S')$ is supported on the standard $\eps$-grid, augmented with some points with indexes $J$ of $S'$ and some points with indexes $I \setminus J$ of $S$. {For a mechanism $M \in \M^\star(S,S')$ we define $I_M$ in the same way as before}. We also adopt the convention that $n'$ denotes $n - \nicefrac{4}{\eps}.$

    \begin{lemma}[Symmetrization]
    \label{lem:symmetrization}
    The following inequality holds:
    \begin{align*}
        &\Esub{S,S'}{\max_{M \in \M^\star(S)} \Bigl\lvert\prof(M,S\setminus S(I_M)) - \prof(M,S'\setminus S'(I_M))\Bigr\rvert}\\
        &\le 2\, \Esub{S,S',\textbf{r}}{\max_{M \in \M^\star(S,S')} \Bigl\lvert \frac{1}{n'}\sum_{p \in S \setminus S(I_M)}\prof(M,p) \cdot r_p\Bigr\rvert},
    \end{align*}
    where $\textbf{r}$ is a $n$-dimensional Rademacher vector. 
    \end{lemma}
    \begin{proof}
        We introduce the vector $\textbf{r}$ of $n$ Rademacher random variables, indexed by the points $p$ (or, alternatively, by the indices in $[n]$). Let $I$ be a set of indexes of cardinality $\nicefrac{4}{\eps}$ and $M \in \M_{S(I)}^{\eps}$; therefore, $I_M = I$.  For every point $p_i \in S \setminus S(I)$ and $p_i' \in S'\setminus S'(I)$, the following equality holds in distribution\footnote{This follows by the fact that $p_i$ and $p_i'$ are independent and have the same distribution; therefore, applying the same measurable function on both and taking the difference renders a random variable with law which is symmetrical about zero.}: 
        \[
            \prof(M,p_i) - \prof(M,p_i') \eqd (\prof(M,p_i) - \prof(M,p_i')) \cdot r_i.
        \]
        Similarly, if we now take any mechanism $M \in \M_{S(I)}^{\eps}$, the following equality holds in distribution: 
        \[
            \Bigl\lvert\sum_{i \notin I}(\prof(M,p_i) - \prof(M,p_i')) \Bigr\rvert \eqd \Bigl\lvert\sum_{i \notin I}(\prof(M,p_i) - \prof(M,p_i')) \cdot r_i \Bigr\rvert.
        \]

        We now would like to argue that, if we take the max over all mechanisms in $\M^\star(S)$, then we maintain the equal distribution. However, this is not the case: the role of $S$ and $S'$ is not symmetric! The sample $p_i$ determines the structure of the mechanisms in $\M_{S(I)}^{\eps}$ for $I$ that contains $i$, \emph{but this is not the case for $p_i'$!} This is why we need to consider the wider family $\M^\star(S,S')$. Note, $\M^\star(S) \subset \M^\star(S, S')$, therefore: 
        \begin{align*}
            \max_{M \in \M^\star(S)} &\Bigl\lvert\prof(M,S\setminus S(I_M)) - \prof(M,S'\setminus S'(I_M))\Bigr\rvert \\
            &\le \max_{M \in \M^\star(S,S')} \Bigl\lvert\prof(M,S\setminus S(I_M)) - \prof(M,S'\setminus S'(I_M))\Bigr\rvert.
        \end{align*}
        We then focus on upper bounding the expected value (with respect to $S$ and $S'$) of the right-hand side term. To avoid the cumbersome dependence of the family of mechanisms involved in the $\max$ on the sampled points, we introduce the following, equivalent stochastic procedure to extract the two independent samples $S$ and $S'$: first, two independent and identical samples $T_1$ and $T_2$ are drawn, then for each index, an independent Bernoulli $X_i$ is drawn: if $X=1$, then $p_i$ is the $i^{th}$ sample in $T_1$ and $p'_i$ is the corresponding sample in $T_2$, otherwise, the converse happens. {This procedure induces the random split of the observations into $S$ and $S'$}.

        Fix any outcome of the samples $T_1$ and $T_2$, they deterministically induce the family of mechanisms $\M^\star(S,S') = \M^\star(T_1,T_2).$ So we can focus on coupling the impact of $X_i$ with that of the Rademacher vector $\mathbf{r}.$ For any mechanism $M \in \M^\star(S,S')$ and any index $i \notin I_M$, we have the following equality, in distribution:
        \[
            (\prof(M,p_i) - \prof(M,p_i'))  \eqd (\prof(M,p_i) - \prof(M,p_i')) \cdot r_i.
        \]
        Note: the randomness on the left-hand side is \emph{only} with respect to $X_i$, as the values of the $i^{th}$ samples from $T_1$ and $T_2$ are fixed\footnote{{In other words, the equality is in the stronger sense of \textit{conditional} distributions, when conditioning on $(T_1, T_2)$.}}; the right-hand side also depends on the Rademacher variable $r_i$. {Let $m$ be the maximum cardinality of $\M^\star(S, S')$ when $\lvert S \rvert = n$. Therefore, any vector indexed by $M \in \M^\star(S, S')$ and $i\notin I_M$ belongs to $[-1, 1]^{m \cdot (n-\nicefrac{4}{\eps})}$, adopting the convention that exceeding entries are null whenever the realization of $\M^\star(S, S')$ - which we are keeping fixed - has cardinality lower than $m$.} We want to prove that - for any fixed realization of $T_1$ and $T_2$, the following holds (for the joint distributions!):
        \begin{align*}
            &\left\{(\prof(M,p_i) - \prof(M,p_i'))\right\}_{M \in \M^\star(S, S'),\, i\notin I_M} \\  \eqd &\left\{(\prof(M,p_i) - \prof(M,p_i')) \cdot r_i\right\}_{M \in \M^\star(S, S'),\, i \notin I_M} 
        \end{align*}
        To see why this is indeed the case, consider any {Borel set $B \subseteq [-1, 1]^{m \cdot (n-\nicefrac{4}{\eps})}$}. We have the following: 
        \begin{align*}
            &\Psub{\mathbf{X},\mathbf{r}}{ (r_i \cdot (\prof(M,p_i) - \prof(M,p_i')))_{i,M} \in B} \\
            &= \frac{1}{2^n}\sum_{\mathbf{\overline r} \in \{-1,1\}^n}\Psub{\mathbf{X}}{(\overline{r_i} \cdot (\prof(M,p_i) - \prof(M,p_i')))_{i,M} \in B} \\
            &=\Psub{\mathbf{X}}{ (\prof(M,p_i) - \prof(M,p_i'))_{i,M} \in B},
        \end{align*}
        where the last equality follows by symmetry and independence of the random variables.

        Given the equal distribution, we can apply any measurable function\footnote{{The composition of the averaging over $i\notin I_M$ for every mechanism and the maximum over the resulting vector.}} to the two random vectors and still retain the same distribution. Therefore, for any fixed realizations of $T_1$ and $T_2$ we get:
        \begin{align*}
            \max_{M \in \M^\star(S,S')} &\Bigl\lvert\prof(M,S\setminus S(I_M)) - \prof(M,S'\setminus S'(I_M))\Bigr\rvert 
            \\
            &\eqd \max_{M \in \M^\star(S,S')} \Bigl\lvert \frac{1}{{n'}} \sum_{i \notin I_M}r_i \cdot (\prof(M,p_i) - \prof(M,p_i'))\Bigr\rvert.
        \end{align*}
            Note, once again, that the above equality holds for any fixed realizations of $T_1$ and $T_2$, while the randomness is with respect to the $\mathbf X$ and $\mathbf r$. If we now take the expectation with respect to the realizations of $T_1$ and $T_2$, by the law of total probability, we get: 
        \begin{align*}
            &\Esub{S,S'}{\max_{M \in \M^\star(S,S')} \Bigl\lvert\prof(M,S\setminus S(I_M)) - \prof(M,S'\setminus S'(I_M))\Bigr\rvert}
            \\
            &= \Esub{S,S',\mathbf{r}}{\max_{M \in \M^\star(S,S')} \Bigl\lvert \frac{1}{{n'}} \sum_{i \notin I_M}r_i \cdot (\prof(M,p_i) - \prof(M,p_i'))\Bigr\rvert}\\
            &\le 2 \Esub{S,S',\mathbf{r}}{\max_{M \in \M^\star(S,S')} \Bigl\lvert \frac{1}{{n'}} \sum_{i \notin I_M}r_i \cdot \prof(M,p_i)\Bigr\rvert},
        \end{align*}
        where the last inequality follows by triangle inequality and by the fact that $S$ and $S'$ are distributed independently and in the same way. 
    \end{proof}
    \begin{lemma}[From Rademacher to Gaussian]
    \label{lem:rad_to_gaus}
    The following inequality holds:
    \begin{align*}
        &\Esub{S,S',\textbf{r}}{\max_{M \in \M^\star(S,S')} \Bigl\lvert \frac{1}{n'}\sum_{p \in S \setminus S(I_M)}\prof(M,p) \cdot r_p\Bigr\rvert}\\
        &\le \sqrt{\frac \pi 2}  \Esub{S,S',\textbf{g}}{\max_{M \in \M^\star(S,S')} \Bigl\lvert \frac{1}{n'}\sum_{p \in S \setminus S(I_M)}\prof(M,p) \cdot g_p\Bigr\rvert}
    \end{align*}
    where $\textbf{r}$ is a $n$-dimensional Rademacher vector and $\textbf{g}$ is a $n$-dimensional standard gaussian vector. 
    \end{lemma}
    \begin{proof}
        Fix any realization of the samples $S'$ and $S$. They induce the set of mechanisms $\M^\star(S,S')$, and the corresponding revenues. For any mechanism $M \in \M^\star(S,S')$ and point $p \in S \setminus S(I_M)$, we have that the following two random variables share the same distribution:
        \begin{equation}
            \label{eq:point-wise_gaus}
            \prof(M,p) \cdot g_p \eqd \prof(M,p) \cdot |g_p| \cdot r_p,
        \end{equation}
        where $r_p$ are independent Rademacher random variables. \Cref{eq:point-wise_gaus} {jointly} holds for any point $p \in S \setminus S(I_M)$, and any mechanism in  $M \in \M^\star(S,S')$. 
        {This follows by an analogous reasoning to the one in the proof of \Cref{lem:symmetrization}, and like in that case we have: }
        \begin{align}
        \max_{M \in \M^\star(S,S')}\Bigl\lvert \frac{1}{n'}\sum_{p \in S \setminus S(I_M)}\prof(M,p) \cdot g_p\Bigr\rvert
        \eqd \max_{M \in \M^\star(S,S')}\Bigl\lvert \frac{1}{n'}\sum_{p \in S \setminus S(I_M)}\prof(M,p) \cdot r_p \cdot |g_p|\Bigr\rvert.
        \label{eq:sum-wise_gaus}
        \end{align}
        We now have the final chain of inequalities:
        \begin{align*}
            &\Esub{S,S',\textbf{g}}{\max_{M \in \M^\star(S,S')} \Bigl\lvert \frac{1}{n'}\sum_{p \in S \setminus S(I_M)}\prof(M,p) \cdot g_p\Bigr\rvert}\\
            \tag{\Cref{eq:sum-wise_gaus}}            &=\Esub{S,S',\textbf{g},\textbf{r}}{\max_{M \in \M^\star(S,S')} \Bigl\lvert \frac{1}{n'}\sum_{p \in S \setminus S(I)}\prof(M,p) \cdot r_p \cdot |g_p|\Bigr\rvert}\\ \tag{Jensen}            &\ge\Esub{S,S',\textbf{r}}{\max_{M \in \M^\star(S,S')} \Bigl\lvert \frac{1}{n'}\sum_{p \in S \setminus S(I)}\prof(M,p) \cdot r_p \cdot \Esub{\textbf{g}}{|g_p|}\Bigr\rvert}\\           &\ge\sqrt{\frac2{\pi}}\Esub{S,S',\textbf{r}}{\max_{M \in \M^\star(S,S')} \Bigl\lvert \frac{1}{n'}\sum_{p \in S \setminus S(I)}\prof(M,p) \cdot r_p\Bigr\rvert}.\tag{$\Esub{\textbf{g}}{|g_p|} =\sqrt{\nicefrac 2{\pi}}$}
        \end{align*}
        The statement of the lemma follows by multiplying both terms of the above inequality by $\sqrt{\nicefrac{\pi}{2}}$.            
    \end{proof}

    \begin{lemma}
    \label{lem:max}
    Let $g_1\ldots g_N$ be $N$ arbitrary Gaussian random variables, with $\mu = 0$ and $\sigma^2_i \le  \sigma^2$. Then
    \[
        \E{\max_{i \in [N]}|g_i|}\leq 2\sqrt{\sigma^2\log N}.
    \]
    \end{lemma}
    \begin{proof}
        This is a standard result, see e.g., Lemma~2.3.4. of \citet{Talagrand14}. We provide here a proof for completeness. For any positive $t$, we have the following chain of inequalities:
        \begin{align*}
            \exp\left(t \cdot \E{\max_{i \in [N]}|g_i|}\right) &\le \E{\exp\left(\max_{i \in [N]}t\cdot |g_i|\right)}
            \tag{By Jensen's Inequality}\\
            &=\E{\max_{i \in [N]}e^{t\cdot |g_i|}} \tag{{Monotonicity of the exponential}}\\
            &\le\E{\max_{i \in [N]}\left(e^{t\cdot g_i}+e^{-t\cdot g_i}\right)}\tag{Because $e^x\ge 0$ for all $x$}\\
            &\le \sum_{i=1}^N \E{e^{t\cdot g_i}  + e^{-t\cdot g_i}} \\
            &= 2\sum_{i=1}^N \E{e^{t\cdot g_i}} \tag{$g_i$ and $-g_i$ identically distributed}\\
            &\le 2 N e^{\frac{t^2\sigma^2}{2}}. \tag{Gaussian moment generating function}
        \end{align*}
        At this point, it is enough to apply the logarithm to both sides of the above inequality and upper bounding $\log(2N)$ with $\log N^2$:
        \[
            \E{\max_{i \in [N]}|g_i|} \le \frac{t\sigma^2}{2} + \frac{2}{t}\log(N) = 2\sqrt{\sigma^2 \log N},
        \]
        where the last equality holds by setting $t = 2 \sqrt{\nicefrac{\log N}{\sigma^2}}$
    \end{proof}

    \begin{lemma}
    \label{lem:telescoping}
        Consider any realization of the samples $S$ and $S'$. We have the following inequality: 
        \begin{align*}
            &\Esub{\textbf{g}}{\max_{M \in \M^{\star}(S,S')} \Bigl\lvert \frac{1}{n'}\sum_{p \in S \setminus S(I_M)}\prof(M,p) \cdot g_p\Bigr\rvert} - 38 \eps\\
            &\le \sum_{h=0}^{H-1}\Esub{\textbf{g}}{\max_{M \in \M^{\star}(S,S')} \Bigl\lvert \frac{1}{n'}\sum_{p \in S \setminus S(I_M)}(\prof(\phi_{h+1}(M),p) - \prof(\phi_h(M),p)) \cdot g_p\Bigr\rvert} 
        \end{align*}
    \end{lemma}
    \begin{proof}
        Note, the revenue of any mechanism $M \in \M^{\eps}_{S(I\setminus J), S'(J)}$ on any point $p \in S \setminus S(I)$ can be written via the following telescopic sum:
        \[
            \prof(M,p) = \prof(\phi_{0}(M),p) + \sum_{h=0}^{H-1} ( \prof(\phi_{h+1}(M),p) - \prof(\phi_{h}(M),p)).
        \]
        In fact, since $M \in \M^{\eps}_{S(I\setminus J), S'(J)}$ it holds that $\phi_H(M) = M$. We study separately the term that depends on $\app{0}{M}$ and the summation. In particular, we only need to address the latter term, as the former one already appears in the statement (simply by taking the summation out of the absolute value and the expected value).
        We want to prove the following inequality, which concludes the proof of the lemma.
        \begin{equation}
        \label{eq:phi_0}
            \Esub{\textbf{g}}{\max_{M \in \M^{\star}(S,S')} \Bigl\lvert \frac{1}{n'}\sum_{p \in S \setminus S(I_M)}\prof(\phi_{0}(M),p)\cdot g_p\Bigr\rvert} \le 38 \eps
        \end{equation}
        The mechanisms of the form $\app{0}{M}$ have a precise structure: they are of the form $[0,v_1]\times [v_2,1]$ for some node $(v_1,v_2)$ in one of the grids that define $\M^{\star}(S,S')$. {In particular, there are at most $3n^2$ such points, as $n\ge \nicefrac{8}{\eps}$. This means that when approximating through $\app{0}{M}$, given the tuple $(S, S')$, there are at most $3n^2$ distinct mechanism.}
        We have the following inequality for any $M \in \M^{\star}(S,S')$ (recall, $n' = n - \nicefrac{4}{\eps}$)
        \begin{align*}
            \max_{M \in M^\star(S,S')}&\frac{1}{n'}\Bigl\lvert \sum_{p \in S \setminus S(I_M)}\prof(\phi_{0}(M),p)\cdot g_p \pm \sum_{p \in S(I_M)}\prof(\phi_{0}(M),p)\cdot g_p\Bigr\rvert\\ &\le \underbrace{\max_{M \in M^\star(S,S')}\frac{1}{n'}\Bigl\lvert \sum_{p \in S}\prof(\phi_{0}(M),p)\cdot g_p \Bigr\rvert}_{(A)} + \underbrace{\max_{M \in M^\star(S,S')}\frac{1}{n'}\Bigl\lvert \sum_{p \in S(I_M)}\prof(\phi_{0}(M),p)\cdot g_p \Bigr\rvert}_{(B)}
        \end{align*}
        In term (A) we are taking the maximum over the absolute values of at most $3n^2$ centered gaussian terms, each with variance $\nicefrac{n}{(n')^2} \leq \nicefrac{4}{n}$. Using \Cref{lem:max}, we get that (A) is bounded in expectation by $8\sqrt{\nicefrac{\log n}{n}}$. For what concerns (B), by \Cref{prop:cardinality_net}, we are taking the maximum over the absolute values of at most $e^{\nicefrac{10\log n}{\eps}} \cdot 2^{\nicefrac{4}{\eps}} \le e^{\nicefrac{14\log n}{\eps}}$ gaussian terms, each with variance at most $\tfrac{16}{n^2\eps}$. Again by \Cref{lem:max}, we have that $(B)$ is bounded in expectation by $30 \sqrt{\tfrac{\log n}{\eps^2 n^2}}\leq 30 \sqrt{\tfrac{\log n}{n}}$, under $n\ge \nicefrac{1}{\eps^2}$.

            {By combining these two bounds, we get $38\sqrt{\nicefrac{(\log n)}{n}}$ as an upper bound for the sum of (A) and (B). Our assumption implies $n\ge \nicefrac{4}{\eps^2}\log\left(\nicefrac{1}{\eps}\right)$, so we conclude\footnote{{Notice that $\log n = \eps^2 n$ has no elementary solution in $n$, therefore the bound follows by upper bounding terms involving the Lambert $W$ function.}} by bounding ${\sqrt{\nicefrac{(\log n)}{n}}}$ with $\eps$.}
    \end{proof}

    We now define the clean event $\cE$ as follows: for all mechanisms $M \in \M^\star(S,S')$, we require that the following inequality is verified for all $h = 0, \dots, H-1$:
    \begin{equation}
    \label{eq:clean_event}
        \prof(\phi_h(M,S\setminus S(I_M)) - \prof(\phi_{h+1}(M,S\setminus S(I_M)) < 31 \cdot 2^{-h}.
    \end{equation}
    
    \begin{lemma}[Clean Event]
    \label{lem:clean_event}
        The clean event $\cE$ happens with probability {at least} $(1-\tfrac{1}{n\log \nicefrac{1}{\eps}})$.        
    \end{lemma}
    \begin{proof}
        Fix any mechanism {$M \in \M^\star(S, S')$}. This means that $M$ is a {candidate} in some family $\M^{\eps}_{S(I\setminus J) \cup S'(J)}$, {recalling \Cref{def:candidates}}. We denote with $\cE_{M,h}$ the following event:
        \[
            \cE_{M,h} = \left\{\prof(\phi_h(M,S\setminus S(I_M)) - \prof(\phi_{h+1}(M,S\setminus S(I_M)) \ge 31\cdot 2^{-h}\right\}.    
        \]
        Similarly to what is done in \Cref{thm:target}, we introduce the random variables $X_p$ and study their sum $X$ for $p \notin S(I_M)$:
        \[
            X_p = \tfrac 14 [\prof(\phi_h(M),p) - \prof(\phi_{h+1}(M),p) + \nicefrac{2}{2^h}].
        \]
        We know that $X_p \in [0,1]$, and, since we assume that $M$ is a candidate, we are under the assumption\footnote{{This is the opposite situation with respect to the argument in the proof of \Cref{thm:target}.}} that $\E{X} < \nicefrac{3}2 \cdot {2^{-h}} \cdot n'$ (as $\E{\prof(\phi_h(M)) - \prof(\phi_{h+1}(M)} \le \nicefrac{4}{2^h}$, and $n' = |S \setminus S(I)|$). We then have the following chain of inequalities, {under our assumptions on $n$}:
        \begin{align*}
        \tag{By def. of $X$}
            \P{\cE_{M,h}} & = \P{X \ge 33 \cdot 2^{-h}\cdot n'} \\
            &\le {2}^{-33 \cdot 2^{-h}\cdot n'} \tag{By \Cref{thm:chernoff}}\\
            &\le {2}^{-\tfrac{33}{2}\eps n} \tag{As $2^{-h}\ge \eps$ and $n'\ge \nicefrac n2.$}
        \end{align*}
        Note, we applied the second inequality in \Cref{thm:chernoff}, using the fact that
        \begin{align*}
            2 e \E{X}& \le  3e \cdot {2^{-h}} \cdot n'\le  33 n'\cdot 2^{-h}.
        \end{align*}
        The only missing ingredient is now to take the union bound over all possible levels and mechanisms in $\M^\star(S,S')$. {The levels are at most $2\log \nicefrac{1}{\eps}$} and by \Cref{prop:cardinality_net} the number of mechanisms is at most {$e^{\nicefrac{10\log n}{\eps}}$}, multiplied by the possible way of choosing $J \subset I$, i.e. $2^{\nicefrac{{4}}{\eps}}$, {the cardinality of the power set of $J$}. So the probability of failure is at most: {
        \[
            2\cdot \log \nicefrac{1}{\eps}\cdot 2^{-\tfrac{33}{2}\eps n + \tfrac{4}{\eps}}e^{\tfrac{10\log n}{\eps}} \le \tfrac{1}{n\log \nicefrac{1}{\eps}},     
        \]
        as long as $n \ge \tfrac{8 \log \nicefrac{1}{\eps}}{\eps^2}$}.
    \end{proof}

    \begin{lemma}
    \label{lem:single_h}
        Consider any realization of the samples $S$ and $S'$ that is compatible with the clean event $\cE$. Then for any $h=0, \dots, H-1$, we have the following inequality:
        \begin{equation*}
            \Esub{\textbf{g}}{\max_{M \in \M^\star(S,S')} \Bigl\lvert \frac{1}{n'}\sum_{p \in S \setminus S(I)}(\prof(\phi_{h+1}(M),p) - \prof(\phi_h(M),p)) \cdot g_p\Bigr\rvert}\le {116\,\tfrac{\eps}{\log \nicefrac{1}{\eps}}}
        \end{equation*}
    \end{lemma}
    \begin{proof}
        We are under the clean event, therefore we know that, for any $M\in \M^\star(S,S')$ we have that \Cref{eq:clean_event} is verified. Fix any such $M$, it belongs to some $\M^{\eps}_{S(I\setminus J),S'(J)}$ for some $I$ and $J \subseteq I$. We want to apply \Cref{prop:local} to the empirical distribution induced by the points in $S\setminus S(I)$. By \Cref{eq:clean_event}, we can then conclude that, for all $h=0, \dots, H-1$ the following inequality is verified: 
        \begin{equation}
        \label{eq:variance_bound}
            \frac{1}{n'}\sum_{p \in S\setminus S(I)}(\prof(\app{h+1}{M}) - \prof(\app{h}{M})^2 \le {136} \cdot 2^{-h}.
        \end{equation}
        Therefore, if we consider the random variables that appear in the absolute value in the statement {of this theorem}, they are all centered gaussian variables with variance bounded as follows:
        \begin{align*}
            &\Var{\frac{1}{n'}\sum_{p \in S \setminus S(I)}(\prof(\phi_{h+1}(M),p) - \prof(\phi_h(M),p)) \cdot g_p}\\
            &=\frac{1}{(n')^2}\sum_{p \in S \setminus S(I)}(\prof(\phi_{h+1}(M),p) - \prof(\phi_h(M),p))^2 \Var{g_p}\le \frac{{136}}{n'} \cdot 2^{-h},
        \end{align*}
        where we used the property of the variance, the fact that the $g_p$ are i.i.d. standard gaussians, and the bound provided by \Cref{eq:variance_bound}. 
        
        The last ingredient we need before applying the result in \Cref{lem:max} is to provide an upper bound on the number of mechanisms that appear in the $\max$. The points in $S$ and $S'$ are fixed, moreover, each mechanism of the form $\phi_h(M)$ for $M \in \M^\star(S,S')$ is supported on the grid $G^{2^{-h}}_{S \cup S'}$, {recalling \Cref{def:data_dependent_net}}, and has at most $2^{h+1}$ ``corners'', {by the same potential argument used in the second part of the proof of \Cref{lem:few_corners}}. Therefore, {using monotonicity of the allocation regions}, the cardinality $N$ of the set onto which the $\max$ is taken over is (note, each term contains both $\phi_h$ and $\phi_{h+1}$):
        \[
            N \le \binom{2n+\nicefrac{{2}}{\eps}}{2^{h+1}} \cdot \binom{2n+\nicefrac{{2}}{\eps}}{2^{h+2}} \le \binom{4n}{2^{h+1}} \cdot \binom{4n}{2^{h+2}}  \le {(64^2n^6e^6)} ^{2^h}.
        \]
        Note, in the above inequality we used that $\binom{n}{k} \le \left(\nicefrac{(en)}{k}\right)^k.$ We can then conclude via \Cref{lem:max} to claim that the left-hand side term in the statement is upper-bounded by:
        {\begin{align*}
            2 \sqrt{\sigma^2 \log N} &\le 2\sqrt{\frac{140}{\nicefrac{n}{2}} \cdot 2^{-h}\cdot \log\left((4en) ^{6 \cdot 2^h}\right)}\\
            &\le 2\sqrt{\frac{140}{\nicefrac{n}{2}} \cdot 2^{-h} \cdot 6 \cdot 2^h \log (11n))} \le 116\sqrt{\frac{\log n}{n}}
        \end{align*}}
    The statement follows since $\sqrt{\tfrac{\log n}{n}}\le \tfrac{\eps}{\log \nicefrac{1}{\eps}}$ when we assume that $n\ge \nicefrac{32}{\eps^2}\log^3\left(\nicefrac{1}{\eps}\right)$.
    \end{proof}

    \begin{lemma}
    \label{lem:final}
        Consider any realization of the samples $S$ that is compatible with the clean event $\cE$. We have the following inequality:
        {\begin{equation*}
            \Esub{\textbf{g}}{\max_{M \in \M^{\star}(S,S')} \Bigl\lvert \frac{1}{|S\setminus S(I_M)|}\sum_{p \in S \setminus S(I)}\prof(M,p) \cdot g_p\Bigr\rvert} \le 270 \eps
        \end{equation*}}
    \end{lemma}
    \begin{proof}
        It follows by combining \Cref{lem:single_h} with \Cref{lem:telescoping}, and noting that $H = \log_2 \nicefrac{1}{\eps}\le 2 \log \nicefrac{1}{\eps}$.
    \end{proof}

    \begin{proof}[Proof of \Cref{thm:expectation_concentration}]
        The result follows by combining the above inequalities:
        {
        
        \begin{align*}
        &\Esub{S}{\max_{M \in \M^\star(S)} \Bigl\lvert\prof(M,S) - \E{\prof(M)}\Bigr\rvert} - \eps    \\
        \tag{By \Cref{lem:ghost}}&\le \Esub{S,S'}{\max_{M \in  \M^{\star}(S)} \Bigl\lvert\prof(M,S\setminus S(I_M)) - \prof(M,S'\setminus S'(I_M))\Bigr\rvert}\\
        \tag{\Cref{lem:symmetrization}}
        &\le 2 \Esub{S,S',\textbf{r}}{\max_{M \in \M^\star(S,S')} \Bigl\lvert \frac{1}{n'}\sum_{p \in S \setminus S(I_M)}\prof(M,p) \cdot r_p\Bigr\rvert}\\
        &\le \sqrt{2\pi} \Esub{S,S',\textbf{g}}{\max_{M \in \M^\star(S,S')} \Bigl\lvert \frac{1}{n'}\sum_{p \in S \setminus S(I_M)}\prof(M,p) \cdot g_p\Bigr\rvert}\tag{\Cref{lem:rad_to_gaus}}\\
        &\le \sqrt{2 \pi} \Biggr(\Esub{S,S',\textbf{g}}{\max_{M \in \M^\star(S,S')} \Bigl\lvert \frac{1}{n'}\sum_{p \in S \setminus S(I_M)}\prof(M,p) \cdot g_p\Bigr\rvert \Bigr|\, \cE} \P{\cE} + 38\eps\\
        &  + \sum_{h=0}^{H-1}\Esub{S,S',\textbf{g}}{\max_{M \in \M^{\star}(S,S')} \Bigl\lvert \frac{1}{n'}\sum_{p \in S \setminus S(I_M)}(\prof(\phi_{h+1}(M),p) - \prof(\phi_h(M),p)) \cdot g_p\Bigr\rvert \Bigr|\, \cE^c}\P{\cE^c} \Biggr)\tag{By \Cref{lem:telescoping} and Law of Total Probability}
        \\
        &\le \sqrt{2 \pi} \left(308 \eps + \tfrac{1}{n\log(\nicefrac{1}{\eps})}\cdot 2\log(\nicefrac{1}{\eps}) \cdot \sqrt{14\log n}\right) \tag{By \Cref{lem:max,lem:clean_event,lem:final}} \\
        & \le 793\eps.
        \end{align*}
        }
        In the second to last inequality, we can bound the first term by \Cref{lem:final}, while the second can be bounded with \Cref{lem:clean_event} and by noticing that, when conditioning on $(S, S')$, we are taking the maximum over at most $\exp(\nicefrac{10\log n}{\eps})\cdot 2^{\nicefrac{4}{\eps}} \le \exp(\nicefrac{14\log n}{\eps})$ centered gaussian terms with variance at most $\nicefrac{1}{n'}\leq \nicefrac{2}{n}$. By \Cref{lem:max}, we have that the inner expectation w.r.t. the gaussian terms is bounded by $\sqrt{14 \log n}$.         
    \end{proof}

%% file: sections/50-lower_bounds.tex
\section{Lower Bounds}
\label{sec:lower-bounds}

    In this section, we prove that the positive result in \Cref{sec:algorithm} is, essentially, tight. In particular, in \Cref{subsec:lower_bound_stochastic} we prove that no learning algorithm can achieve better than $\sqrt{T}$ regret in the stochastic setting. Finally, in \Cref{subsec:lower_bound_adversarial} we prove that no learning can happen in the adversarial setting.

    \subsection{Stochastic Lower Bound}
    \label{subsec:lower_bound_stochastic}

        We derive the $\sqrt{T}$ lower bound in the stochastic setting by embedding into our framework the hard instance for prediction with two experts. In particular, we consider two stochastic instances that are statistically close but admit two well-separated optimal mechanisms.

        Consider the probability measure $\Pbz$ under which $V$ is drawn uniformly at random in the set $\{(0,1), (\nicefrac 14,\nicefrac 34)\}$. We use $\Pbz$ as ``baseline distribution'' to compare the following two measures $\Pbo$ and $\Pbt$. Let $\eps>0$ a precision parameter we set later, under measure $\Pbo$, the random valuation $V$ is drawn according to the following distribution:
        \[
            \begin{cases}
                (0,1) &\text{ with probability $\nicefrac{(1+\eps)}{2}$}\\
                 (\nicefrac 14,\nicefrac 34) &\text{ with probability $\nicefrac{(1-\eps)}{2}$}.
            \end{cases}
        \]
        Under $\Pbt$, the distribution is the same, but with flipped signs, so that $\Pt{V=
        (0,1)} = \nicefrac{(1-\eps)}{2}$ and  $\Pt{V=(\nicefrac 14,\nicefrac 34)} = \nicefrac{(1+\eps)}{2}$.

        \begin{remark}
            To avoid overloading the notation, we do not consider ``different'' random variables but vary the underlying probability measure. We denote with $\mathbb{E}^i$, for $i = 0,1,2$ the corresponding expectation. 
        \end{remark}

        Under these three probability measures, there is a clear way to partition the mechanism space, according to whether the allocation region contains $(\nicefrac 14, \nicefrac 34)$ or not. To this end, denote with $M_1$ the mechanism that allocates only in the point $(0,1)$, and with $M_2$ the one that allocates in the rectangle $[0,\nicefrac 14] \times [\nicefrac 34,1].$ Clearly, $M_1$ or $M_2$ dominate any mechanism, as formalized by the following lemma.

        \begin{lemma}[Domination]
        \label{lem:lower_bound_domination}
            Consider any mechanism $M$ with allocation region $A$, then one of the following is verified:
            \begin{itemize}
                \item If $(\nicefrac 14, \nicefrac 34) \notin A$, then $\Ei{\prof(M_1)} \ge \Ei{\prof(M)}$ for $i=0,1,2$
                \item If $(\nicefrac 14, \nicefrac 34) \in A$, then $\Ei{\prof(M_2)} \ge \Ei{\prof(M)}$ for $i=0,1,2$
            \end{itemize}
        \end{lemma}
        \begin{proof}
            Consider the first case: $(\nicefrac 14, \nicefrac 34) \notin A$. If $A$ is empty, then clearly $M_1$ dominates $M$, as its expected profit is approximately $\nicefrac{1}{2}$ under $\Pbi$ for $i=0,1,2.$ If $A$ is not empty, then it contains $(0,1)$, therefore the profit it extracts from it is at most that extracted by $M^1$, whose allocation region is contained in $A$ (\Cref{prop:trade-off}). 

            We move our attention to any mechanism $M$ whose allocation region $A$ contains $(\nicefrac 14, \nicefrac 34)$. The smallest monotone allocation region that contains $(\nicefrac 14, \nicefrac 34)$ is the rectangle $[0,\nicefrac 14]\times[\nicefrac 34,1]$, i.e., the allocation region of $M_2$. Therefore, the profit extracted by $M_2$ is at least that extracted by $M$ (again, by \Cref{prop:trade-off}).
        \end{proof}

        We have all the ingredients to prove our first lower bound. 

        \begin{theorem}[Lower Bound in the Stochastic Case]
        \label{thm:iid_lower_bound}
            Consider the profit maximization problem in bilateral trade against the stochastic i.i.d. adversary. Then any algorithm $\A$ suffers at least  
            \(
                 \nicefrac{1}{16} \sqrt{T}
            \) regret.
        \end{theorem}
        \begin{proof}
            We prove that any learning algorithm $\A$ suffers $\Omega(\sqrt{T})$ regret rate against at least one between $\Pbo$ and $\Pbt$. For $i=0,1,2$, denote with $\Pbi_V$ the push-forward measure on $[0,1]^2$ (equipped with the Borel $\sigma$-algebra) induced by the random variable $V$ when the underlying probability measure is $\Pbi$. We start formalizing the intuition that the $\Pbo$ and $\Pbt$ are ``close'' to $\Pbz$ in a statistical sense by studying their Kullback-Leibler divergence $\KL$.
            \begin{claim}
            \label{cl:kl}
                 If $\eps \in (0,\nicefrac 12)$, it holds that $ \max\left\{\kl{\mathbb P_V^0}{\mathbb P_V^1}, \kl{\mathbb P_V^0}{\mathbb P_V^2}\right\} \le 2 \eps^2$.
            \end{claim}
            \begin{proof}[Proof of the Claim]
                We bound the $\KL$ divergence between $\Pbo$ and $\Pbz$; the other term is equal. 
                \begin{align*}
                    \kl{\mathbb P_V^0}{\mathbb P_V^1} &= \frac 12 \log \left(\frac{1}{1+\eps}\right) + \frac 12 \log \left(\frac{1}{1-\eps}\right) \tag{By definition}\\
                    &= - \frac{1}{2} \log (1-\eps^2) \le 2\eps^2,
                \end{align*}
                where the last inequality holds for any $\eps \in (0,\nicefrac 12).$
            \end{proof}

            Under the stochastic instance represented by $\Pbz$, both $M_1$ and $M_2$ are optimal (by \Cref{lem:lower_bound_domination}), and achieve expected profit: 
            \(
                \Ez{\prof(M_1)} = \Ez{\prof(M_2)} = \tfrac 12.
            \)
            When the instance is induced by $\Pbi$, the domination result in \Cref{lem:lower_bound_domination} tells us that the optimal mechanism is $M_i$, with the performance of $M_j$ for $j \neq i$ that is $\Theta(\eps)$ far:
            \begin{equation}
            \label{eq:profit_Pi}
                \begin{cases}
                    \Eo{\prof(M_1)} = \tfrac{1+\eps}{2}, \quad \Eo{\prof(M_2)} = \tfrac 12\\
                    \Et{\prof(M_1)} = \tfrac{1-\eps}{2}, \quad \Et{\prof(M_2)} = \tfrac 12
                \end{cases}
            \end{equation}  

            Fix now any deterministic learning algorithm $\A$, we prove that its expected regret with respect to the randomized instance that is either $\Pbo$ or $\Pbt$ with equal probability is $\Omega(\sqrt{T})$. Note, the randomized instance is constructed by first choosing $\Pbi$ uniformly at random and then sampling from it, \emph{not} by changing distribution at each time step. We remark that proving this result (restricting to deterministic algorithms) is enough to prove the theorem, as the randomized instance we use is independent of the algorithm used by the learner. 

            We denote with $N_i$ the number of times that algorithm $\A$ plays the mechanism $M_i$, for $i=1,2$. Note, there is no point for $\A$ to play any other mechanism, by \Cref{lem:lower_bound_domination}.
            Denote with $V_1, V_2, \dots$ the i.i.d. samples of $V$, and denote with $\Pbi_{V_1, \dots, V_t}$ the push-forward to $[0,1]^{2t}$. We have the following chain of inequalities:
            \begin{align}
            \nonumber
                |\Ei{N_i} - \Ez{N_i}| &\le \sum_{t=1}^T | \Pi{M^t = M_i} - \Pz{M^t = M_i}|\\
                &\le \sum_{t=2}^T ||\Pbi_{V_1, \dots, V_t} - \Pbz_{V_1, \dots, V_t}||_{\tv} \tag{By definition of total variation}\\
                &\le \sum_{t=2}^T \sqrt{\frac{t}{2}\kl{\mathbb P_V^i}{\mathbb P_V^0}} \tag{By Pinsker Inequality and i.i.d. design}\\
                &\le \eps T \sqrt{T}, \label{eq:Eo_vs_Ei}
            \end{align}
            where the last inequality follows by \Cref{cl:kl}.

            We are ready to bound the expected performance of $\A$ against the randomized instance constructed by picking either $\Pbo$ or $\Pbt$ with equal probability. Note, under $\Pbi$, the regret of $\A$ is induced by the $N_j$ time steps in which it plays mechanism $M_j$ for $j\neq i$, multiplied by the corresponding instantaneous regret of $\nicefrac{\eps}2$ (by \Cref{{eq:profit_Pi}}): 
            \begin{align*}
                \E{R_T(\A)} &= \frac{\eps}4\left(\Eo{N_2} + \Et{N_1}\right)\\
                &\ge \frac{\eps}4 \left(\Ez{N_1 + N_2} - 2 \eps T \sqrt T\right) \tag{By \Cref{eq:Eo_vs_Ei}}\\
                &= \frac{\eps}2 T \left(1 - 2 \eps \sqrt T\right) \tag{As $N_1 + N_2 = T$} \\
                &= \frac{1}{16}\sqrt{T}. \tag{By setting $\eps = \nicefrac{1}{4\sqrt{T}}$}
            \end{align*}
            This concludes the proof. 
        \end{proof}

    \subsection{Adversarial Lower Bound}
    \label{subsec:lower_bound_adversarial}

        This section is devoted to our adversarial lower bound. We not only prove that it is impossible to achieve sublinear regret in the adversarial case, but argue that no sublinear $\alpha$-regret is possible, for any constant $\alpha \in [1,\nicefrac 32)$.\footnote{For $\alpha = 1$ the definition of $\alpha$-regret (as in the statement of \Cref{thm:adv_lower_bound}) collapses to that of regret.} 

        \begin{theorem}[Lower Bound in the Adversarial Case]
        \label{thm:adv_lower_bound}
            Consider the profit maximization problem in bilateral trade in the adversarial setting. For any constant $\eps \in(0,1]$ and any learning algorithm, there exists an adversarial sequence such that the following $(\nicefrac 32-\eps)$-regret bound holds: 
            \[
            \sup_{M \in \M} \sum_{t=1}^T \prof_t(M,v^t) - \left(\frac 32-\eps\right) \cdot \sum_{t=1}^T\E{ \prof_t(M_t,v^t)} \ge  \frac{\eps}{4} T.
        \]
        \end{theorem}
        \begin{proof}
            We prove this lower bound via Yao's minimax principle: we construct a randomized (adversarial) instance against which any deterministic algorithm suffers expected linear $(\nicefrac 32-\eps)$ regret. Our hard instance is based on an auxiliary sequence constructed in Theorem~4 of \citet{AggarwalBDF24}.

            Let $\delta$ be a small auxiliary parameter we set later. We construct our randomized instance while maintaining two auxiliary sequences, $a_t$ and $b_t$, that are initialized as $a_1 = \nicefrac{\delta}{3}$ and $b_1 = \nicefrac{2\delta}{3}$; we also denote with $\Delta_t$ the difference $b_t-a_t$, which is maintained at $\delta \cdot 3^{-t}$ throughout. The sequence of valuations $(v_s^t,v_b^t)$ is generated iteratively as follows:
            {
            \begin{itemize}
                \item with probability $\nicefrac 12$, the valuations are $(v_s^t, v_b^t) = (b_t, \nicefrac 12)$, and the auxiliary sequence is updated as follows: $(a_{t+1}, b_{t+1}) = (b_t + \Delta_{t+1}, b_t + 2\Delta_{t+1})$
                \item with probability $\nicefrac 12$, the valuations are $(v_s^t, v_b^t) = (a_t, 1)$, and the next pair in the auxiliary sequence is $(a_{t+1}, b_{t+1}) = (a_t - 2\Delta_{t+1}, a_t - \Delta_{t+1})$.
            \end{itemize}  
            }
        Arguing similarly to \citet{AggarwalBDF24}, it is possible to verify that this random sequence is well defined, as all realizations of the pairs $(v_s^t,v_b^t)$ belong to the $[0,1]^2$ square. Crucially, {\em for any realization}, there exists a threshold $\tau$ that separates the valuations of the form $(b_t,\nicefrac 12)$ from those of the form $(a_t,1)$. Formally, denote with $L$ (for left) the time steps characterized by a $(a_t,1)$ valuation, and with $R$ (for right) the ones corresponding to $(b_t,\nicefrac 12)$ valuations, the following proposition holds (this is a restatement of Proposition~4 of \citet{AggarwalBDF24}, we provide a more direct proof in \Cref{app:sequence} for completeness).
        \begin{restatable}{proposition}{sequence}       \label{prop:properties}
            For any realization of the valuations and auxiliary sequence, we have:
            \begin{itemize}
                \item[(i)] $a_t \in [0,\delta]$ and $b_t \in  [0,\delta]$, for all $t = 1,\ldots, T$.
                 \item[(ii)] there exists a value $\tau \in  (0,\delta)$ such that $b_t < \tau$ for all $t \in R$, while $a_t > \tau$ for all $t \in L$. 
            \end{itemize}
        \end{restatable}  

        The threshold value $\tau$ can be used to design a mechanism $M^\star$ that allocates in $[0,\tau]\times[\nicefrac 12,1]$ united with the $(0,1)$-$(\delta,1)$ segment. $M^\star$ accepts all the bids of the form $(x,\nicefrac 12)$ for $x \le \tau$ (with profit at least $\nicefrac 12-\delta$), and all the bids of the form $(x,1)$ with $x \in [\tau, \delta]$ (with profit at least $1-\delta$).
        In each iteration, the expected profit that $M^\star$ extracts is thus at least 
        \begin{equation}
        \label{eq:prof_mstar}
            \E{\prof_t(M^\star)} \ge \tfrac{1}{2}\left(1-\delta\right) + \tfrac{1}{2}(\tfrac{1}{2}-\delta) =  \tfrac{1}2\left(\tfrac 32 -2\delta\right).
        \end{equation}
        
        Consider now any deterministic algorithm and its expected revenue at time $t$. Fix the history of the sequence up to time $t$ (excluded); there are two options for the next valuations, with equal probability: either $(b_t,\nicefrac 12)$ or $(a_t,1)$. Since the latter dominates the former, any mechanism $M^t$ that the learning algorithm posts can do one of three things: it never trades, it makes the trade happen for both possible valuations $(b_t, \nicefrac 12)$ and $(a_t, 1)$ (with an expected profit of at most $\nicefrac 12$), or only for $(b_t,\nicefrac{1}{2})$ (with an expected profit of at most $\nicefrac 12$). All in all, the expected revenue is at most 
        \begin{equation}
        \label{eq:prof_mt} 
            \E{\prof_t(M^t)} \le \tfrac 12.
        \end{equation}

        Combining \Cref{eq:prof_mstar} and \Cref{eq:prof_mt}, we see that the expected $\nicefrac 32 - \eps$ instantaneous regret at the generic time step $t$ is at least:
        \[
            \E{\prof_t(M^\star)} - \left(\tfrac 32 - \eps \right)\E{\prof_t(M^t)} \ge \tfrac{1}2\left(\tfrac 32 -2\delta\right) - \left(\tfrac 32 - \eps \right) \tfrac 12 = \tfrac{\eps}{4},
        \]
        where in the last equality we set $\delta = \nicefrac{\eps}{4}.$ Summing up over all the time steps $t$, yields the desired lower bound on the $(\nicefrac 32 - \eps)$-regret.
        \end{proof}

%% file: sections/60-joint_ads.tex
\section{The Joint Ads Problem}

    In this section, we showcase the applicability of our chaining analysis from \Cref{sec:algorithm} by providing an upper bound of $\sqrt{T}$ on the regret for the joint ads problem \citep{AggarwalBDF24}, in the i.i.d. stochastic setting. This result, together with the $\Omega(\sqrt{T})$ lower bound from \citet{AggarwalBDF24} provides a tight (up to poly-logarithmic factors) characterization of the minimax regret rate for the problem. For completeness, we recall that an $\Omega(T)$ lower bound in the adversarial setting is provided in \citet{AggarwalBDF24}.

    \subsection{The Joint Ads Problem}

        In the Joint Ads problem, two buyers participate in an auction for a non-excludable good (e.g., a public good or a shared ad slot that is either allocated to \emph{both} agents or to none). The goal of the mechanism designer is to maximize the revenue extracted from the trade. We now briefly introduce the problem; for further details we refer to \citet{AggarwalBDF24}.

        \paragraph{The Learning Protocol.} Formally, at each time step $t = 1, 2,\dots, T$, a new pair of agents arrives, characterized by private valuations $(v^t_1, v^t_2)$ in the $[0,1]^2$ square. Independently, the learner proposes a mechanism $M^t$ to the agents, who then declare bids $(b_1^t,b_2^t)$. A mechanism $M^t$ is composed by a non-excludable allocation rule $x^t: [0,1]^2 \to \{0,1\}$ (either both agents win the auction or both lose) and two payment functions $p_1^t,p_2^t : [0,1]^2 \to [0,1]$ that map bids to allocation and payment, respectively. The revenue of mechanism $M^t$ when the bids are $b^t_1$ and $b_2^t$ is defined as
        \[        
            \rev(M^t,b^t_1,b^t_2) = x^t(b^t_1,b^t_2) \cdot [p^t_1(b^t_1,b^t_2) + p^t_2(b^t_1,b^t_2)].
        \]
   
        \paragraph{Structure of Incentive Compatible Mechanisms.} As in the bilateral trade problem, the agents behave strategically, meaning that they strive to maximize their own quasi-linear utility $u^t_i(b_1^t,b_2^t) = v^t_i \cdot x^t(b_1^t,b_2^t) - p^t_i(b_1^t,b_2^t)$ for $i \in \{1,2\}$. Note that, differently from the bilateral trade problem, here the two agents are symmetric, as they are both interested in purchasing a good. We require the mechanisms $M^t$ to be DSIC and IR. The family of mechanisms $\M$ that respect these two properties can be easily characterized because this model falls within the framework of single parameter auctions \citep{myerson81}, as each agent is characterized by a single value representing its value for the good for sale. This tells us that $(i)$ an allocation rule is implementable if and only if it is monotone and $(ii)$ for each such allocation rule there exists a unique payment rule that completes it to a DSIC and IR mechanism.\footnote{As in bilateral trade, payments are unique only up to constant shifts. Since our goal is to maximize revenue, we adopt the maximum payments to achieve DSIC and IR.} For this reason, we only talk about valuations in the rest of the section, as the agents are incentivized to behave truthfully.

        \paragraph{Differences with Bilateral Trade.} There are two main differences with bilateral trade. First, the two agents are symmetric; this induces a different notion of monotonicity: valuations $(v_1,v_2)$ dominate $(v_1',v_2')$ if $v_i\ge v_i'$ for $i=1,2.$ Stated differently, if at least one of the two agents raises its bid, then the likelihood of having a trade does not decrease. Second, in the joint ads problem, the objective of the learner, i.e., the revenue, lies in $[0,2]$, as opposed to the profit in bilateral trade, which lives in $[-1,1]$. This is a \emph{crucial} distinction: while in both problems the event ``no trade'' induces zero gain to the learner, in bilateral trade, there is a notion of ``bad trade'' that induces \emph{negative} utility for the learner. On the contrary, in the joint ads problem, every trade guarantees at least non-negative payments \emph{to} the learner. 

        \begin{figure}[t!]
            \centering
                \includegraphics[width=0.4\linewidth]{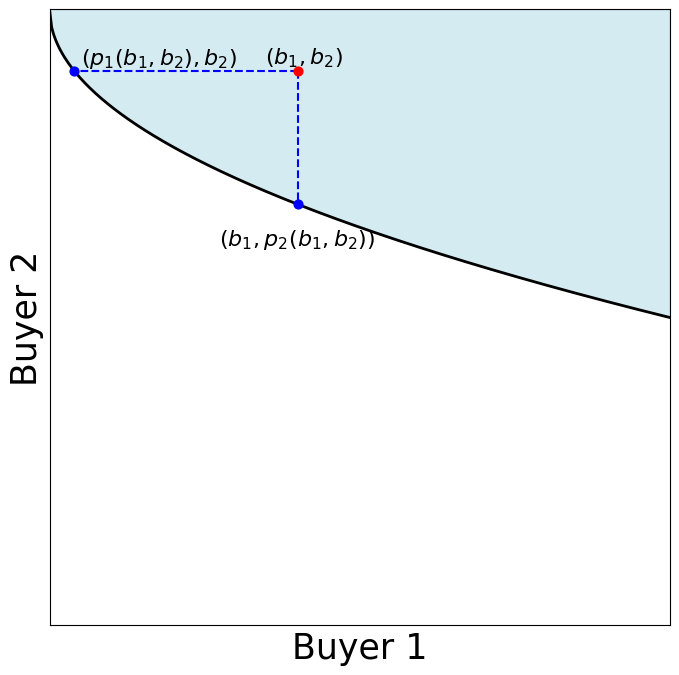}
            \caption{Visualization of the allocation region of a joint ad mechanism.}
            \label{fig:joint}
        \end{figure}
        \subsection{The Stochastic Algorithm}

            In this section, we argue that the algorithm \simplify easily adapts to this new problem and achieves the desired $\tilde O(\sqrt T)$ regret bound in the stochastic setting. From the technical point of view, the main difference with respect to the bilateral trade problem is that the revenue function lies in the $[0,2]$ interval, so that (i) observing ``no trade'' is the worst occurrence and (ii) the gap between the best revenue and no-trade is $2$ instead of $1$. Point (i) implies that we do not want to focus our attention only on mechanisms living in the bottom left triangle (note, with the notion of monotonicity induced by this problem, the ``low-value'' corner of the square is the one close to the origin), so that we do not restrict the mechanism space as in \Cref{lem:Mplus}. 

            For the joint ads problem, \simplify collapses to the algorithm used in \citet{AggarwalBDF24} (\augment). There is, however, one big difference: our chaining-based analysis allows us to prove a greatly improved sample complexity of $\nicefrac{1}{\eps^2}$, as opposed to the sample complexity of $\nicefrac{1}{\eps^4}$ implicit in \citet{AggarwalBDF24}. This difference is observable in the different choice of precision parameter: we set $\eps_t = 200 \sqrt{\nicefrac{\log^3 T}{t-1}}$, as opposed to approximately $t^{-\nicefrac{1}{4}}$. 

            \begin{theorem}
            \label{thm:regret-bound-joint-ads}
                Consider revenue-maximization in the joint ad problem against the stochastic i.i.d. adversary. The regret of $\augment$ is at most $C\sqrt{T \log^3 T}$, for some constant $C$.
            \end{theorem}
            \begin{proof}
                The proof of this result follows verbatim from the proof of \Cref{thm:regret-bound} by replacing the $\prof$ function with $\rev$, and by considering the different notion of monotonicity (this implies, for instance, that the approximating sequences approximate mechanisms from the bottom-left direction, and not from the bottom-right, see for instance \Cref{fig:joint}). Note, we do not need to worry about negative revenue (this simplifies some of the arguments), while we get an extra multiplicative $2$ in the regret rate as $\rev \in [0,2]$, while $|\prof| \in [0,1]$.
            \end{proof}

%% file: sections/100-appendix.tex
\section{Characterization of DSIC and IR Mechanisms}
\label{app:mechanisms}

    In this section, we prove the characterization of DSIC and IR mechanisms.

\characterization*
\begin{proof}
    Consider a mechanism with allocation region $A$ and payment functions $q(\bs,\bb)$ and $p(\bs,\bb)$. Let us use the shorthand $\trade{b_1}{b_2} = \ind{(\bs, \bb) \in A)}$ to denote the mechanism's decision whether a pair of bids $\bs,\bb$ should trade. 
    
    First assume that $A$ is monotone and that payments are Myerson payments. 
    We want to show that then $M$ is both DSIC and IR. We take into consideration the perspective of the buyer. We fix a bid $\bs$ for the seller and a valuation $\vb$ for the buyer such that bidding truthfully enables a trade, i.e., $(\bs, \vb) \in A$. In this case, overbidding with $b'_b > \vb$ leaves the utility invariant, since a trade still happens and $q(\bs, b'_b) = q(\bs, \vb)$. Specifically, in this case, $\util{\bs}{b'_b}{b} = \util{\bs}{\vb}{b} = \vb - \min\{y \in [0,1] \mid (\bs,y) \in A\}$, which is non-negative because $\min\{y \in [0,1] \mid (\bs,y) \in A_M\} \leq \vb$. Underbidding renders the same scenario, up until we reach bids $b'_b$ such that $(\bs, b'_b)\notin A_M$; this leads to the trade not happening and the buyer obtaining zero utility, thus losing the surplus described before. If, on the other hand, the buyer's value $\vb$ is such that with a truthful bid no trade happens, i.e., $(\bs, \vb) \notin A$, then any bid $b'_b$ such that $(\bs,b'_b) \not\in A$ keeps the utility at zero, while any bid $b'_b$ such that $(\bs,b'_b) \in A$ will entail a payment $> \vb$ and thus result in negative utility. Therefore, bidding truthfully is the dominant strategy independently of the value of $\vb$. For what concerns individual rationality, notice that we have already considered the relevant cases for the value of $\vb$ and argued that the induced utility $\util{\bs}{\vb}{b}$ is non-negative.
    The specular reasoning for the seller, shows that the mechanism is DSIC and IR also for the seller.

    We next establish 
    the reverse implication. 
    To this end, we show that 
    DSIC 
    and 
    IR 
    imply both monotonicity of $A$ and Myerson  payments. Let us consider 
    the perspective of the seller: we fix a bid $\bb$ for the buyer and two real numbers $0\leq y\leq z$, with $(y, z)\in [0, 1]^2$. At first, we consider $y$ as the real valuation of the seller, and $z$ a possibly non-truthful bid; incentive compatibility gives the inequality $p(y, \bb) - x(y, \bb)y \geq p(z, \bb) - x(z, \bb)y$. From now on, we disregard the $\bb$ argument since we are keeping it fixed. Conversely, if we let $z$ be the valuation and $y$ the bid, we have: $p(z) - x(z)z \geq p(y) - x(y)z$. By simply factoring out $z$ in the second inequality and $y$ in the first, we get both an upper bound and a lower bound for price changes: 
    \begin{equation}
    \label{eq:price_sandwich}
        z(x(z)-x(y))\leq p(z) - p(y) \leq y(x(z)-x(y))
    \end{equation}
    Notice that the assumption of $x$ increasing immediately leads to a contradiction with the inequalities above, so $x$ is weakly decreasing when taking as argument only the bid for the seller.
    An analogous argument applies to the the buyer, and shows that $x$ is weakly increasing when taking as argument only the bid of the buyer. We have thus established monotonicity. 
    
    To show the need for Myerson payments, let us focus again on the seller. Notice that $x(\cdot,\bb)$ can only take two values: either we are in the allocation region and the trade happens, or not. This immediately implies that the payment $p(\cdot,\bb)$ can only take two different values, one for the case where seller and buyer trade and one for the case where there is no trade. We next argue that if the discontinuity of $x(\cdot,\bb)$ happens at $y$ (i.e., $x(y,\bb) = 1$ while for any $z > y$ we have $x(z,\bb) = 0$), then the difference between these two payments must be $y$.
    To this end, note that 
    the lower and upper bounds from \Cref{eq:price_sandwich} both converge to $-y$ as $z\rightarrow y$ from above, meaning that $\lim_{z\rightarrow y^+} p(z)-p(y) = -y$. 
    We have thus shown that $p(\bs,\bb) = y + c_{\bb}$ for $\bs \leq y$ and $p(\bs,\bb) = c_{\bb}$ otherwise, where $c_{\bb} \geq 0$ is a non-negative constant that may depend on the buyer's bid.  
    By our requirement that $p(\bs,\bb) = 0$ when $(\bs,\bb) \not \in A$ and there is no trade, we must have $c_{\bb} = 0$.
    An analogous reasoning for the buyer completes the proof. 
    \end{proof}

    \begin{remark}\label{rem:wlog}
    The requirement that payments are bounded in $[0,1]$ and zero when there is no trade is without loss of generality, when focusing on DSIC and IR mechanisms and the goal of maximizing the broker's profit. This can be seen by considering general payment functions $p,q: [0,1]^2 \rightarrow \mathbb{R}$. Then an argument analogous to the one used to establish \Cref{thm:mechanisms} shows that the payments must be Myerson payments, up to constant terms $c_{\bb}$ and $c_{\bs}$ that are added to the seller's and buyer's payment, respectively. By IR these constant must be such that $c_{\bb} \geq 0$ and $c_{\bs} \leq 0$, and, to maximize profit, we want $c_{\bb} = c_{\bs} = 0$.
    \end{remark}

    \section{Efficiency Maximization and Budget Balance}
    \label{app:efficiency}
    In this section, we discuss the efficiency maximization version of the bilateral trade problem. Let $M$ be a mechanism with allocation region $A$, the social welfare and gain from trade induced by $M$ on agents with valuations $v_s$ and $v_b$ are defined as follows:
    \begin{align*}
        \gft(M,v_s,v_b) &= \ind{(v_s,v_b) \in A} (v_b - v_s) \tag{Gain-From-Trade}\\
        \sw(M,v_s,v_b) &= \ind{(v_s,v_b) \in A} (v_b - v_s) + v_s\tag{Social Welfare}
    \end{align*}

    Note that social welfare is an additive factor away from gain from trade; therefore, from the regret perspective, these two objectives are exactly equivalent. 

    \paragraph{Budget Balance.} Beyond DSIC and IR, another natural property required to efficiency maximizing mechanisms for bilateral trade is budget balance: we want infact to avoid situations in which the mechanism subsidizes the market. There are two main notions of budget balance in the literature \citep[see e.g.,][]{Cesa-BianchiCCF24,ColiniBaldeschiKLT16,BernasconiCCF24}: 
    \begin{align*}
        \prof(M,v_s,v_b) = 0 \quad \forall (v_s,v_b)\in [0,1]^2 \tag{Strong Budget Balance}\\
        \prof(M,v_s,v_b) \ge 0 \quad \forall (v_s,v_b)\in [0,1]^2 \tag{Weak Budget Balance}
    \end{align*}

    We have all the notions to formally recall the impossibility result by \citet{MyersonS83}: no mechanism can, at the same time, achieve full efficiency (e.g., maximize social welfare/gain from trade), individual rationality, incentive compatibility and (weak) budget balance. 

    \paragraph{Fixed-Price Mechanisms.} An important class of mechanisms is composed by \emph{fixed-price mechanisms}, where the broker simply posts two fixed prices $p$ to the seller and $q$ to the buyer, and the trade happens if and only if both agents agree. In the rest of the section, we clarify what we mean in the introduction by claiming that fixed price mechanisms ``characterize'' efficiency maximizing mechanisms. We start by recalling a folklore result on strong budget balance mechanisms.  

    \begin{proposition}[e.g., \citet{ColiniBaldeschiKLT16}]
    \label{prop:strong}
        A mechanism $M$ for bilateral trade is dominant-strategy incentive compatible, individually rational, and strongly budget balanced if and only if it is a fixed price mechanism for some price $p$, i.e., if its allocation region is $[0,p]\times [p,1]$.
    \end{proposition}

    Weak budget balance has a nice geometric interpretation following the characterization of \Cref{thm:mechanisms}. For any non-trivial mechanism $M$ (i.e., with non-empty allocation region), let $\overline s$ be the price paid to the seller and $\underline b$ be the price paid by the buyer, under valuation $(0,1)$. From a geometrical point of view, these two numbers correspond to the intersections of the boundary of $M$ with the left and top side of the $[0,1]^2$. We have the following result.

    \begin{proposition}[Characterization of Weak-Budget Balanced Mechansism]
    \label{prop:weak}
        A mechanism $M$ for bilateral trade that is dominant-strategy incentive compatible and individually rational enforces weak budget balance if and only if $\overline s \le \underline b$. Moreover, the social welfare and gain from trade of any such mechanism is dominated by the fixed price mechanism with allocation region $[0,\overline s] \times [\underline b,1].$
    \end{proposition}
    \begin{proof}
        Clearly, if a mechanism does not respect $\overline s \le \underline b$, then it violates budget balance under valuation $(0,1)$. The other implication follows by noting that, for any valuation $(v_s,v_b)$ that results in a trade, we have:
        \(
            p(v_s,v_b) \le \overline s \le \underline b \le q(v_s,v_b). 
        \)

        Consider now the last part of the statement. Increasing the allocation region while retaining budget balance and monotonicity always (weakly) increases efficiency. Therefore taking the maximal monotone region compatible with $\underline b$ and $\overline s$ is clearly a (weak) improvement.
    \end{proof}

    \paragraph{Budget Balance and Profit.} Since the focus of this paper is on profit maximization, there is no reason to constrain the optimization procedure to DSIC and IR mechanisms that also enforce weak budget balance (clearly, the profit of any strongly budget balance mechanism is always $0$). In fact, there are simple examples in which there is a large gap between the expected profit of the best mechanism and the best budget balance one. 

    \begin{example}[Budget Balance is Suboptimal]
        Consider the uniform distribution over $(0,\nicefrac 14)$ and $(\nicefrac 34,1)$. The optimal mechanism allocates in the $(0,\nicefrac 14)$, $(\nicefrac 34,1)$, $(0,1)$ triangle, for an expected profit of $\nicefrac 14.$ Such mechanism is not budget balanced, as $\overline s = \nicefrac 34$ and $\underline b = \nicefrac 14$. The allocation region of any budget balance mechanism can either contain $(0,\nicefrac 14)$ or $(\nicefrac 34,1)$, for an expected profit of most $\nicefrac 18$.
    \end{example}

    \paragraph{Beyond Dominant-Strategy Incentive Compatibility.} The impossibility results of \citet{MyersonS83} hold in the more general setting where mechanisms only need to enforce incentive compatibility, individual rationality, and budget balance \emph{in expectation.} These mechanisms clearly do not respect the characterization provided in \Cref{thm:mechanisms,prop:strong,prop:weak}, but have also been studied. In particular, the first mechanism to achieve a constant factor approximation of the optimal gain from trade in the Bayesian setting \citep{DengMSW21} is BIC but not DSIC.

\section[The Learning Complexity of M]{The Learning Complexity of $\M$}
\label{app:complexity}

    In this section, we argue that it is not possible to uniformly learn the expected profit of \emph{all} mechanisms in $\M$ for bilateral trade. To provide some further intuition, we also explicitly prove that standard statistical-complexity measures are not helpful (VC dimension, pseudo-dimension, and Rademacher complexity). Formally, we have the following theorem. 

    \begin{theorem}
    \label{thm:uniform-learning}
        There exists a distribution $\D$ such that for any number $n$ of i.i.d. samples $S$, there exists a mechanism $M \in \M$ that respects the following two properties:
        \begin{itemize}
            \item The expected profit of $M$ under $\D$ is constant: $\E{\prof(M)} \ge \nicefrac 14$
            \item The empirical profit of $M$ on the samples $S$ is zero.
        \end{itemize}
    \end{theorem}
    \begin{proof}
        This proof formalizes the argument seen in \Cref{ex:hard_distribution}. Let $\D$ be the uniform distribution on the $(0,\nicefrac 34)$-$(\nicefrac 14,1)$ segment, and consider any realization of the $n$ samples $S$. 
        Denote with $\delta>0$ a small parameter we set later, and with $B_{\delta}(v)$ the generic $\ell^\infty$ open ball around $v$ of radius $\delta.$ Consider the mechanism $M$ whose allocation region is the  $(0,\nicefrac 34),$ $(\nicefrac 14,1),$ $(0,1)$ triangle minus the balls $B_{\delta}(v)$ for $v$ in $S.$ 
        
        By construction, the empirical profit of $M$ on the samples is zero, while the expected profit on $\D$ is at least $\nicefrac 14$, for $\delta$ small enough. In fact, $S$ is finite, so it is possible to take $\delta$ small enough so that at least half of the support of $\D$ falls in the allocation region of $M$. For such $\delta$, a trade happens with probability at least $\nicefrac 12$, with an expected profit of at least $\nicefrac 14$.  
    \end{proof}
\begin{definition}[VC Dimension for Range Spaces]
\label{def:vc_dimension}
Let $\X$ be a set, and $\G$ a collection of subsets of $\X$. The tuple $(\X, \G)$ is called a range space. For each finite set $S\subseteq X$ we define the projection of $\G$ onto $S$ as $\G_S = \left\{G \cap S : G\in \G\right\}$; we say that S is shattered by $\G$ if $\lvert \G_S\rvert = 2^{\lvert S\rvert}$, with $\G_S$ thus the power set of $S$. The VC-dimension of the range space $(\X, \G)$ is the cardinality of the largest $S$ that can be shattered by $\G$. If sets of arbitrary high cardinality can be shattered, we say that the VC-dimension is unbounded.
    
\end{definition}
\begin{theorem}[Unbounded VC Dimension]
    \label{thm:vc_dimension_monotone}
        The VC dimension of the collection of sets induced by monotone allocations regions is unbounded. 
    \end{theorem}
    \begin{proof}
        Let $\mathcal{A}$ be the family of all monotone allocation regions. Consider any set $S$ of $n$ points on on the diagonal $\{(x, y): x = y\} \cap [0, 1]^2$, we prove that $S$ can be shattered by $\A.$ As an auxiliary construction, for any point $(x_i,x_i)$ in $S$, define the rectangle 
         \begin{equation}
         \label{eq:R_x}
            R_i = [0, x_i]\times[x_i,1].
        \end{equation} 
         For any subset $S'$ of $S$, consider the region $A' = \bigcup_{i:\, (x_i, x_i) \in S'} R_i$; to conclude the proof it is enough to argue that the following three properties are verified: (i) $A'$ is monotone, (ii) $A'$ contains all the points in $S'$, and (iii) $A'$ does not contain any point in $S\setminus S'$.

        We start proving monotonicity. Let $\hat v$ be any point in $[0,1]^2$ that dominates some $v'$ in $A'$, by definition $v'$ belongs to some $R_i$ with $(x_i, x_i) \in S'$, implying that $(x_i, x_i) \preceq v'$. By transitivity, $\hat v$ dominates $(x_i, x_i)$, meaning that it belongs to $R_i$ (since each $R_i$ is monotone), thus belonging to $A'$. We are left with proving points (ii) and (iii). The monotone region $A'$ is the union of the $R_i$ for $i: (x_i, x_i) \in S'$, and each $R_i$ only touches the main diagonal in the corresponding point $(x_i, x_i)$. Thus, the intersection of $A'$ with the main diagonal is exactly $S'$.
    \end{proof}

    \begin{figure}[t!]
    \centering
    \begin{subfigure}[t]{0.35\textwidth}
    \includegraphics[width=\textwidth]{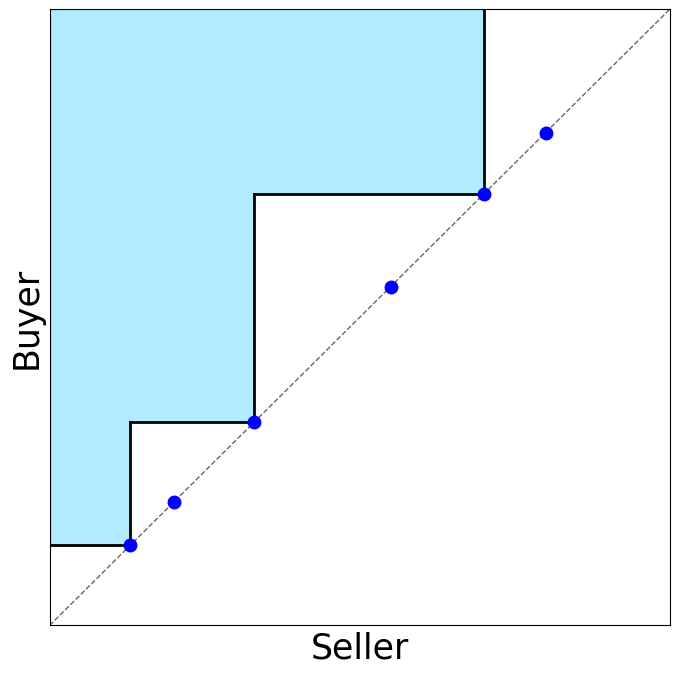}
    \end{subfigure}
    \hspace{1cm}
    \begin{subfigure}[t]{0.35\textwidth}
        \includegraphics[width=\textwidth]{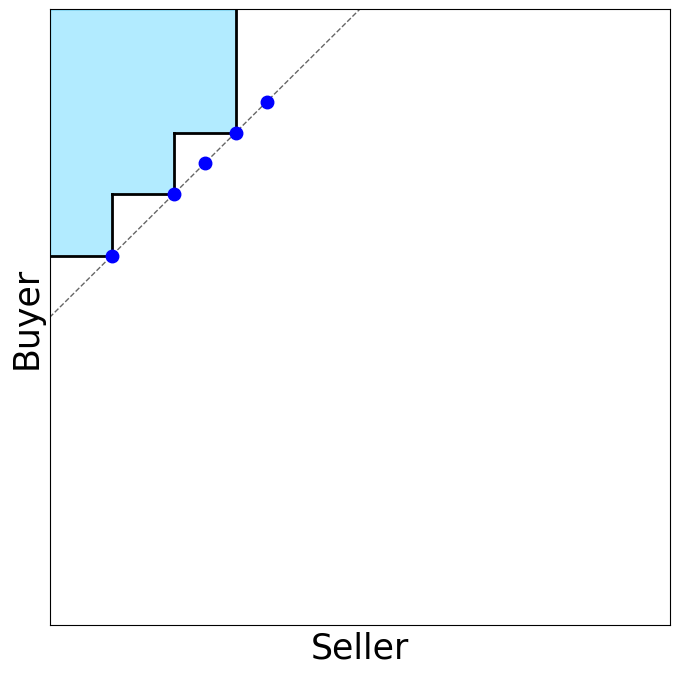}
    \end{subfigure}
        
    \caption{Visualization of the proofs of \Cref{thm:vc_dimension_monotone} (left) and \Cref{thm:rademacher_lower_bound,thm:pseudo_profit} (right).}
    \label{fig:complexity_plots}
    \end{figure}
    \begin{definition}[Pseudo-dimension]
    \label{def:pseudo_dimension}
        For a class of real-valued functions, we say that a set of points, here $S = \{\vi\}_{i=1}^m$, is pseudo-shattered if $\exists$ a vector $r \in \mathbb{R}^m$ such that $\forall\, b \in \{\pm 1\}^n \ \exists f_b \in \F$ such that $\text{sign}(f_b(\vi)-r_i) = b_i \ \forall \, \vi \in C$. The pseudo-dimension is the maximum cardinality of a set which is pseudo-shattered by $\F$.
    \end{definition}
    \begin{theorem}[Unbounded Pseudo-dimension]
    \label{thm:pseudo_profit}
    The pseudo-dimension of the class of profit functions $\F$ is unbounded.
    \begin{proof}
         Fix any set $S \subset C = \left\{(x, y): x = y-\frac{1}{2}\right\}\cap [0, 1]^2$, with $\lvert  S \rvert = m$. If $S$ is a set of valuations, we have that each valuation vector lies on $C$, meaning that the difference between its entries is always $\frac{1}{2}$. With respect to \Cref{def:pseudo_dimension}, here $\F =  \{\prof(M, \cdot), M\in \M\}$, i.e. the profit functions induced by valid mechanisms, as characterized in \Cref{thm:mechanisms}.

        Let us fix an arbitrary vector of signs $b$: this induces a partition of $S$ into $S'$ and $S\setminus S'$. We can build a mechanism $M'$ and a monotone allocation region $A'$ such that the corresponding profit function $\prof(M', \cdot) \in \F$ assigns $\nicefrac{1}{2}$ to points in $S'$ and $0$ otherwise. More precisely, almost identically to \Cref{thm:vc_dimension_monotone}, $A'=\cup_{i:\, \vi \in S'} R_i$, where the rectangles $R_i = [0, x_i]\times[y_i, 1]$ and $\vi = (x_i, y_i)$.  If we take $r \in \mathbb{R}^m$ to be the constant vector with entries equal to $\nicefrac 14$, we immediately have that $\text{sign}(\prof(M', \vi)-r_i) = b_i$. No assumption was made on $b$, so the same holds for any choice of $b$. Therefore, $S$ is pseudo-shattered by $\F$. The procedure is valid $\forall \, m \in \mathbb{N}$: the pseudo-dimension of $\F$ is thus unbounded.
    \end{proof}    
    
    \end{theorem}
    \begin{definition}[Rademacher Complexity]
    \label{def:rademacher}
    Let $\{\vi\}_{i=1}^m$ be a i.i.d. sample of size $m$, such that each $\vi \sim \D$, where $\D$ is a distribution over $\R^n$. Let $\F$ be generic class of real-valued functions on $\R^n$. Let $\{\sigma_i\}_{i=1}^n$ be a set of i.i.d Rademacher variables (i.e., $\sigma_i \in\{-1,+1\}$ with equal probability). The Rademacher complexity of the tuple $(\D, \F)$ is defined as follows:
            \[
                R_m(\D, \F) = \E{\sup_{f\in \F} \left\lvert\frac{1}{m} \sum_{i=1}^m \sigma_i f(\vi)\right\rvert}.        
            \]
\end{definition}
    \begin{theorem}[Lower Bound on Rademacher Complexity]
    \label{thm:rademacher_lower_bound}
        There exists a distribution over $[0, 1]^2$ such that the Rademacher complexity for the class $\F$ of profit functions is at least $\nicefrac 14$.
        \begin{proof}
        We use the same construction of \Cref{thm:pseudo_profit}.
        Consider the uniform distribution $\D$ over the segment
        $C = \left\{(x, y): x = y-\frac{1}{2}\right\}\cap [0, 1]^2$. The segment $C$ is characterized by a constant difference of $\nicefrac 12$ between the buyer and the seller valuations.

        First, we fix a realization of the i.i.d. sample $S=\{\vi\}_{i=1}^m$, where $v_i \sim \D$. Here we abuse notation using the same symbol for both a random variable and its realization. Therefore, as for the proof of \Cref{thm:pseudo_profit}, $S \subset C$; also, $\mathcal{F} = \{\prof(M, \cdot), M\in \M\}$. Note, we assume without loss of generality that the points in $S$ are disjoint: this happens almost surely.

            Fix any realization of the Rademacher random variables and denote with $S'$ the subset of $S$ associated to points whose $\sigma_i$ is $1$. Without loss of generality, we assume that $|S'|\ge \nicefrac{m}{2}$ (if this is not the case, then the same argument holds for $S\setminus S'$). Consider now the mechanism associated with the allocation region $A'=\cup_{i:\, \vi \in S'} R_i$, where the rectangles $R_i$ are defined as in the proof of \Cref{thm:pseudo_profit}. $A'$ contains all the points in $S'$ and no point in $S \setminus S'.$ Moreover, the associated mechanism $M'$ extracts exactly profit $\nicefrac{1}2$ from each point in $S'$ and $0$ from the points not in $S$. We then have: 
            \[
                \sup_{M\in \M} \left\lvert\frac{1}{m} \sum_{i=1}^m \sigma_i \cdot  \prof(M,\vi)\right\rvert \ge \left\lvert\frac{1}{m} \sum_{\vi \in S'} \sigma_i \cdot \prof(M',\vi)\right\rvert = \frac{1}{m} \sum_{\vi \in S'} \frac{1}{2} \ge \frac{1}{4},
            \]
            where in the last inequality we use the assumption $|S'| \ge \nicefrac{m}{2}$.

            If we now take the expectation with respect to the Rademacher random variables, we get a $\nicefrac{1}{4}$ lower bound. This means that the same bound also holds when taking an outer expectation with respect to the distribution of $S$, since it holds for almost every realization of $S$. We have then proved the desired statement: 
            \[
                \Esub{S,\sigma}{\sup_{M\in \M} \left\lvert\frac{1}{m} \sum_{i=1}^m \sigma_i \cdot  \prof(M,\vi)\right\rvert} \ge \frac{1}{4},
            \]
            Note, our bound crucially \emph{does not} depend on the number $m$ of samples.
        \end{proof}
    \end{theorem}

\section{Concentration Inequality}

    In the main body, we use the following version of multiplicative Chernoff–Hoeffding bound that we use (restatement of Theorem 1.1. of \citet{DubhashiP09} for variables in $[0,1]$).

    \begin{theorem}
    \label{thm:chernoff}
        Let $X$ be defined as the sum of $m$ random variables $X_i$, independently distributed in $[0,1]$. Let $\eps \in (0,1)$ and $t > 2 e \E{X}$ be any precision parameters, then the following inequalities hold:
        \begin{itemize}\item[\textnormal{(i)}] $\P{X - \E{X} < - \eps \E{X}} \le \exp\left(-\tfrac 12 \eps^2\E{X}\right)$ 
            \item[\textnormal{(ii)}] $\P{X > t} \le 2^{- t}$ 
        \end{itemize}
    \end{theorem}

    \section[Proof of Proposition 7]{Proof of \Cref{prop:properties}}
    \label{app:sequence}

        For completeness, we report here a direct proof of \Cref{app:sequence} from \citet{AggarwalBDF24}.

        \sequence*

        \begin{proof}
            Consider any realization of the random sequence $(a_t,b_t)$. We start by providing a more direct characterization of its elements, summarizing the ``left'' and ``right'' decisions in an auxiliary variable $\gamma_t$, which has value $-1$ if at time $t$ the sequence goes left, and value $1$ otherwise. We also introduce the middle point $m_t$ of the segment $[a_t,b_t]$, and rewrite $a_t$ and $b_t$ accordingly:
            \begin{equation}
            \label{eq:equivalent}
                m_t = \frac{\delta}{2} + \sum_{s=1}^{t-1} \gamma_s \frac{\delta}{3^s}, \quad a_t = m_t - \frac{\Delta_{t}}{2}, \quad b_t = m_t + \frac{\Delta_{t}}{2}.
            \end{equation}
            \begin{claim}
                The new definition is equivalent to the old one. 
            \end{claim}
            \begin{proof}[Proof of the Claim]
                We prove this result by induction. The base case is simple to verify, as $m_1=\nicefrac 12$, $a_1 = \nicefrac 13$, and $b_1 = \nicefrac 23$ in both cases.
                Consider now the generic time step $t+1$ and assume that the two definitions are equivalent up to time $t$; we prove that \Cref{eq:equivalent} is equivalent to the one in the main body for $t \in L$ and $\gamma_{t}=-1$ (the other case is analogous):
                \begin{align*}
                    a_{t+1} = m_{t+1} - \frac{\Delta_{t+1}}{2} = m_{t} - \frac{\delta}{3^t} - \frac{\Delta_{t+1}}{2} = a_t + \frac{\Delta_t}{2} -\frac{\delta}{3^t} - \frac{\Delta_{t+1}}{2}  = a_t - 2 \Delta_{t+1}\\
                    b_{t+1} = m_{t+1} + \frac{\Delta_{t+1}}{2} = m_t - \frac{\delta}{3^t} + \frac{\Delta_{t+1}}{2} = a_t + \frac{\Delta_t}{2} -\frac{\delta}{3^t} + \frac{\Delta_{t+1}}{2}  = a_t - \Delta_{t+1}.
                \end{align*}
                Note, all equalities follow by \Cref{eq:equivalent}, except the third equalities in both formulae that are due to the inductive hypotheses.
                \end{proof}
                The proof of point (i) is fairly direct; we prove it for $a_t$, the other case is analogous. Any $\gamma_t$ lives in $\{-1,1\}$, therefore we have:
                \begin{align*}
                    a_t &= m_t - \frac{\Delta_t}{2} = \frac{\delta}{2} + \sum_{s=1}^{t-1} \gamma_s \frac{\delta}{3^{s}} - \frac{\Delta_t}{2} \ge \frac{\delta}{2} - \sum_{s=1}^{t} \frac{\delta}{3^{s}} \ge \frac{\delta}{2} - \sum_{s=1}^{\infty} \frac{\delta}{3^{s}} = 0\\
                    a_t &= m_t - \frac{\Delta_t}{2} = \frac{\delta}{2} + \sum_{s=1}^{t-1} \gamma_s \frac{\delta}{3^{s}} - \frac{\Delta_t}{2}\le \frac{\delta}{2} + \sum_{s=1}^{\infty} \frac{\delta}{3^{s}} = \delta.
                \end{align*}
                As a preliminary step to tackle point (ii), we argue that the $a_t$ and $b_t$ are monotone \emph{when} restricted to $L$ and $R$, respectively.
                \begin{claim}
                    The sequence $a_t$ is decreasing for $t \in L$, while $b_t$ is increasing for $t \in R$.
                \end{claim}
                \begin{proof}[Proof of the Claim]
                    We prove the Claim for $a_t$, the other case follows by negating the expression in \Cref{eq:monotone_seq}. Consider any consecutive $t_1 \le t_2$ in $L$ (meaning that all integer in $(t_1,t_2)$ are in $R$, and so the corresponding $\gamma_t$ are positive), we have the following:
                    \begin{equation}
                    \label{eq:monotone_seq}
                        a_{t_1} - a_{t_2} = \sum_{s=t_1+1}^{t_2-1} \frac{\delta}{3^s} -\frac{\delta}{3^{t_1}}- \frac{\Delta_{t_1}}{2} + \frac{\Delta_{t_2}}2 \ge \sum_{s=t_1+1}^{t_2-1} \frac{\delta}{3^s} + \frac{\Delta_{t_2}}{2} > 0
                     \end{equation}
                \end{proof}
                Given the result of the last Claim, in order to prove point (ii) of the Proposition, we only need to argue that $a_{t_L} \ge b_{t_R}$, where $t_L$, respectively $t_R$, are the last ``left'', respectively ``right'' rime step. Once we have this result, it is enough to set $\tau$ in $(a_{t_L}, b_{t_R})$.

                We prove the desired inequality for the case in which $t_L > t_R$ (and thus $\gamma_t=-1$ for every integer $t$ in $(t_R,t_L]$).

                \begin{align*}
                    b_{t_R} - a_{t_L} &= -\sum_{s=t_{R}}^{t_L-1} \gamma_s \frac{\delta}{3^{s}} + \frac{\Delta_{t_L}}{2} + \frac{\Delta_{t_R}}{2} \tag{By \Cref{eq:equivalent}}\\
                    &=-\frac{\delta}{3^{t_R}}+\sum_{s=t_{R}+1}^{t_L-1} \frac{\delta}{3^{s}} + \frac{\delta}{2\cdot 3^{t_L}} + \frac{\delta}{2\cdot 3^{t_R}} \tag{As $\gamma_t=-1$ after $t_R$ and $\gamma_{t_R}=1$}\\
                    &< \sum_{s=t_{R}+1}^{t_L} \frac{\delta}{3^{s}} - \frac{\delta}{2\cdot 3^{t_R}} < \sum_{s=t_{R}+1}^{\infty} \frac{\delta}{3^{s}} - \frac{\delta}{2\cdot 3^{t_R}} \le 0,
                \end{align*}
                where the last inequality follows by the closed-form expression of the tail sum of a geometric series.

                When $t_R > t_L$ we have (recall that in this case $\gamma_t = 1$ for every integer $t\in (t_L, t_R])$: 
                \begin{align*}
                    b_{t_R} - a_{t_L} &= \sum_{s=t_{L}}^{t_R-1} \gamma_s \frac{\delta}{3^{s}} + \frac{\Delta_{t_L}}{2} + \frac{\Delta_{t_R}}{2} \tag{By \Cref{eq:equivalent}}\\
                    &=-\frac{\delta}{3^{t_L}}+\sum_{s=t_{L}+1}^{t_R-1} \frac{\delta}{3^{s}} + \frac{\delta}{2\cdot 3^{t_L}} + \frac{\delta}{2\cdot 3^{t_R}} \tag{As $\gamma_t=1$ after $t_L$ and $\gamma_{t_L}=-1$}\\
                    &< \sum_{s=t_{L}+1}^{t_R} \frac{\delta}{3^{s}} - \frac{\delta}{2\cdot 3^{t_L}} < \sum_{s=t_{L}+1}^{\infty} \frac{\delta}{3^{s}} - \frac{\delta}{2\cdot 3^{t_L}} \le 0,
                \end{align*}
        \end{proof}

%% file: main.bbl
\begin{thebibliography}{58}
\providecommand{\natexlab}[1]{#1}
\providecommand{\url}[1]{\texttt{#1}}
\expandafter\ifx\csname urlstyle\endcsname\relax
  \providecommand{\doi}[1]{doi: #1}\else
  \providecommand{\doi}{doi: \begingroup \urlstyle{rm}\Url}\fi

\bibitem[Aggarwal et~al.(2024)Aggarwal, Badanidiyuru, D\"utting, and Fusco]{AggarwalBDF24}
Gagan Aggarwal, Ashwinkumar Badanidiyuru, Paul D\"utting, and Federico Fusco.
\newblock Selling joint ads: {A} regret minimization perspective.
\newblock In \emph{{EC}}, pages 164--194. {ACM}, 2024.

\bibitem[Alon et~al.(2015)Alon, Cesa{-}Bianchi, Dekel, and Koren]{AlonCDK15}
Noga Alon, Nicol{\`{o}} Cesa{-}Bianchi, Ofer Dekel, and Tomer Koren.
\newblock Online learning with feedback graphs: Beyond bandits.
\newblock In \emph{{COLT}}, volume~40, pages 23--35, 2015.

\bibitem[Audibert and Bousquet(2007)]{AudibertB07}
Jean{-}Yves Audibert and Olivier Bousquet.
\newblock Combining pac-bayesian and generic chaining bounds.
\newblock \emph{J. Mach. Learn. Res.}, 8:\penalty0 863--889, 2007.

\bibitem[Azar et~al.(2024)Azar, Fiat, and Fusco]{AzarFF24}
Yossi Azar, Amos Fiat, and Federico Fusco.
\newblock An {\(\alpha\)}-regret analysis of adversarial bilateral trade.
\newblock \emph{Artif. Intell.}, 337:\penalty0 104231, 2024.

\bibitem[Babaioff et~al.(2020)Babaioff, Goldner, and Gonczarowski]{BabaioffGG20}
Moshe Babaioff, Kira Goldner, and Yannai~A. Gonczarowski.
\newblock Bulow-klemperer-style results for welfare maximization in two-sided markets.
\newblock In \emph{{SODA}}, pages 2452--2471. {ACM-SIAM}, 2020.

\bibitem[Bart{\'{o}}k et~al.(2014)Bart{\'{o}}k, Foster, P{\'{a}}l, Rakhlin, and Szepesv{\'{a}}ri]{BartokFPRS14}
G{\'{a}}bor Bart{\'{o}}k, Dean~P. Foster, D{\'{a}}vid P{\'{a}}l, Alexander Rakhlin, and Csaba Szepesv{\'{a}}ri.
\newblock Partial monitoring - classification, regret bounds, and algorithms.
\newblock \emph{Math. Oper. Res.}, 39\penalty0 (4):\penalty0 967--997, 2014.

\bibitem[Bernasconi et~al.(2024)Bernasconi, Castiglioni, Celli, and Fusco]{BernasconiCCF24}
Martino Bernasconi, Matteo Castiglioni, Andrea Celli, and Federico Fusco.
\newblock No-regret learning in bilateral trade via global budget balance.
\newblock In \emph{{STOC}}, pages 247--258. {ACM}, 2024.

\bibitem[Block et~al.(2022)Block, Dagan, Golowich, and Rakhlin]{BlockDGR22}
Adam Block, Yuval Dagan, Noah Golowich, and Alexander Rakhlin.
\newblock Smoothed online learning is as easy as statistical learning.
\newblock In \emph{{COLT}}, volume 178, pages 1716--1786, 2022.

\bibitem[Blum et~al.(2003)Blum, Kumar, Rudra, and Wu]{BlumKRW03}
Avrim Blum, Vijay Kumar, Atri Rudra, and Felix Wu.
\newblock Online learning in online auctions.
\newblock In \emph{{SODA 2003}}, pages 202--204. ACM-SIAM, 2003.

\bibitem[Blumrosen and Dobzinski(2021)]{BlumrosenD21}
Liad Blumrosen and Shahar Dobzinski.
\newblock ({A}lmost) efficient mechanisms for bilateral trading.
\newblock \emph{Games Econ. Behav.}, 130:\penalty0 369--383, 2021.

\bibitem[Bourgain et~al.(2015)Bourgain, Dirksen, and Nelson]{BourgainDN15}
Jean Bourgain, Sjoerd Dirksen, and Jelani Nelson.
\newblock Toward a unified theory of sparse dimensionality reduction in euclidean space.
\newblock In \emph{STOC}, pages 499--508. {ACM}, 2015.

\bibitem[Braverman et~al.(2017)Braverman, Chestnut, Ivkin, Nelson, Wang, and Woodruff]{BravermanCINWW17}
Vladimir Braverman, Stephen~R. Chestnut, Nikita Ivkin, Jelani Nelson, Zhengyu Wang, and David~P. Woodruff.
\newblock Bptree: An $\ell_2$ heavy hitters algorithm using constant memory.
\newblock In \emph{PODS}, pages 361--376. {ACM}, 2017.

\bibitem[Brustle et~al.(2017)Brustle, Cai, Wu, and Zhao]{BrustleCWZ17}
Johannes Brustle, Yang Cai, Fa~Wu, and Mingfei Zhao.
\newblock Approximating gains from trade in two-sided markets via simple mechanisms.
\newblock In \emph{{EC}}, pages 589--590. {ACM}, 2017.

\bibitem[Cai and Daskalakis(2017)]{CaiD17}
Yang Cai and Constantinos Daskalakis.
\newblock Learning multi-item auctions with (or without) samples.
\newblock In \emph{{FOCS}}, pages 516--527. {IEEE}, 2017.

\bibitem[Cai and Wu(2023)]{CaiW23}
Yang Cai and Jinzhao Wu.
\newblock On the optimal fixed-price mechanism in bilateral trade.
\newblock In \emph{{STOC}}, pages 737--750. {ACM}, 2023.

\bibitem[Cai et~al.(2021)Cai, Goldner, Ma, and Zhao]{CaiGMZ21}
Yang Cai, Kira Goldner, Steven Ma, and Mingfei Zhao.
\newblock On multi-dimensional gains from trade maximization.
\newblock In \emph{SODA 2021}, pages 1079--1098. {ACM-SIAM}, 2021.

\bibitem[Cesa{-}Bianchi and Lugosi(2006)]{nicolo06}
Nicol{\`{o}} Cesa{-}Bianchi and G{\'{a}}bor Lugosi.
\newblock \emph{Prediction, learning, and games}.
\newblock Cambridge University Press, Cambridge, UK, 2006.

\bibitem[Cesa{-}Bianchi et~al.(2015)Cesa{-}Bianchi, Gentile, and Mansour]{Cesa-BianchiGM15}
Nicol{\`{o}} Cesa{-}Bianchi, Claudio Gentile, and Yishay Mansour.
\newblock Regret minimization for reserve prices in second-price auctions.
\newblock \emph{{IEEE} Trans. Inf. Theory}, 61\penalty0 (1):\penalty0 549--564, 2015.

\bibitem[Cesa{-}Bianchi et~al.(2024{\natexlab{a}})Cesa{-}Bianchi, Cesari, Colomboni, Fusco, and Leonardi]{Cesa-BianchiCCF24}
Nicol{\`{o}} Cesa{-}Bianchi, Tommaso Cesari, Roberto Colomboni, Federico Fusco, and Stefano Leonardi.
\newblock Bilateral trade: {A} regret minimization perspective.
\newblock \emph{Math. Oper. Res.}, 49\penalty0 (1):\penalty0 171--203, 2024{\natexlab{a}}.

\bibitem[Cesa{-}Bianchi et~al.(2024{\natexlab{b}})Cesa{-}Bianchi, Cesari, Colomboni, Fusco, and Leonardi]{Cesa-BianchiCCF24jmlr}
Nicol{\`{o}} Cesa{-}Bianchi, Tommaso Cesari, Roberto Colomboni, Federico Fusco, and Stefano Leonardi.
\newblock Regret analysis of bilateral trade with a smoothed adversary.
\newblock \emph{J. Mach. Learn. Res.}, 25:\penalty0 234:1--234:36, 2024{\natexlab{b}}.

\bibitem[Cesa{-}Bianchi et~al.(2024{\natexlab{c}})Cesa{-}Bianchi, Cesari, Colomboni, Fusco, and Leonardi]{CesaBianchiCCFS24}
Nicol{\`{o}} Cesa{-}Bianchi, Tommaso Cesari, Roberto Colomboni, Federico Fusco, and Stefano Leonardi.
\newblock The role of transparency in repeated first-price auctions with unknown valuations.
\newblock In \emph{{STOC}}, pages 225--236. {ACM}, 2024{\natexlab{c}}.

\bibitem[Cohen{-}Addad et~al.(2022)Cohen{-}Addad, Larsen, Saulpic, and Schwiegelshohn]{Cohen-AddadLSS22}
Vincent Cohen{-}Addad, Kasper~Green Larsen, David Saulpic, and Chris Schwiegelshohn.
\newblock Towards optimal lower bounds for k-median and k-means coresets.
\newblock In \emph{STOC}, pages 1038--1051. {ACM}, 2022.

\bibitem[Cohen{-}Addad et~al.(2025{\natexlab{a}})Cohen{-}Addad, Draganov, Russo, Saulpic, and Schwiegelshohn]{Cohen-AddadD0SS25}
Vincent Cohen{-}Addad, Andrew Draganov, Matteo Russo, David Saulpic, and Chris Schwiegelshohn.
\newblock A tight {VC}-dimension analysis of clustering coresets with applications.
\newblock In \emph{SODA}, pages 4783--4808. {SIAM}, 2025{\natexlab{a}}.

\bibitem[Cohen{-}Addad et~al.(2025{\natexlab{b}})Cohen{-}Addad, Lattanzi, and Schwiegelshohn]{Cohen-AddadLS25}
Vincent Cohen{-}Addad, Silvio Lattanzi, and Chris Schwiegelshohn.
\newblock Almost optimal {PAC} learning for k-means.
\newblock In \emph{STOC}, pages 2019--2030. {ACM}, 2025{\natexlab{b}}.

\bibitem[Cole and Roughgarden(2014)]{ColeR14}
Richard Cole and Tim Roughgarden.
\newblock The sample complexity of revenue maximization.
\newblock In \emph{{STOC 2014}}, pages 243--252. ACM, 2014.

\bibitem[Colini{-}Baldeschi et~al.(2016)Colini{-}Baldeschi, de~Keijzer, Leonardi, and Turchetta]{ColiniBaldeschiKLT16}
Riccardo Colini{-}Baldeschi, Bart de~Keijzer, Stefano Leonardi, and Stefano Turchetta.
\newblock Approximately efficient double auctions with strong budget balance.
\newblock In \emph{{SODA}}, pages 1424--1443. {ACM-SIAM}, 2016.

\bibitem[Deng et~al.(2022)Deng, Mao, Sivan, and Wang]{DengMSW21}
Yuan Deng, Jieming Mao, Balasubramanian Sivan, and Kangning Wang.
\newblock Approximately efficient bilateral trade.
\newblock In \emph{{STOC}}, pages 718--721. {ACM}, 2022.

\bibitem[Devanur et~al.(2016)Devanur, Huang, and Psomas]{Devanur0P16}
Nikhil~R. Devanur, Zhiyi Huang, and Christos{-}Alexandros Psomas.
\newblock The sample complexity of auctions with side information.
\newblock In \emph{{STOC}}, pages 426--439. {ACM}, 2016.

\bibitem[Dubhashi and Panconesi(2009)]{DubhashiP09}
Devdatt~P. Dubhashi and Alessandro Panconesi.
\newblock \emph{Concentration of Measure for the Analysis of Randomized Algorithms}.
\newblock Cambridge University Press, Cambridge, UK, 2009.

\bibitem[Durvasula et~al.(2023)Durvasula, Haghtalab, and Zampetakis]{DurvasulaHZ23}
Naveen Durvasula, Nika Haghtalab, and Manolis Zampetakis.
\newblock Smoothed analysis of online non-parametric auctions.
\newblock In \emph{{EC 2023}}, pages 540--560. ACM, 2023.

\bibitem[D{\"{u}}tting et~al.(2021)D{\"{u}}tting, Fusco, Lazos, Leonardi, and Reiffenh{\"{a}}user]{DuttingFLLR21}
Paul D{\"{u}}tting, Federico Fusco, Philip Lazos, Stefano Leonardi, and Rebecca Reiffenh{\"{a}}user.
\newblock Efficient two-sided markets with limited information.
\newblock In \emph{STOC}, pages 1452--1465. {ACM}, 2021.

\bibitem[D\"utting et~al.(2023)D\"utting, Guruganesh, Schneider, and Wang]{DuettinGSW23}
Paul D\"utting, Guru Guruganesh, Jon Schneider, and Joshua~Ruizhi Wang.
\newblock Optimal no-regret learning for one-sided lipschitz functions.
\newblock In \emph{{ICML 2023}}, pages 8836--8850, 2023.

\bibitem[Fei(2022)]{Fei22}
Yumou Fei.
\newblock Improved approximation to first-best gains-from-trade.
\newblock In \emph{{WINE}}, pages 204--218. Springer, 2022.

\bibitem[Feng et~al.(2018)Feng, Podimata, and Syrgkanis]{FengPS18}
Zhe Feng, Chara Podimata, and Vasilis Syrgkanis.
\newblock Learning to bid without knowing your value.
\newblock In \emph{{EC 2018}}, pages 505--522. ACM, 2018.

\bibitem[Gatmiry et~al.(2024)Gatmiry, Kesselheim, Singla, and Wang]{GatmiryKSW24}
Khashayar Gatmiry, Thomas Kesselheim, Sahil Singla, and Yifan Wang.
\newblock Bandit algorithms for prophet inequality and pandora's box.
\newblock In \emph{SODA 2024}, pages 462--500, 2024.

\bibitem[Gonczarowski and Weinberg(2021)]{GonczarowskiW21}
Yannai~A. Gonczarowski and S.~Matthew Weinberg.
\newblock The sample complexity of up-to-{\(\epsilon\)} multi-dimensional revenue maximization.
\newblock \emph{J. {ACM}}, 68\penalty0 (3):\penalty0 15:1--15:28, 2021.

\bibitem[Guo et~al.(2021)Guo, Huang, Tang, and Zhang]{GuoHGZ21}
Chenghao Guo, Zhiyi Huang, Zhihao~Gavin Tang, and Xinzhi Zhang.
\newblock Generalizing complex hypotheses on product distributions: {A}uctions, prophet inequalities, and pandora’s problem.
\newblock In \emph{COLT 2021}, pages 2248--2288, 2021.

\bibitem[Haghtalab et~al.(2024)Haghtalab, Roughgarden, and Shetty]{HaghtalabRS24}
Nika Haghtalab, Tim Roughgarden, and Abhishek Shetty.
\newblock Smoothed analysis with adaptive adversaries.
\newblock \emph{J. {ACM}}, 71\penalty0 (3):\penalty0 19, 2024.

\bibitem[Hajiaghayi et~al.(2025)Hajiaghayi, Hajiaghayi, Peng, and Shin]{HajiaghayiHPS25}
Ilya Hajiaghayi, MohammadTaghi Hajiaghayi, Gary Peng, and Suho Shin.
\newblock Gains-from-trade in bilateral trade with a broker.
\newblock In \emph{{SODA}}, pages 4827--4860. {SIAM}, 2025.

\bibitem[Jambulapati et~al.(2023)Jambulapati, Liu, and Sidford]{JambulapatiLS23}
Arun Jambulapati, Yang~P. Liu, and Aaron Sidford.
\newblock Chaining, group leverage score overestimates, and fast spectral hypergraph sparsification.
\newblock In \emph{STOC}, pages 196--206. {ACM}, 2023.

\bibitem[Kang et~al.(2022)Kang, Pernice, and Vondr{\'{a}}k]{kang22fixed}
Zi~Yang Kang, Francisco Pernice, and Jan Vondr{\'{a}}k.
\newblock Fixed-price approximations in bilateral trade.
\newblock In \emph{{SODA}}, pages 2964--2985. {ACM-SIAM}, 2022.

\bibitem[Kleinberg et~al.(2019)Kleinberg, Slivkins, and Upfal]{KleinbergSU19}
Robert Kleinberg, Aleksandrs Slivkins, and Eli Upfal.
\newblock Bandits and experts in metric spaces.
\newblock \emph{J. {ACM}}, 66\penalty0 (4):\penalty0 30:1--30:77, 2019.

\bibitem[Kleinberg and Leighton(2003)]{KleinbergL03}
Robert~D. Kleinberg and Frank~Thomson Leighton.
\newblock The value of knowing a demand curve: Bounds on regret for online posted-price auctions.
\newblock In \emph{{FOCS 2003}}, pages 594--605. IEEE, 2003.

\bibitem[Lee(2023)]{Lee23}
James~R. Lee.
\newblock Spectral hypergraph sparsification via chaining.
\newblock In \emph{STOC}, pages 207--218. {ACM}, 2023.

\bibitem[Li et~al.(2000)Li, Long, and Srinivasan]{LiLS00}
Yi~Li, Philip~M. Long, and Aravind Srinivasan.
\newblock Improved bounds on the sample complexity of learning.
\newblock In \emph{SODA}, pages 309--318. {ACM-SIAM}, 2000.

\bibitem[Liu et~al.(2023)Liu, Ren, and Wang]{LiuR023}
Zhengyang Liu, Zeyu Ren, and Zihe Wang.
\newblock Improved approximation ratios of fixed-price mechanisms in bilateral trades.
\newblock In \emph{{STOC}}, pages 751--760. {ACM}, 2023.

\bibitem[McAfee(2008)]{Mcafee08}
R~Preston McAfee.
\newblock The gains from trade under fixed price mechanisms.
\newblock \emph{Appl. Econ. Res. Bull.}, 1\penalty0 (1):\penalty0 1--10, 2008.

\bibitem[Morgenstern and Roughgarden(2015)]{MorgensternR15}
Jamie Morgenstern and Tim Roughgarden.
\newblock On the pseudo-dimension of nearly optimal auctions.
\newblock In \emph{{NIPS 2015}}, pages 136--144, 2015.

\bibitem[Myerson(1981)]{myerson81}
Roger Myerson.
\newblock Optimal auction design.
\newblock \emph{Math. Oper. Res.}, 6\penalty0 (1):\penalty0 58--73, 1981.

\bibitem[Myerson and Satterthwaite(1983)]{MyersonS83}
Roger~B Myerson and Mark~A Satterthwaite.
\newblock Efficient mechanisms for bilateral trading.
\newblock \emph{J. Econ. Theory}, 29\penalty0 (2):\penalty0 265--281, 1983.

\bibitem[Narayanan and Nelson(2019)]{NarayananN19}
Shyam Narayanan and Jelani Nelson.
\newblock Optimal terminal dimensionality reduction in euclidean space.
\newblock In \emph{STOC}, pages 1064--1069. {ACM}, 2019.

\bibitem[Nelson(2016)]{Nelson16}
Jelani Nelson.
\newblock Chaining introduction with some computer science applications.
\newblock \emph{Bull. {EATCS}}, 120, 2016.

\bibitem[Roughgarden(2020)]{Roughgarden20}
Tim Roughgarden.
\newblock \emph{Beyond the Worst-Case Analysis of Algorithms}.
\newblock Cambridge University Press, Cambridge, UK, 2020.

\bibitem[Roughgarden and Wang(2019)]{RoughgardenW19}
Tim Roughgarden and Joshua~R. Wang.
\newblock Minimizing regret with multiple reserves.
\newblock \emph{{ACM} Trans. Econ. Comput.}, 7\penalty0 (3):\penalty0 17:1--17:18, 2019.

\bibitem[Rudra and Wootters(2014)]{RudraW14}
Atri Rudra and Mary Wootters.
\newblock Every list-decodable code for high noise has abundant near-optimal rate puncturings.
\newblock In \emph{STOC}, pages 764--773. {ACM}, 2014.

\bibitem[Talagrand(2014)]{Talagrand14}
Michel Talagrand.
\newblock \emph{Upper and lower bounds for stochastic processes}, volume~60.
\newblock Springer, Berlin, Heidelberg, Germany, 2014.

\bibitem[Vickrey(1961)]{Vickrey61}
William Vickrey.
\newblock Counterspeculation, auctions, and competitive sealed tenders.
\newblock \emph{J. Finance}, 16\penalty0 (1):\penalty0 8--37, 1961.

\bibitem[Zhu et~al.(2023)Zhu, Bates, Yang, Wang, Jiao, and Jordan]{ZhuBYWYJ22}
Banghua Zhu, Stephen Bates, Zhuoran Yang, Yixin Wang, Jiantao Jiao, and Michael~I. Jordan.
\newblock The sample complexity of online contract design.
\newblock In \emph{{EC 2023}}, page 1188. ACM, 2023.

\end{thebibliography}
